\documentclass[conference]{IEEEtran}

\usepackage{ascmac}
\usepackage{amsmath}
\usepackage{amssymb}
\usepackage{stmaryrd}
\usepackage{graphicx}
\usepackage{amsthm}
\usepackage{thmtools}
\usepackage{color}
\usepackage{enumitem}
\usepackage[hyphens]{url}
\usepackage{dsfont}
\usepackage{multirow}
\usepackage[linesnumbered, lined, boxed]{algorithm2e}
\usepackage{enumitem}
\usepackage{comment}
\newcommand{\nats}{\mathbb{N}}

\newcommand{\nnints}{\mathbb{Z}_{\ge0}}
\newcommand{\reals}{\mathbb{R}}
\newcommand{\nnreals}{\mathbb{R}_{\ge0}}
\newcommand{\mathone}{\mathds{1}} 

\renewcommand{\epsilon}{\varepsilon}

\newcommand{\calA}{\mathcal{A}}

\newcommand{\calD}{\mathcal{D}}

\newcommand{\calG}{\mathcal{G}}
\newcommand{\calH}{\mathcal{H}}
\newcommand{\calM}{\mathcal{M}}
\newcommand{\calR}{\mathcal{R}}
\newcommand{\calS}{\mathcal{S}}
\newcommand{\calT}{\mathcal{T}}
\newcommand{\calU}{\mathcal{U}}
\newcommand{\calX}{\mathcal{X}}
\newcommand{\calY}{\mathcal{Y}}

\newcommand{\E}{\operatorname{\mathbb{E}}}
\newcommand{\V}{\operatorname{\mathbb{V}}}

\newcommand{\bmf}{\mathbf{f}}

\newcommand{\bmm}{\mathbf{m}}

\newcommand{\bmPhi}{\mathbf{\Phi}}
\newcommand{\bmPsi}{\mathbf{\Psi}}

\newcommand{\hf}{\hat{f}}

\newcommand{\hPhi}{\hat{\Phi}}
\newcommand{\hPsi}{\hat{\Psi}}
\newcommand{\tc}{\tilde{c}}

\newcommand{\tf}{\tilde{f}}

\newcommand{\tn}{\tilde{n}}

\newcommand{\bmhf}{\mathbf{\hf}}

\newcommand{\bmhPhi}{\mathbf{\hPhi}}
\newcommand{\bmhPsi}{\mathbf{\hPsi}}

\newcommand{\la}{\langle}
\newcommand{\ra}{\rangle}

\newcommand{\pkd}{{\textsf{pk}_\textsf{d}}}
\newcommand{\pks}{{\textsf{pk}_\textsf{s}}}

\newcommand{\LNF}{\textsf{LNF}}
\newcommand{\CH}{\textsf{CH}}
\newcommand{\FME}{\textsf{FME}}
\newcommand{\plain}{\textsf{plain}}
\newcommand{\KV}{\textsf{KV}}
\newcommand{\PS}{\textsf{PS}}

\newcommand{\GMGA}{G_f^{\textsf{max}}}
\newcommand{\GMTGAPhi}{G_\Phi^{\textsf{max}}}
\newcommand{\GMTGAPsi}{G_\Psi^{\textsf{max}}}

\def\ProposalL{\textsf{Proposal (large $l$)}}
\def\ProposalS{\textsf{Proposal (small $l$)}}
\def\ProposalA{\textsf{Proposal*}}

\def\BC{\textsf{BC20}}
\def\CM{\textsf{CM22}}

\def\LWYL{\textsf{LWY22-Large}}

\def\Geo{\textsf{Geo}}

\newcommand{\colorB}[1]{\textcolor{black}{#1}}

\newtheorem{definition}{Definition}
\newtheorem{theorem}{Theorem}
\newtheorem{lemma}{Lemma}

\SetKwInput{KwData}{Input}
\SetKwInput{KwResult}{Output}

\usepackage{hyperref}

\newif\ifconferenceon\conferenceonfalse
\ifconferenceon
\newcommand{\conference}[1]{#1}
\newcommand{\arxiv}[1]{}
\else
\newcommand{\conference}[1]{}
\newcommand{\arxiv}[1]{#1}
\usepackage{balance}
\fi

\ifCLASSINFOpdf
\else
\fi

\begin{document}
\title{Augmented Shuffle Differential Privacy Protocols for Large-Domain Categorical and Key-Value Data}

\author{\IEEEauthorblockN{Takao Murakami}
	\IEEEauthorblockA{ISM/AIST/RIKEN AIP\\
		tmura@ism.ac.jp}
	\and
	\IEEEauthorblockN{Yuichi Sei}
	\IEEEauthorblockA{UEC\\
		seiuny@uec.ac.jp}
	\and
	\IEEEauthorblockN{Reo Eriguchi}
	\IEEEauthorblockA{AIST\\
		eriguchi-reo@aist.go.jp}}

\IEEEoverridecommandlockouts
\makeatletter\def\@IEEEpubidpullup{6.5\baselineskip}\makeatother
\conference{\IEEEpubid{\parbox{\columnwidth}{
		Network and Distributed System Security (NDSS) Symposium 2026\\
		23 - 27 February 2026 , San Diego, CA, USA\\
		ISBN 979-8-9919276-8-0\\  
		https://dx.doi.org/10.14722/ndss.2026.231124\\
		www.ndss-symposium.org
}
\hspace{\columnsep}\makebox[\columnwidth]{}}}

\maketitle

\begin{abstract}
Shuffle DP (Differential Privacy) protocols  
provide high accuracy and privacy by introducing a shuffler who randomly shuffles data 
in a distributed system. 
However, most shuffle DP protocols are vulnerable to two attacks: collusion attacks by the data collector and users and data poisoning attacks. 
A recent study addresses this issue by introducing an augmented shuffle DP protocol, where users do not add noise and the shuffler performs random sampling and dummy data addition. 
However, it focuses on frequency estimation over categorical data with a small domain and 
cannot be applied to 
a large domain due to prohibitively high communication and computational costs. 

In this paper, we 
fill this gap by introducing a novel augmented shuffle DP protocol called the FME (Filtering-with-Multiple-Encryption) protocol. 
Our FME protocol uses a hash function to filter out unpopular items and then accurately calculates frequencies for popular items. 
To perform this 
within one round of interaction between users and the shuffler, our protocol carefully 
communicates within a system using multiple encryption. 
We also apply our FME protocol to more advanced KV (Key-Value) statistics estimation with an additional technique to reduce bias. 
For both categorical and KV data, we prove that our protocol provides computational DP, high robustness to the above two attacks, accuracy, and efficiency. 
We show the effectiveness of our proposals through comparisons with twelve existing protocols. 
\end{abstract}

\IEEEpeerreviewmaketitle

\section{Introduction}
\label{sec:intro}
DP (Differential Privacy)~\cite{DP} has been widely adopted by industry~\cite{Erlingsson_CCS14,Ding_NIPS17,Thakurta_USPatent17} and government agencies~\cite{Drechsler_JASA23} 
to perform data analysis 
while protecting individual privacy. 
DP has been originally studied in the central model, where a single data collector holds all users' personal data and adds noise to the statistical results. 
Although central DP enables accurate data analysis, all personal data can be leaked by 
data breach incidents~\cite{data_breach_2024}. 
LDP (Local DP)~\cite{Kasiviswanathan_FOCS08,Wang_NDSS20,Ye_NDSS25} addresses this issue by 
adding noise to each user's personal data before sending it to the data collector. 
However, LDP suffers from low accuracy, as it needs to add a lot of noise to each user's personal data. 

The shuffle model of DP~\cite{Bittau_SOSP17,Erlingsson_SODA19,Cheu_EUROCRYPT19,Balle_CRYPTO19,Girgis_CCS21,Feldman_FOCS21,Feldman_SODA23} has been recently studied to achieve high accuracy and privacy. 
It assumes a distributed system involving an intermediate server called the \textit{shuffler} 
and typically works as follows. 
Each user adds noise to her input data and sends an encrypted version of the noisy data to the shuffler. 
The shuffler randomly shuffles the noisy data and sends it to the data collector. 
Finally, the data collector decrypts the shuffled 
data. 
Under the assumption that the shuffler does not collude with the data collector, the random shuffling amplifies privacy. 
Specifically, DP strongly protects privacy when the privacy budget $\epsilon$ is small (e.g., $\epsilon \leq 1$), and the shuffling significantly reduces $\epsilon$. 
Thus, the shuffle model achieves the same value of $\epsilon$ as the local model with less noise, i.e., higher accuracy. 

However, most shuffle DP protocols have the following two vulnerabilities. 
First, they are vulnerable to \textit{collusion attacks by the data collector and users} (or \textit{collusion with users})~\cite{Wang_PVLDB20,Murakami_SP25}. 
Specifically, the data collector can obtain noisy data of some users by colluding with them or compromising their accounts. 
In this case, the data collector can reduce the number of shuffled values and thereby reduce the effect of shuffling. 
As a result, the actual value of $\epsilon$ can be significantly increased (e.g., from around $1$ to $8$; see Section~\ref{sub:collusion}). 
Second, most shuffle protocols are vulnerable to \textit{data poisoning attacks}~\cite{Murakami_SP25,Cao_USENIX21,Wu_USENIX22}, which inject fake users and carefully craft data sent from the fake users to manipulate the statistical results. 

A recent study \cite{Murakami_SP25} addresses these two issues by introducing an \textit{augmented} shuffle model, where the shuffler performs additional operations, such as random sampling and adding dummies, before shuffling. 
Specifically, \cite{Murakami_SP25} proposes the LNF (Local-Noise-Free) protocol, in which each user sends her (encrypted) input data to the shuffler without adding noise, and the shuffler performs random sampling, adding dummies, and shuffling. 
The key idea of this protocol is to prevent malicious users' behavior by adding noise on the shuffler side rather than the user side. 
The protocol can also be easily implemented using any PKE (Public Key Encryption) scheme, such as RSA and ECIES. 
It is shown in \cite{Murakami_SP25} that the LNF protocol is robust to both collusion with users and poisoning from users while also providing higher accuracy than other shuffle protocols. 

Unfortunately, the LNF protocol in \cite{Murakami_SP25} is limited to 
a simple frequency estimation task with a small domain and cannot be applied to data analysis with a large domain. 
Specifically, both the communication and computational costs of the LNF protocol are linear in the number $d$ of items and are prohibitively large when $d$ is large. 
For example, our experimental results show that the LNF protocol would require about $100$ Terabits of communications and $3$ years of run time to calculate a frequency distribution when $d=10^9$ (see Section~\ref{sec:exp}). 
The LNF protocol also cannot be applied to a more complicated task, such as KV (Key-Value) statistics estimation~\cite{Ye_SP19,Ye_TDSC23,Gu_USENIX20}, with large $d$ for the same reason. 

In this work, we fill this gap by 
introducing a novel protocol for large-domain data. 
We first focus on frequency estimation over categorical data with a large domain and consider a simple extension of the LNF protocol 
that reduces the domain size using a hash function common to all users. 
We analyze the theoretical properties of this protocol and show that it suffers from low accuracy due to hash collision. 
We also show that this issue cannot be addressed by introducing a different hash function for each user (or each user group assigned in advance), as dummies need to be added for each hash function. 

To achieve high accuracy and efficiency, we propose a novel protocol called the \textit{FME (Filtering-with-Multiple-Encryption) protocol}. 
\colorB{Below, we explain its technical overview.}

\subsection{\colorB{Technical Overview}}

\noindent{\textbf{\colorB{Our Protocol.}}}~~In our FME protocol, we use a hash function to filter out unpopular items with low (or zero) frequencies rather than to calculate frequencies. 
Then, we accurately calculate frequencies for the selected items. 
This can be realized by introducing a two-round interaction between users and the shuffler, where each user sends her hash value for in the first round and her input data or a symbol ``$\bot$'' representing an unselected item in the second round. 
However, the two-round interaction significantly reduces the usability, as each user must respond to the shuffler twice. 
Moreover, it needs \textit{synchronization} \cite{Imola_CCS22} in that the shuffler must wait for all users' responses before shuffling in each round. 
Thus, the two-round protocol would not be practical for many systems.

We overcome this issue by 
\textit{replacing unselected items with $\bot$ on the data collector side} and \textit{introducing multiple encryption} \cite{AppliedCrypto_book}. 
Specifically, 
our FME protocol achieves \textit{one round of interaction} between users and the shuffler as follows. 
Each user sends her hash value and input data simultaneously to the shuffler. 
The shuffler performs augmented shuffling for the hash values and sends them and the corresponding input data to the data collector. 
The data collector filters items based on the shuffled hash values, replaces unselected items in the input data with $\bot$, and sends them back to the shuffler. 
Finally, the shuffler performs augmented shuffling for the input data, 
and the data collector calculates frequencies from the shuffled input data. 
Note that the input data are communicated between the shuffler and the data collector \textit{three times} in this protocol. 
Thus, a lot of information can be leaked by comparing them. 
To prevent this leakage, we use multiple encryption for the input data and have the shuffler and data collector decrypt the input data each time they receive them. 
We rigorously analyze 
the privacy 
of our protocol and prove it achieves 
computational DP~\cite{MPRV09,EIKN23}. 
We also show that it achieves high robustness against collusion and poisoning attacks, accuracy, and efficiency. 
Furthermore, we optimize the range of the hash function in terms of efficiency. 

Then, we apply our FME protocol to KV statistics estimation, where each user has KV pairs (e.g., movies and ratings) and the goal is to estimate the frequency and mean value for each key (item). 
For this data type, we propose an additional technique called \textit{TKV-FK (Transforming KV Pairs and Filtering Keys)}, which transforms KV pairs into one-dimensional data and filters the data at a key level to 
reduce 
bias in the estimates. 
We extensively evaluate our proposals and show they are effective for both categorical and KV data. 

\smallskip{}
\noindent{\textbf{\colorB{Technical Novelty.}}}~~\colorB{Our main technical novelty lies in a technique that carefully uses multiple encryption to provide (computational) DP within one round of interaction for users, and in the rigorous proof of DP. 
To our knowledge, we are the first to use multiple encryption to reduce the number of rounds while providing DP (see Section~\ref{sec:related} for details). 
In addition, our FME protocol does not use multiple encryption merely to implement onion routing~\cite{Scherer_PoPETS24}. 
More importantly, multiple encryption enables the data collector to \textit{remove ciphertexts corresponding to unpopular items without revealing this fact to the shuffler} by replacing them with ciphertexts that encrypt $\bot$. 
We prove the DP guarantee of such a protocol by reducing it to the security of a PKE scheme. 
Our work also includes other technical contributions, such as the optimization of the hash range and a technique to reduce bias (i.e., TKV-FK).}

\colorB{Furthermore, we prove many theorems that are not simple extensions of \cite{Murakami_SP25}. 
Examples include our theoretical results for KV data (e.g., Theorems~\ref{thm:KV_poisoning} and \ref{thm:KV_accuracy}), as \cite{Murakami_SP25} does not deal with KV data. 
In particular, the existing robustness analysis for KV data \cite{Cao_USENIX21} assumes that the number of each user's KV pairs does not exceed the padding length \cite{Gu_USENIX20}, which may not hold in practice. 
Theorem~\ref{thm:KV_poisoning} removes this assumption and shows the robustness of our protocol in a general setting where the number of KV pairs can exceed the padding length. 
Another example is Theorem~\ref{thm:FME_communication}, which analyzes the communication cost of our FME protocol based on the size of single, double, or triple ciphertexts and a bound on the expected number of selected items. This theorem is also new and serves as a basis for optimizing the hash range. 
We also present a new theorem (Theorem~\ref{thm:CH_accuracy}) that shows the low accuracy of a simple extension of the LNF protocol \cite{Murakami_SP25} using the common hash function under the 2-wise independence assumption \cite{Patrascu_TALG15}. 
}

\subsection{\colorB{Our Contributions}}
We make the following contributions: 
\begin{itemize}
\item We propose a novel augmented shuffle protocol called the FME protocol for large-domain categorical and KV data. 
Our protocol achieves high privacy, robustness, accuracy, and efficiency within one round of interaction for users using multiple encryption. 
For KV data, we propose an additional technique called TKV-FK to reduce bias. 
\item We demonstrate the effectiveness of our proposals through theoretical analysis and extensive experiments that compare ours with twelve state-of-the-art shuffle protocols (eight for categorical and four for KV data). 
\end{itemize}
Compared to the LNF protocol \cite{Murakami_SP25}, our protocol reduces the communication cost from about $100$ Terabits to $260$ Gigabits and the run time from about $3$ years to $1$ day ($d=10^9$). 
\colorB{Note that $d$ is much smaller than $10^9$ in most practical applications (e.g., there are $d = 8 \times 10^6$ census blocks in the US; Amazon has $d = 6 \times 10^8$ products in total~\cite{amazon_statistics}), in which case the communication and computational costs are also smaller.} 
We also show that our protocol can be applied to a system with large-scale users and items (e.g., item rating system~\cite{amazon_dataset}) by introducing user sampling. 
\conference{The proofs of all statements are given in our full paper~\cite{Murakami_arXiv25}.}\arxiv{The proofs of all statements are given in Appendices~\ref{sec:proofs_baseline} to \ref{sec:proofs_key-value}.} 
Our code is available in \cite{FME_code_Zenodo}.

\section{Related Work}
\label{sec:related}
\noindent{\textbf{Shuffle DP.}}~~The shuffle model of DP 
can be divided into two models: a \textit{pure shuffle model} and an \textit{augmented shuffle model}~\cite{Murakami_SP25}. 
Most existing work (e.g.,~\cite{Erlingsson_SODA19,Cheu_EUROCRYPT19,Balle_CRYPTO19,Girgis_CCS21,Feldman_FOCS21,Feldman_SODA23,Luo_CCS22,Balcer_ITC20,Cheu_SP22}) assumes the pure shuffle model, where the shuffler performs only shuffling. 
However, this model is vulnerable to collusion with users and data poisoning. 
This vulnerability is inevitable 
in this model because genuine users need to add noise to their input data; 
in this case, 
the data collector can increase $\epsilon$ by obtaining noisy data of some users, and fake users can effectively manipulate the statistical results by \textit{not} adding noise to their input data~\cite{Cao_USENIX21}. 

Beimel \textit{et al.}~\cite{Beimel_TCC20} show 
a pure shuffle protocol for general tasks, which achieves accuracy comparable to the central 
one. 
However, their protocol requires two rounds of interaction 
for 
users whereas ours requires 
only one round of interaction. 

Some protocols, including Google's Prochlo~\cite{Bittau_SOSP17}, assume the augmented shuffle model, where the shuffler performs additional operations, e.g., randomized thresholding~\cite{Bittau_SOSP17}, random sampling~\cite{Girgis_NeurIPS21}, and adding dummies~\cite{Wang_PVLDB20}. 
A recent study~\cite{Murakami_SP25} proposes the LNF protocol that does not add noise on the user side and shows that it is robust to both collusion and poisoning attacks. 
However, this protocol cannot be applied to large-domain data. 
We address this issue by introducing a novel protocol with multiple encryption.  

\smallskip{}
\noindent{\textbf{Collusion/Poisoning Attacks.}}~~Some studies~\cite{Wang_PVLDB20,Murakami_SP25} show that pure shuffle protocols are vulnerable to collusion with users. 
To address this issue, 
Wang \textit{et al.}~\cite{Wang_PVLDB20} 
add dummies uniformly at random from the domain of noisy data on the shuffler side. 
However, their protocol still suffers from the increase in $\epsilon$ by collusion with users, as shown in Appendix~\ref{sub:existing_defenses}. 

Data poisoning attacks have been studied in various data types, e.g., categorical~\cite{Cao_USENIX21,Cheu_SP21}, KV~\cite{Wu_USENIX22}, numerical~\cite{Li_USENIX23}, and set-valued data~\cite{Tong_CCS24}. 
Although these attacks assume the local model, they can also be applied to the shuffle model. 
Defenses 
have also been studied for categorical or KV data. 
The defenses in~\cite{Huang_TKDE24,Song_TIFS23,Kato_DBSec21,Horigome_MDAI23} introduce multiple rounds for users and reduce usability. 
The defenses in~\cite{Cao_USENIX21,Sun_ICDE24} have limited effectiveness~\cite{Murakami_SP25}.
We also show that the defense in~\cite{Wu_USENIX22} has limited effectiveness in Appendix~\ref{sub:existing_defenses}. 

\smallskip{}
\noindent{\textbf{Cryptographic Protocols.}}~~Multiple encryption has been studied 
in the field of cryptography~\cite{AppliedCrypto_book,Dodis_TCC05,Dai_CRYPTO14}. 
Its applications include key-insulated encryption~\cite{Dodis_TCC05}, onion routing~\cite{Scherer_PoPETS24}, mix-net~\cite{Furukawa_FC06}, and 
an instant messenger~\cite{GoldBug}. 
To our knowledge, we are the first to use multiple encryption to provide 
DP 
for distributed systems 
within one round of interaction for users.

Finally, a DP protocol using 
secure multi-party computation 
is proposed in \cite{Bell_CCS22} to calculate a frequency distribution. 
Their protocol requires a PKE scheme with a homomorphic property, 
whereas ours can use any PKE scheme based on a wider class of assumptions.
Moreover, the protocol in \cite{Bell_CCS22} 
cannot be applied to more advanced KV statistics estimation, as (i) each user may hold multiple KV pairs, and (ii) the data collector needs to estimate both frequency and mean 
for each key. 
In contrast, our protocol can be used to estimate KV statistics. 

\section{Preliminaries}
\label{sec:preliminaries}
\subsection{Notations}
\label{sub:notations}
Let $\reals$, $\nnreals$, $\nats$, and $\nnints$ be the sets of real numbers, non-negative real numbers, natural numbers, and non-negative integers, respectively. 
Let $n \in \nats$ be the number of users, and $d \in \nats$ be the number of items. 
For $i \in [n]$ $(= \{1,2,\ldots,n\})$, let $u_i$ be the $i$-th user. 
Let $\calX$ be the space of input data, and $x_i \in \calX$ be the input data of user $u_i$. 
Sections~\ref{sec:baseline} and \ref{sec:proposed} focus on frequency estimation over categorical data, whereas Section~\ref{sec:key-value} focuses on frequency and mean estimation over KV data. 
Below, we introduce the notations for each case. 

\smallskip{}
\noindent{\textbf{Categorical Data.}}~~For categorical data, 
we follow \cite{Murakami_SP25,Cao_USENIX21,Wang_USENIX17} and 
represent an item as an integer from $1$ to $d$. 
Each user's input data $x_i$ ($i \in [n]$) is an item, 
i.e., 
$\calX = [d]$. 
The data collector estimates the (relative) frequency $f_i \in [0,1]$ for each item $i \in [d]$. 
The frequency $f_i$ is given by $f_i = \frac{1}{n}\sum_{j=1}^n \mathone_{x_j = i}$, where $\mathone_{x_j = i}$ takes $1$ if $x_j=i$ and $0$ otherwise. 
We denote the estimate of $f_i$ by $\hf_i \in \reals$. 
Let $\bmf = (f_1, \ldots, f_d)$ and $\bmhf = (\hf_1, \ldots, \hf_d)$. 
The goal for the data collector is to calculate $\bmhf$ as close as possible to $\bmf$ under DP. 

\smallskip{}
\noindent{\textbf{KV Data.}}~~For KV data, we follow \cite{Wu_USENIX22,Ye_SP19,Ye_TDSC23,Gu_USENIX20} and represent a key (item) as an integer from $1$ to $d$ and a numerical value as a real value between $-1$ and $1$. 
Note that we can assume that the values are in the range $[-1,1]$ without loss of generality, as numerical values can be transformed into $[-1,1]$. 
In this use case, each user's input data $x_i$ ($i \in [n]$) is a set of KV pairs $\langle k, v \rangle$, where $k \in [d]$ and $v \in [-1,1]$. 
Note that each user has at most one KV pair per key, 
i.e., $\calX = \bigcup_{i=1}^d ([d] \times [-1,1])^i$. 
The data collector estimates the frequency $\Phi_i \in [0,1]$ and the mean value $\Psi_i \in [-1,1]$ for each key $i \in [d]$. 
They are given by 
\begin{align*}
\textstyle{\Phi_i = \frac{1}{n}\sum_{j=1}^n \mathone_{\langle k, \cdot \rangle \in x_j}, ~~ 
\Psi_i = \frac{1}{n \Phi_i}\sum_{j \in [n],\langle k, v \rangle \in x_j} v,} 
\end{align*}
where $\mathone_{\langle k, \cdot \rangle \in x_j}$ takes $1$ if $x_j$ includes key $k$ and $0$ otherwise. 
Let $\hPhi_i$ (resp.~$\hPsi_i$) $\in \reals$ be the estimates of $\Phi_i$ (resp.~$\Psi_i$). 
Let $\bmPhi = (\Phi_1, \ldots, \Phi_d)$, $\bmPsi = (\Psi_1, \ldots, \Psi_d)$, $\bmhPhi = (\hPhi_1, \ldots, \hPhi_d)$, and $\bmhPsi = (\hPsi_1, \ldots, \hPsi_d)$. 
The goal is to calculate $\bmhPhi$ and $\bmhPsi$ as close as possible to $\bmPhi$ and $\bmPsi$, respectively, under DP. 

We also summarize our notations in Table~\ref{tab:notations} of Appendix~\ref{sec:notation_table}. 

\subsection{Differential Privacy}
\label{sub:DP}
\colorB{In this work, we use DP~\cite{DP} and its computational version called CDP (Computational DP)~\cite{MPRV09,EIKN23} as privacy notions:} 

\begin{definition} [$(\epsilon,\delta)$-DP/CDP] \label{def:DP} 
Let $\epsilon \in \nnreals$ and $\delta \in [0,1]$. 
We say a randomized algorithm $\calM$ with domain 
$\calX^n$ 
provides \emph{$(\epsilon,\delta)$-DP} if for any neighboring databases 
$D = (x_1, \ldots, x_n) \in \calX^n$ and $D' = (x'_1, \ldots, x'_n) \in \calX^n$ 
that differ on one entry 
and any $S \subseteq \mathrm{Range}(\calM)$, 
\begin{align}
\Pr[\calM(D) \in S] \leq e^\epsilon \Pr[\calM(D') \in S] + \delta.
\label{eq:DP_inequality}
\end{align}
We also say $\calM$ 
provides \emph{$(\epsilon,\delta)$-CDP (Computational DP)} if for any 
attacker $\calA$ running in time polynomial in $\gamma \in \nats$, 
\begin{align}
\Pr[\calA(\calM(D)) = 1] \leq e^\epsilon \Pr[\calA(\calM(D')) = 1] + \delta + \mathsf{negl}(\gamma), 
\label{eq:DP_inequality_comp}
\end{align}
where 
$\mathsf{negl}$ is a function that approaches $0$ faster than the reciprocal of any polynomial in $\gamma$. 
\end{definition}
In our work, $\gamma$ coincides with the security parameter of a PKE scheme.
$(\epsilon,\delta)$-CDP (resp.~DP) can be used for shuffle DP protocols with (resp.~without) PKE schemes. 
In shuffle protocols, the shuffler or the data collector can be an attacker $\calA$. In either case, all messages the attacker receives during the protocols are outputs of $\calM$. 
Note that $\epsilon \geq 5$ is unsuitable in most applications~\cite{DP_Li}. 
$\delta$ should be much smaller than $\frac{1}{n}$~\cite{DP}.  

We also introduce LDP~\cite{Kasiviswanathan_FOCS08,Duchi_FOCS13}, \colorB{which can be used as a building block for pure shuffle DP protocols:} 
\begin{definition} [$\epsilon$-LDP] \label{def:LDP} 
Let $\epsilon \in \nnreals$ and $\delta \in [0,1]$. 
We say a randomized algorithm $\calR$ with domain $\calX$ 
provides \emph{$\epsilon$-LDP} if for any input values $x,x' \in \calX$ and any $S \subseteq \mathrm{Range}(\calR)$, 
\begin{align}
\Pr[\calR(x) \in S] \leq e^\epsilon \Pr[\calR(x') \in S].
\label{eq:LDP_inequality}
\end{align}
\end{definition}
Examples of LDP mechanisms $\calR$ include the GRR~\cite{Wang_PVLDB20}, RAPPOR~\cite{Erlingsson_CCS14}, OUE~\cite{Wang_USENIX17}, and OLH~\cite{Wang_USENIX17}. 

\subsection{Pure Shuffle Model}
\label{sub:shuffle}
\colorB{Below, we explain the pure shuffle model assumed in most existing shuffle DP protocols.} 
Assume 
that the data collector has a private key and publishes the corresponding public key. 
In the pure shuffle protocols, 
each user $u_i$ adds noise to her data $x_i$ using a randomized algorithm $\calR$, encrypts it using the public key, and sends the encrypted version of $\calR(x_i)$ to the shuffler. 
The shuffler randomly shuffles (encrypted) noisy data $\calR(x_1), \dots, \calR(x_n)$ and sends them to the data collector. 
The data collector decrypts the shuffled data using the private key. 

Under the assumption that the shuffler does not collude with the data collector, 
the data collector obtains only the shuffled data. 
The shuffled data provides $(\epsilon, \delta)$-DP, where 
$\epsilon = g(n,\delta)$ 
and $g$ is a monotonically decreasing function of $n$ and $\delta$. 
For example, when $\calR$ provides LDP, $g$ is expressed 
using the state-of-the-art privacy amplification result \cite{Feldman_SODA23} 
as follows: 

\begin{theorem}[Privacy amplification result in~\cite{Feldman_SODA23}]
\label{thm:privacy_amplification}
Let $\epsilon_0 \in \nnreals$ and $D = (x_1, \ldots, x_n) \in \calX^n$. 
Let $\calR: \calX \rightarrow \calY$ be a randomized algorithm. 
Let $\calM_S: \calX^n \rightarrow \calY^n$ be a pure shuffle algorithm that takes $D$ as input, samples a uniform random permutation $\pi$ over $[n]$, and outputs $\calM_S(D) = (\calR(x_{\pi(1)}), \ldots, \calR(x_{\pi(n)}))$. 
If $\calR$ provides $\epsilon_0$-LDP, then for any $\delta \in [0,1]$, $\calM_S$ provides $(\epsilon, \delta)$-DP with $\epsilon = g(n,\delta)$, where 
\begin{align}
g(n,\delta) = \textstyle{\ln\left(1 + (e^{\epsilon_0} - 1) \frac{4\sqrt{2 \ln(4/\delta)}}{\sqrt{(e^{\epsilon_0} + 1)n}} + \frac{4}{n} \right)}
\label{eq:g_SODA23}
\end{align}
if $\epsilon_0 \leq \ln(\frac{n}{8\ln(2/\delta)}-1)$ and $g(n,\delta)=\epsilon_0$ otherwise. 
\end{theorem}
By (\ref{eq:g_SODA23}), $\epsilon$ is much smaller than $\epsilon_0$ 
in the LDP mechanism 
when $n$ is large. 
In multi-message protocols~\cite{Luo_CCS22,Balcer_ITC20,Cheu_SP22}, $g(n,\delta)$ 
is 
different from (\ref{eq:g_SODA23}). 
See~\cite{Luo_CCS22,Balcer_ITC20,Cheu_SP22} for details.

\subsection{Communication Cost and Accuracy}
\label{sub:communication_utility}

\colorB{We use the following measures for the communication cost and accuracy (i.e., utility):} 

\smallskip{}
\noindent{\textbf{Communication Cost.}}~~Let $C_{U-S}$ (resp.~$C_{S-D}$) $\in \nnreals$ be the expected number of bits sent from users to the shuffler (resp.~from the shuffler to the data collector). 
Then, the 
expected total 
number $C_{tot}$ of bits sent from one party to another is given by $C_{tot} = C_{U-S} + C_{S-D}$. 
We use $C_{tot}$ as a measure of the communication cost.

\smallskip{}
\noindent{\textbf{Accuracy.}}~~For categorical data, we follow \cite{Wang_PVLDB20,Murakami_SP25,Wang_USENIX17,Kairouz_ICML16} and measure 
the expected squared error 
$\E[(\hf_i - f_i)^2]$ 
of the estimate $\tf_i$ for each item $i\in[d]$. 
If the estimate $\tf_i$ is unbiased, then $\E[(\hf_i - f_i)^2]$ is equal to the variance $\V[\hf_i]$. 
Similarly, we follow \cite{Gu_USENIX20} and measure the expected squared errors $\E[(\hPhi_i - \Phi_i)^2]$ and $\E[(\hPsi_i - \Psi_i)^2]$ ($i\in[d]$) for KV data. 

\section{\colorB{Collusion and Poisoning Attacks}}
\label{sec:collusion_poisoning}
\colorB{In this section, we explain collusion and poisoning attacks in detail. 
Section~\ref{sub:threat} defines our threat model and clarifies why we focus on these attacks. 
Sections~\ref{sub:collusion} and \ref{sub:data_poisoning} explain collusion and poisoning attacks, respectively.} 

\subsection{\colorB{Threat Model}}
\label{sub:threat}
\colorB{We 
assume that 
anyone except a single user (victim), including the shuffler, the data collector, and other users, 
can be an attacker who attempts to infer the input data of 
the victim. 
We assume that input data $x_1, \ldots, x_n$ are independent (which is necessary for DP guarantees~\cite{Kifer_PODS12}) and that the attacker can obtain input data of all users except the victim as background knowledge.} 

\colorB{We also assume that the data collector can collude with some users except the victim (or compromise their accounts) to obtain their noisy data sent to the shuffler. 
In addition, the attacker can inject fake users and send an arbitrary message from each fake user to manipulate the statistics. 
The former and latter attacks are collusion and poisoning attacks, respectively. 
We focus on these attacks because attacks by malicious users pose a threat in practice --~\cite{Thomas_USENIX13} shows that the attacker can inject a large number of fake users in practical systems. 
In particular, collusion attacks are threatening because $\epsilon$ can be increased from about $1$ to $8$ in the pure shuffle protocols when $10\%$ of users collude with the data collector; see Section~\ref{sub:collusion}. 
It is also difficult to know the number of colluding users (and hence the actual value of $\epsilon$) in these protocols.}

\colorB{As with most existing shuffle protocols, we assume that 
the shuffler and the data collector are semi-honest and do not collude with each other. 
Our protocol cannot achieve DP or accurate estimates when 
these servers deviate 
from the protocol (e.g., 
when the shuffler 
leaks data sent from users or does not perform shuffling; when the data collector alters the estimates). 
One way to address this issue is to 
ensure data confidentiality and enforce 
the servers 
to follow the protocol 
via a legally binding contract or 
a TEE (Trusted Execution Environment)~\cite{Bittau_SOSP17,Allen_NeurIPS19}. 
Another way is to provide security against 
the servers who deviate 
from the protocol (i.e., malicious security) by using MPC (Multi-Party Computation)~\cite{Xu_CCS24}. 
We leave the use of the TEE or MPC for future work.}

\subsection{Collusion with Users}
\label{sub:collusion}
Wang \textit{et al.}~\cite{Wang_PVLDB20} point out that 
$\epsilon$ 
in the pure shuffle model 
can be increased when the data collector 
colludes with some users 
to obtain their noisy data sent to the shuffler. 

Specifically, 
let $\Omega \subset [n]$. 
Assume that the data collector colludes with users $\{u_i | i \in \Omega\}$ and obtains their noisy data 
sent to the shuffler. 
Then, the privacy budget $\epsilon$ for the remaining users is increased from $g(n,\delta)$ to $g(n-|\Omega|,\delta)$~\cite{Murakami_SP25}. 
For example, if $n=10^6$, $|\Omega|=10^5$, and $\delta=10^{-12}$ in Theorem~\ref{thm:privacy_amplification}, then $\epsilon$ can be increased from $1.1$ to $8.3$. 

Ideally, the value of $\epsilon$ should not be increased even if 
some users share their data sent to the shuffler with the data collector. 
Below, we formally define such \textit{robustness to collusion with users}. 
We first introduce $\Omega$-neighboring databases~\cite{Beimel_CRYPTO08}: 

\begin{definition} [$\Omega$-Neighboring databases] \label{def:omega_neighboring} 
Let $\Omega \subset [n]$. 
We say two databases 
$D = (x_1, \ldots, x_n) \in \calX^n$ and $D' = (x'_1, \ldots, x'_n) \in \calX^n$ 
are \emph{$\Omega$-neighboring} if they differ on one entry whose index $i$ is \emph{not} in $\Omega$, i.e., $x_i \ne x'_i$ for some $i \notin \Omega$ and $x_j = x'_j$ for any $j \ne i$. 
\end{definition}

$\Omega$-neighboring databases consider an attacker who colludes with users $\{u_i | i \in \Omega\}$. 
Based on this, we can define the robustness to collusion with users as follows: 
\begin{definition} [Robustness to collusion with users] 
\label{def:robustness_to_collusion}
Let $\calM$ be a shuffle protocol that provides $(\epsilon,\delta)$-DP (or CDP). 
For $i \in [n]$, let $\nu_i$ be data sent from user $u_i$ to the shuffler in $\calM$. 
For $\Omega \subset [n]$, let $\calM_\Omega$ be a protocol that takes a database 
$D \in \calX^n$ 
as input and outputs $\calM_\Omega(D) = (\calM(D), (\nu_i)_{i\in\Omega})$. 
We say $\calM$ is \emph{robust to collusion with users} if 
for any $\Omega \subset [n]$, 
any $\Omega$-neighboring databases $D$ and $D'$, 
and any $S \subseteq \mathrm{Range}(\calM_\Omega)$, 
$\calM_\Omega$ also satisfies 
(\ref{eq:DP_inequality}) (or (\ref{eq:DP_inequality_comp})), 
i.e., if $\epsilon$ and $\delta$ are not increased by collusion with users. 
\end{definition}

The data $\nu_i$ sent from user $u_i$ to the shuffler depends on the protocol $\calM$. 
For example, 
in the baseline protocol in Section~\ref{sub:CH_protocol}, $\nu_i = h(x_i)$, where $h$ is a hash function. 
In our protocol in Section~\ref{sec:proposed}, $\nu_i = (x_i,h(x_i))$. 

Unfortunately, pure shuffle protocols cannot provide the robustness in Definition~\ref{def:robustness_to_collusion}, 
as they need to add noise on the user side. 
The robustness in Definition~\ref{def:robustness_to_collusion} can be achieved by introducing the augmented shuffle model and adding noise on the shuffler side, as shown 
in Sections~\ref{sec:baseline} to \ref{sec:key-value}. 

\subsection{Data Poisoning Attacks}
\label{sub:data_poisoning}
Following~\cite{Cao_USENIX21,Wu_USENIX22}, we consider the following targeted attacks as data poisoning attacks. 
Let $\calT \subseteq [d]$ be the set of target items. 
We assume that the attacker injects $n' \in \nats$ fake users; there are $n+n'$ users in total, including $n$ genuine users. 
Each fake user can send an arbitrary message to the shuffler. 
This is called the output poisoning attack~\cite{Li_USENIX23}. 
For $i \in [n']$, let $m_i$ be a message sent from the $i$-th fake user. 
Let $\bmm = (m_1, \ldots, m_{n'})$ be the messages of $n'$ fake users.

\smallskip{}
\noindent{\textbf{Categorical Data.}}~~For categorical data, the attacker 
attempts 
to increase the estimates for the target items $\calT$ (i.e., to promote $\calT$) as much as possible. 
Formally, let $\hf'_i \in \reals$ be the estimate of $f_i$ after poisoning. 
Then, the attacker's \textit{overall gain} is defined as $G_f(\bmm) = \sum_{i \in \calT} \E[\hf'_i - \hf_i]$~\cite{Cao_USENIX21}. 
The attacker's goal is to maximize $G_f(\bmm)$. 

Let $\GMGA = \max_{\bmm} G_f(\bmm)$ be the maximum value of $G_f(\bmm)$. 
Cao \textit{et al.}~\cite{Cao_USENIX21} propose the MGA (Maximal Gain Attack) that crafts the messages $\bmm$ to achieve $\GMGA$. 
We use the maximum gain $\GMGA$ to measure the robustness to data poisoning attacks in categorical data.

\smallskip{}
\noindent{\textbf{KV Data.}}~~For KV data, the attacker attempts 
to maximize the frequency and mean estimates for the target items $\calT$ (i.e., to promote $\calT$). 
Let $\hPhi'_i$ (resp.~$\hPsi'_i$) $\in \reals$ be the estimate of $\Phi_i$ (resp.~$\Psi_i$) after poisoning. 
Then, the \textit{frequency gain} and the \textit{mean gain} are given by $G_\Phi(\bmm) = \sum_{i \in \calT} \E[\hPhi'_i - \hPhi_i]$ and $G_\Psi(\bmm) = \sum_{i \in \calT} \E[\hPsi'_i - \hPsi_i]$, respectively~\cite{Wu_USENIX22}. 
The 
goal is to maximize $G_\Phi(\bmm)$ and $G_\Psi(\bmm)$ simultaneously.  

Let $\GMTGAPhi = \max_{\bmm} G_\Phi(\bmm)$ and $\GMTGAPsi = \max_{\bmm} G_\Psi(\bmm)$. 
Wu \textit{et al.}~\cite{Wu_USENIX22} propose the M2GA (Maximal Gain Attack) that crafts $\bmm$ to achieve $\GMTGAPhi$ and $\GMTGAPsi$ simultaneously. 
We use the maximum frequency gain $\GMTGAPhi$ and the maximum mean gain $\GMTGAPsi$ as robustness measures in KV data. 

\section{Baseline Protocols}
\label{sec:baseline}
In this section, we present 
two 
baseline protocols for categorical data and explain 
why they are unsuitable 
for large-domain data. 
Section~\ref{sub:LNF} describes the LNF (Local-Noise-Free) protocol in \cite{Murakami_SP25} and explains that it suffers from prohibitively high communication and computational costs. 
Section~\ref{sub:CH_protocol} introduces the CH (Common Hash) protocol to address this issue and shows that it suffers from low accuracy. 

\subsection{Local-Noise-Free Protocol}
\label{sub:LNF}

\noindent{\textbf{Protocol.}}~~Fig.~\ref{fig:LNF} shows the overview of the LNF protocol. 
\colorB{We show an algorithmic description of the LNF protocol in Algorithm~\ref{alg:S_LNF} of Appendix~\ref{sub:LNF_algorithm}.} 
We denote the LNF protocol by $\calS_{\calD, \beta}^{\LNF}$.  
$\calS_{\calD, \beta}^{\LNF}$ 
has a \textit{dummy-count distribution $\calD$} over $\nnints$ with mean $\mu \in \nnreals$ and variance $\sigma^2 \in \nnreals$ and a \textit{sampling probability} $\beta \in [0,1]$ as parameters.

The LNF protocol $\calS_{\calD, \beta}^{\LNF}$ is simple and works as follows (we omit the encryption/decryption process). 
First, each user $u_i$ ($i \in [n]$) sends her input value $x_i \in [d]$ without adding noise. 
Then, the shuffler performs three operations: random sampling, adding dummies, and shuffling. 
Specifically, the shuffler randomly selects each input value with probability $\beta$. 
For each item $i \in [d]$, the shuffler randomly generates $z_i$ from the dummy-count distribution $\calD$ ($z_i \sim \calD$) and adds $z_i$ dummy values. 
Let $y_1, \ldots, y_{\tn} \in [d]$ be the selected input values and dummies, where $\tn \in \nnints$ is the total number of these values. 
The shuffler samples a random permutation $\pi$ over $[\tn]$ and sends 
$y_{\pi(1)}, \ldots, y_{\pi(\tn)}$ to the data collector. 
In 
Fig.~\ref{fig:LNF}, 
$y_{\pi(1)}, \ldots, y_{\pi(\tn)} = (1,2,3,3,2,1,3,1)$ ($\tn = 8$). 

After receiving the shuffled values $y_{\pi(1)}, \ldots, y_{\pi(\tn)}$, the data collector calculates their histogram. 
Specifically, the data collector 
calculates a count (absolute frequency) $\tc_i \in \nnints$ for each item $i \in [d]$ from $y_{\pi(1)}, \ldots, y_{\pi(\tn)}$. 
\colorB{In the example of Fig.~\ref{fig:LNF}, $(\tc_1,\tc_2,\tc_3) = (3,2,3)$.} 
Finally, the data collector calculates an unbiased estimate $\hf_i$ of $f_i$ as $\hf_i = \frac{1}{n\beta}(\tc_i-\mu)$ and outputs $\bmhf = (\hf_1, \cdots, \hf_d)$. 

\smallskip{}
\noindent{\textbf{Theoretical Properties.}}~~$\calS_{\calD, \beta}^{\LNF}$ provides DP and is robust to both data poisoning and collusion attacks if a simpler mechanism called the \textit{binary input mechanism} provides DP: 

\begin{definition} [Binary input mechanism] \label{def:binary_input} 
Let $\calD$ be a dummy-count distribution over $\nnints$. 
The binary input mechanism $\calM_{\calD, \beta}: \{0,1\} \rightarrow \nnints$ 
with parameters $\calD$ 
and $\beta \in [0,1]$ 
takes binary data $x \in \{0,1\}$ and outputs 
\begin{align*}
\calM_{\calD, \beta}(x) = \alpha x + z,
\end{align*}
where $\alpha$ and $z$ are random variables that follow the Bernoulli distribution $\text{Ber}(\beta)$ 
and 
$\calD$, respectively 
($\alpha \sim \text{Ber}(\beta)$, $z \sim \calD$). 
\end{definition}

\begin{theorem}[\cite{Murakami_SP25}]
\label{thm:LNF_DP}
If $\calM_{\calD, \beta}$ provides $(\frac{\epsilon}{2}, \frac{\delta}{2})$-DP, then 
$\calS_{\calD, \beta}^{\LNF}$ 
provides $(\epsilon, \delta)$-CDP\footnote{Note that $\calM_{\calD, \beta}$ provides CDP, as it uses a PKE scheme. 
Specifically, it provides $(\epsilon, \delta)$-DP for the data collector and $(0,0)$-CDP for the shuffler. 
In addition, 
$\epsilon$ and $\delta$ in 
$\calS_{\calD, \beta}^{\LNF}$ 
are doubled. 
This is because neighboring data $x,x' \in \{0,1\}$ in $\calM_{\calD, \beta}$ differ by 1 in one dimension, whereas neighboring databases $D,D' \in [d]^n$ in $\calS_{\calD, \beta}^{\LNF}$ differ by 1 in two dimensions.} and is robust to collusion with users.
\end{theorem}

\begin{figure}[t]
  \centering
  \includegraphics[width=0.99\linewidth]{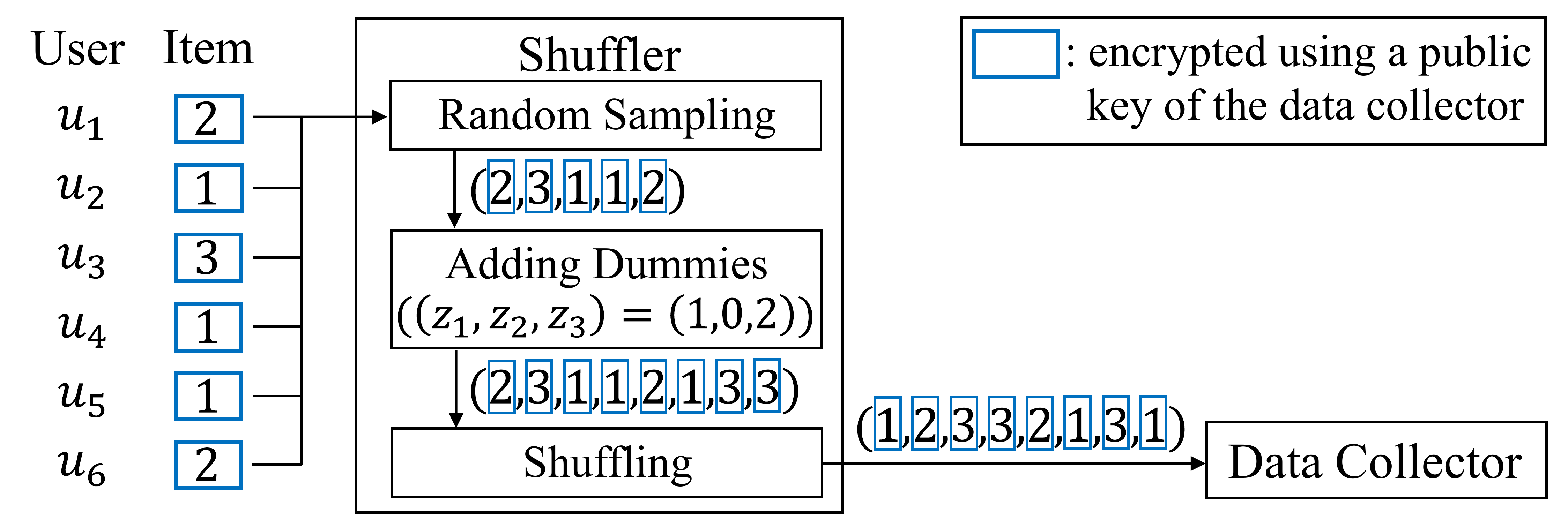}
  \vspace{-7mm}
  \caption{Overview of the LNF protocol ($n=6$, $d=3$). 
  In this example, the shuffler discards input data of $u_2$ and adds $(z_1,z_2,z_3) = (1,0,2)$ dummies.} 
  \label{fig:LNF}
\end{figure}

\begin{theorem}[\cite{Murakami_SP25}]
\label{thm:LNF_poisoning}
Let $\lambda = \frac{n'}{n+n'}$ and $f_\calT = \sum_{i \in \calT} f_i$. 
$\calS_{\calD, \beta}^{\LNF}$ provides the following robustness 
against poisoning attacks:
\begin{align}
\GMGA = \lambda (1 - f_\calT).
\label{LNF_GMGA}
\end{align}
\end{theorem}

$\GMGA$ in (\ref{LNF_GMGA}) does not depend on $\epsilon$ and is smaller than existing pure shuffle protocols \cite{Cao_USENIX21,Luo_CCS22,Balcer_ITC20,Cheu_SP22}. 
Thus, Theorems~\ref{thm:LNF_DP} and \ref{thm:LNF_poisoning} 
mean that 
$\calS_{\calD, \beta}^{\LNF}$ provides DP and is robust to collusion and data poisoning attacks if 
$\calM_{\calD, \beta}$ provides DP. 

In addition, $\calS_{\calD, \beta}^{\LNF}$ 
achieves the following accuracy: 
\begin{theorem}[\cite{Murakami_SP25}]
\label{thm:LNF_accuracy}
For any $i\in[d]$, $\calS_{\calD, \beta}^{\LNF}$ outputs an unbiased estimate 
(i.e., $\E[\hf_i] = f_i$) and achieves the following expected squared error: 
\begin{align}
\textstyle{\E[ (\hf_i - f_i)^2 ]} = \textstyle{\frac{f_i (1-\beta)}{n \beta} + \frac{\sigma^2}{n^2 \beta^2}}.
\label{eq:LNF_l2_loss}
\end{align}
\end{theorem}
The error in (\ref{eq:LNF_l2_loss}) decreases as $\sigma^2$ decreases and $\beta$ increases. 

\smallskip{}
\noindent{\textbf{Dummy-Count Distribution $\calD$.}}~~Examples of $\calD$ providing DP include the binomial distribution and the asymmetric geometric distribution. 
In particular, \cite{Murakami_SP25} shows that the asymmetric geometric distribution provides higher accuracy than the binomial distribution and existing pure shuffle protocols when $\beta=1$ and achieves pure DP ($\delta = 0$) when $\beta = 1 - e^{-\epsilon/2}$.

\smallskip{}
\noindent{\textbf{Limitations.}}~~The drawback of $\calS_{\calD, \beta}^{\LNF}$ is that it cannot be applied to large-domain data due to its high communication and computational costs. 
Specifically, the communication costs 
of $\calS_{\calD, \beta}^{\LNF}$ are: $C_{U-S} = \tau n$, $C_{S-D} = \tau(\beta n + \mu d)$, and 
\begin{align}
C_{tot} = \tau((1 + \beta) n + \mu d), 
\label{eq:LNF_C_tot}
\end{align}
where $\tau$ is the size of 
each ciphertext. 
Thus, $C_{tot}$ can be expressed as $O(n+d)$. 
Note that we can effectively reduce $n$ by randomly sampling users, as shown in our experiments. 
However, we cannot reduce $d$ by randomly sampling items without losing accuracy. 
Moreover, 
the mean $\mu$ of the dummy-count distribution $\calD$ 
is much larger than $1$, 
e.g., 
$\mu = 108$ 
in the asymmetric geometric distribution 
($\beta=1$, $\epsilon=1$, $\delta=10^{-12}$). 
Thus, 
$C_{tot}$ in (\ref{eq:LNF_C_tot}) can be prohibitively large when $d$ is large. 
For example, when 
$n=10^5$, $d=10^9$, $\mu = 108$, $\beta = 1$, 
and the $2048$-bit RSA is used, 
$\calS_{\calD, \beta}^{\LNF}$ 
needs 
$C_{tot} = 220$ Terabits. 

Similarly, the computational cost of $\calS_{\calD, \beta}^{\LNF}$ is $O(n+d)$ and can be prohibitively large when $d$ is large. 
For example, $\calS_{\calD, \beta}^{\LNF}$ requires about $3$ years 
when $d=10^9$ in our experiments. 

\subsection{Common Hash Protocol}
\label{sub:CH_protocol}
\noindent{\textbf{Protocol.}}~~The LNF protocol $\calS_{\calD, \beta}^{\LNF}$ cannot be applied to categorical data with large domain size $d$, as both the communication and computational costs are $O(n+d)$. 
A natural approach to improving the efficiency would be to reduce the domain size using a hash function. 
Our baseline, called the CH (Common Hash) protocol, 
simply uses a hash function common to all users. 
We denote this baseline by 
$\calS_{\calD, \beta}^{\CH}$.

Specifically, let $\calH$ be a universal hash function family whose domain is $[d]$ and whose range is $[b]$ ($b \leq d$). 
The CH protocol $\calS_{\calD, \beta}^{\CH}$ randomly selects a hash function $h: [d] \rightarrow [b]$ from $\calH$, applies $h$ 
to input values $x_1,\ldots,x_n$, and runs the LNF protocol $\calS_{\calD, \beta}^{\LNF}$. 
In other words, 
$\calS_{\calD, \beta}^{\CH}$ runs 
Algorithm~\ref{alg:S_LNF} (lines 1-11) 
with inputs $(h(x_1), \ldots, h(x_n))$, $b$, $\calD$, and $\beta$. 
Then, 
the data collector obtains $(\tc_1, \ldots, \tc_b)$, where $\tc_j$ ($j \in [b]$) is a count of a hash value $j$ calculated from the shuffled values.
Finally, the data collector calculates an unbiased estimate $\hf_i$ ($i \in [d]$) of $f_i$ as $\hf_i = \frac{b}{n\beta(b-1)}(\tc_{h(i)} - \frac{n\beta}{b} - \mu)$.

\smallskip{}
\noindent{\textbf{Theoretical Properties.}}~~We show DP and the robustness of $\calS_{\calD, \beta}^{\CH}$ in Appendix~\ref{sub:CH_DP_robustness}. 
Below, we analyze the communication cost of 
$\calS_{\calD, \beta}^{\CH}$. 
Since $\calS_{\calD, \beta}^{\CH}$ reduces the domain size from $d$ to $b$ via the hash function $h$, the communication costs 
of $\calS_{\calD, \beta}^{\CH}$ are: $C_{U-S} = \tau n$, $C_{S-D} = \tau(\beta n + \mu b)$, and 
\begin{align}
C_{tot} = \tau((1 + \beta) n + \mu b), 
\label{eq:CH_C_tot}
\end{align}
where $\tau$ is the size of each ciphertext. 
$C_{tot}$ in (\ref{eq:CH_C_tot}) can be expressed as $O(n + b)$. 
Similarly, the computational cost of $\calS_{\calD, \beta}^{\CH}$ is $O(n + b)$. 
Therefore, 
the CH protocol $\calS_{\calD, \beta}^{\CH}$ is much more efficient than the LNF protocol $\calS_{\calD, \beta}^{\LNF}$ when $b \ll d$. 

\smallskip{}
\noindent{\textbf{Limitations.}}~~Unfortunately, 
the CH protocol $\calS_{\calD, \beta}^{\CH}$ 
is also unsuitable for large-domain data because it suffers from low accuracy due to hash collision. 

Specifically, assume that the hash function $h$ is 2-wise independent; i.e., for any $i_1,i_2 \in [d]$ and $j_1,j_2 \in [b]$, $\Pr(h(i_1) = j_1 | h(i_2) = j_2) = \Pr(h(i_1) = j_1) = \frac{1}{b}$. 
For example, let $p \in [d, 2d)$ be a prime, $a_1 \in [p-1]$, and $a_0 \in [p]$. 
Then, a hash function $h$ defined by $h(x) = ((a_1 x + a_0) \bmod p) \bmod b$ is (almost) 2-wise independent \cite{Patrascu_TALG15}. 
This hash function is used in \cite{Luo_CCS22} and is also used in our experiments. 
Under this assumption, the accuracy of $\calS_{\calD, \beta}^{\CH}$ can be quantified as follows: 

\begin{theorem}
\label{thm:CH_accuracy}
For any $i\in[d]$, $\calS_{\calD, \beta}^{\CH}$ 
with a 2-wise independent hash function $h$ 
outputs an unbiased estimate 
(i.e., $\E[\hf_i] = f_i$) 
and achieves the following expected squared error: 
\begin{align}
\textstyle{\E[ (\hf_i - f_i)^2} ] = \textstyle{\frac{b^2}{n^2 \beta^2 (b-1)^2} \left(nf_i \beta(1-\beta) + \sigma^2 + \omega \right),}
\label{eq:CH_l2_loss}
\end{align}
where 
\begin{align*}
\hspace{-1mm} \omega &= \textstyle{\sum_{j=1,j \ne i}^d \left( \frac{( n^2 f_j^2 \beta^2 + n f_j \beta(1-\beta))(b-1)}{b^2} + \frac{n f_j \beta(1-\beta)}{b^2} \right).}
\end{align*}
\end{theorem}
The first and second terms in (\ref{eq:CH_l2_loss}) are almost the same as the squared error in (\ref{eq:LNF_l2_loss}). 
However, the third term in (\ref{eq:CH_l2_loss}) is introduced by 
hash collision and is very large when $b$ is small. 
For example, when $\beta = 1$ and $b = n$, 
the squared error in (\ref{eq:LNF_l2_loss}) is $\frac{\sigma^2}{n^2}$. 
In contrast, the squared error in (\ref{eq:CH_l2_loss}) is about $\frac{\sigma^2 + n}{n^2}$ in the worst case. 
In other words, the squared error is increased from $O(n^{-2})$ to $O(n^{-1})$ due to the hash collision. 
We also show that $\calS_{\calD, \beta}^{\CH}$ suffers from low accuracy in our experiments. 

\colorB{In Appendix~\ref{sub:other_baselines}, we describe other baselines than $\calS_{\calD, \beta}^{\LNF}$ and $\calS_{\calD, \beta}^{\CH}$ and explain that they also suffer from low accuracy.}

\section{Filtering-with-Multiple-Encryption Protocol}
\label{sec:proposed}

In this section, we propose a novel DP protocol called the \textit{FME (Filtering-with-Multiple-Encryption) protocol} to address the issues in the baselines. 
Section~\ref{sub:technical_motivation} explains its technical motivation 
and overview. 
Section~\ref{sub:FME_protocol} describes the details of our FME protocol. 
Section~\ref{sub:FME_properties} analyzes its theoretical properties. 
Section~\ref{sub:FME_optimization} proposes a method to optimize the range $b$ of the hash function in our FME protocol. 

\subsection{Technical Motivation and Overview}
\label{sub:technical_motivation}

\smallskip{}
\noindent{\textbf{Technical Motivation.}}~~As explained in Section~\ref{sub:CH_protocol}, hash values $h(x_1), \ldots, h(x_n)$ cannot be used for frequency estimation due to hash collision. 
However, they can be used to \textit{filter out} unpopular items with low (or zero) frequencies. 
This is helpful in large-domain data because the frequency distribution $\bmf$ is 
\textit{sparse}; e.g., most items have a count of zero when $n \ll d$. 
After reducing the input domain 
by filtering, 
we can efficiently use the LNF protocol 
to obtain high accuracy. 

Based on this, we can consider a \textit{two-round} protocol shown in Fig.~\ref{fig:two_round}. 
In the first round, each user $u_i$ sends her hash value $h(x_i)$ to the shuffler. 
The shuffler performs augmented shuffling (i.e., sampling, adding dummies to hash values, and shuffling) and sends the shuffled hash values to the data collector. 
The data collector filters out unpopular items and sends a set $\Lambda \subseteq [d]$ of popular items to the users. 
In the second round, each user $u_i$ replaces her input value $x_i$ with $\bot$ representing an \textit{unselected item} if $x_i \notin \Lambda$. 
Then, the parties run the LNF protocol 
for selected items $\Lambda$ and $\bot$; i.e., $[d]$ is replaced with $\Lambda \cup \{\bot\}$. 

\begin{figure}[t]
  \centering
  \includegraphics[width=0.99\linewidth]{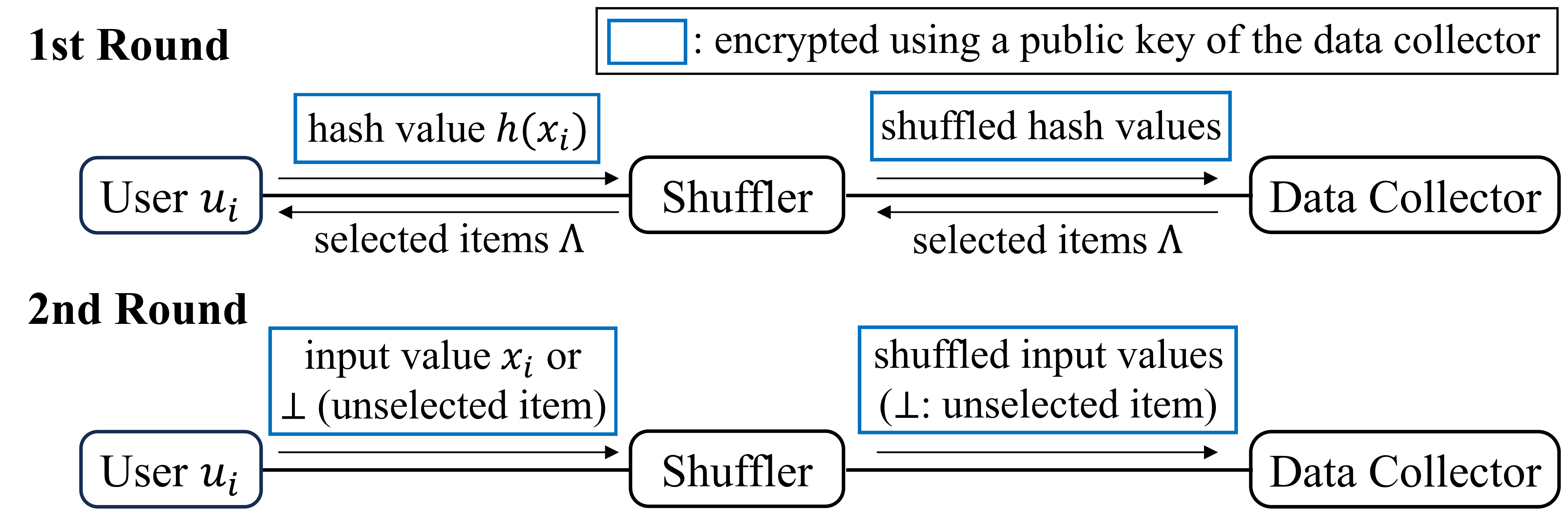}
  \vspace{-8mm}
  \caption{Two-round protocol using a hash function $h$.} 
  \label{fig:two_round}
  \centering
  \includegraphics[width=0.99\linewidth]{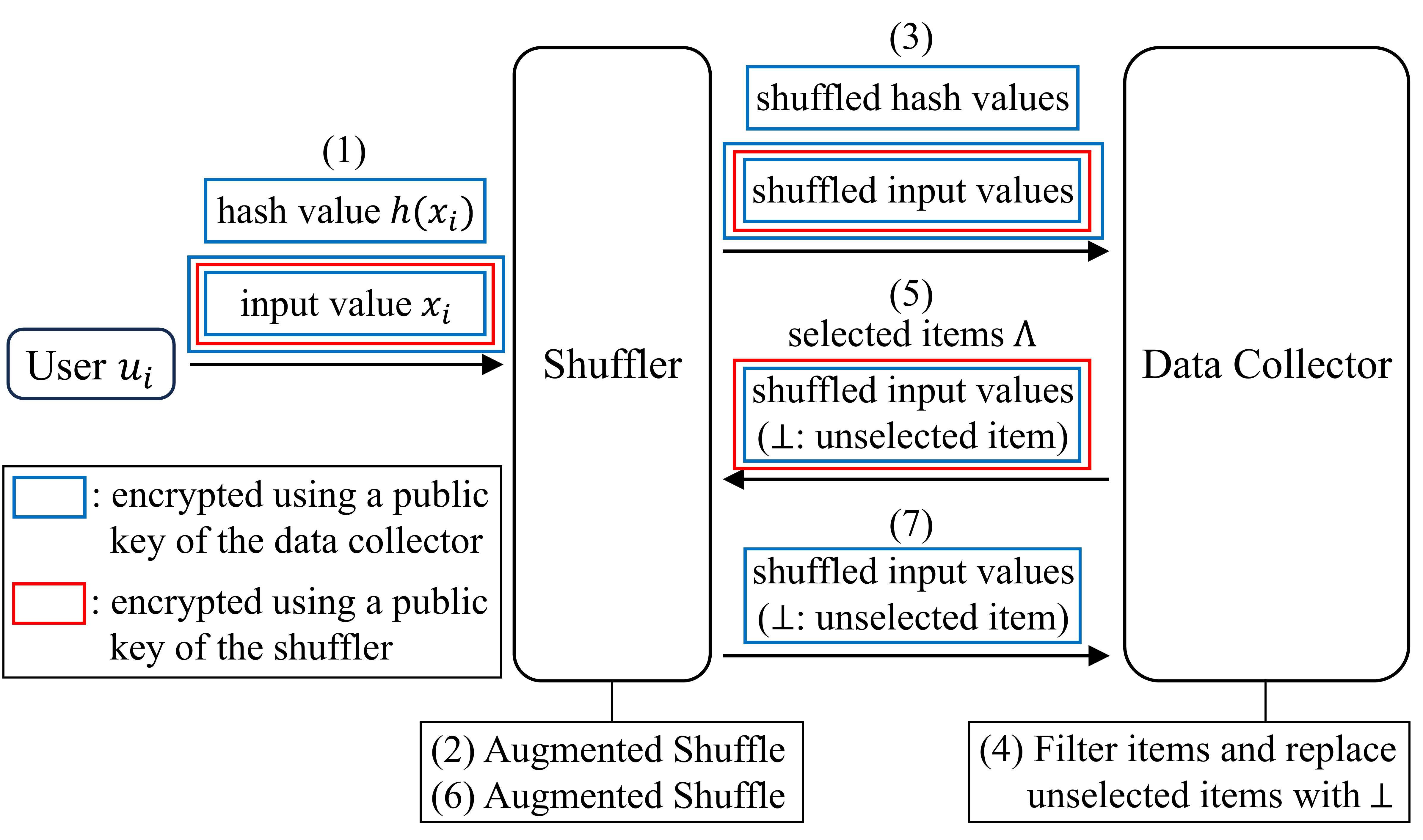}
  \vspace{-8mm}
  \caption{Overview of our FME protocol. 
  It reduces the number of rounds between users and the shuffler in Fig.~\ref{fig:two_round} by introducing multiple encryption.}
  \label{fig:proposal}
\end{figure}

The two-round protocol provides high accuracy for popular items $\Lambda$, as it uses the LNF protocol for them. 
In addition, 
the second round can be efficiently performed since the domain is restricted to $\Lambda \cup \{\bot\}$. 
Similar approaches have been taken in \cite{Ye_TDSC23,Qin_CCS16} -- 
they introduce the first round to find popular items and the second round to calculate statistics for them. 

However, 
the two-round protocol is not desirable for many practical systems, 
as it significantly 
requires a great deal of effort from users and synchronization, as described in Section~\ref{sec:intro}. 
We are interested in providing 
high accuracy and communication/computation efficiency 
\textit{without} introducing two rounds of interaction between users and the shuffler.

\smallskip{}
\noindent{\textbf{Overview.}}~~Our key idea is to 
remove the two rounds of interaction between users and the shuffler in Fig.~\ref{fig:two_round} by \textit{replacing unselected items with $\bot$ on the data collector side} and \textit{introducing multiple encryption}. 
Fig.~\ref{fig:proposal} shows the overview of our FME protocol that embodies our key idea. 

For ease of explanation, we begin with our FME protocol without encryption. 
In our protocol, each user $u_i$ sends her hash value $h(x_i)$ and input value $x_i$ simultaneously. 
After receiving each user's pair $\la h(x_i), x_i \ra$ ($i \in [n]$), the shuffler 
performs 
augmented shuffling for the pairs, where dummies are added to hash values (i.e., in the form of $\la h(i), \bot \ra$ for $i\in[b]$), and 
sends 
the shuffled pairs to the data collector. 
Then, the data collector 
selects 
a set $\Lambda$ of popular items based on the shuffled hash values and 
replaces 
unselected items $x_i \notin \Lambda$ in the shuffled pairs with $\bot$. 
After that, the data collector 
sends 
$\Lambda$ and the shuffled input values (including $\bot$) 
back to the shuffler. 
Note that the shuffled input values are the same as the ones the shuffler receives in the second round in Fig.~\ref{fig:two_round} (except that dummies are added to $\bot$). 
Thus, 
as with Fig.~\ref{fig:two_round}, 
the shuffler 
performs 
augmented shuffling for $\Lambda$ and $\bot$ and 
sends 
shuffled input values to the data collector. 

The above protocol achieves \textit{one round of interaction} between users and the shuffler by 
replacing unselected items with $\bot$ on the data collector side rather than the user side. 
Note, however, that we must handle the shuffled input values very carefully to prevent information leakage. 
For example, if each user $u_i$ encrypts $\la h(x_i), x_i \ra$ using a public key of the data collector, then the data collector can obtain $x_1,\ldots,x_n$ by decrypting them. 
Moreover, the shuffled input values are communicated between the shuffler and the data collector \textit{three times} (shuffler $\rightarrow$ data collector $\rightarrow$ shuffler $\rightarrow$ data collector), 
and 
a lot of information can be leaked by comparing them. 
For example, the shuffler would know whose input values are replaced with $\bot$ by comparing the first and second shuffled data. 
The data collector would know which values are added as dummies by comparing the second and third shuffled data. 

We address this issue by introducing multiple encryption, as shown in Fig.~\ref{fig:proposal}. 
Specifically, each user $u_i$ encrypts her input value $x_i$ \textit{three times} using public keys of the data collector, the shuffler, and the data collector. 
The data collector and the shuffler decrypt the shuffled data each time they receive them. 
Then, intuitively, the above information leakage can be prevented for two reasons: (i) the data collector cannot see the contents of the first shuffled data, and (ii) the three shuffled data are completely different from each other. 
We prove that this approach provides (computational) DP in Section~\ref{sub:FME_properties}.

\subsection{Details}
\label{sub:FME_protocol}

\setlength{\algomargin}{5mm}
\begin{algorithm*}[pt]
  \SetAlgoLined
  \KwData{Input values $(x_1, \ldots, x_n) \in [d]^n$, 
  hash function $h: [d] \rightarrow [b]$, 
  \colorB{dummy-count distributions $\calD^* = (\calD_1,\calD_2)$ ($\calD_1$ has mean $\mu_1$ and variance $\sigma_1^2$; $\calD_2$ has mean $\mu_2$ and variance $\sigma_2^2$)}, 
  sampling probability $\beta \in [0,1]$, 
  significance level $\alpha \in [0,1]$, maximum number of selected hashes $l \in [b]$.
  }
  \KwResult{Selected items $\Lambda \subseteq [d]$, estimates $\{\hf_i | i \in \Lambda \}$.}
  \tcc{Send hash values and input values (users $\rightarrow$ shuffler)}
  \ForEach{$i \in [n]$}{
    [$u_i$] Send $\la E_\pkd[h(x_i)], E_\pkd[E_\pks[E_\pkd[x_i]]] \ra$ to the shuffler\;
  }
  \tcc{Random sampling}
  [s] Sample $\la E_\pkd[h(x_i)], E_\pkd[E_\pks[E_\pkd[x_i]]] \ra$ $(i\in[n])$ with probability $\beta$\;
  \tcc{Dummy hash data addition}
  \ForEach{$i \in [b]$}{
    [s] $z_i \leftarrow \colorB{\calD_1}$;
    Add a dummy $\la E_\pkd[i], E_\pkd[E_\pks[E_\pkd[\bot]]] \ra$ for $z_i$ times\;
  }
  \tcc{Random shuffling}
  [s] Let $y^H_1, \ldots, y^H_{\tn} \in [b]$ (resp.~$y_1, \ldots, y_{\tn} \in [d]\cup\{\bot\}$) be the selected hash (resp.~input) values and dummies. 
  Sample a random permutation $\pi$ over $[\tn]$\;
  \tcc{Send shuffled data (shuffler $\rightarrow$ data collector)}
  [s] Send $(E_\pkd[y^H_{\pi(1)}], \ldots, E_\pkd[y^H_{\pi(\tn)}])$ and $(E_\pkd[E_\pks[E_\pkd[y_{\pi(1)}]]], \ldots, E_\pkd[E_\pks[E_\pkd[y_{\pi(\tn)}]]])$ to the data collector\;
  \tcc{Filtering ($\Lambda^H \subseteq [b]$: selected hash values, $\Lambda \subseteq [d]$: selected items)}
  [d] Decrypt $y^H_{\pi(1)}, \ldots, y^H_{\pi(\tn)}$ and $E_\pks[E_\pkd[y_{\pi(1)}]], \ldots, E_\pks[E_\pkd[y_{\pi(\tn)}]]$\;
  [d] $(\tc^H_1, \ldots, \tc^H_b) 
  \leftarrow$\texttt{Count}$(y^H_{\pi(1)}, \ldots, y^H_{\pi(\tn)})$\;
  [d] $\Lambda^H, \Lambda \leftarrow$\texttt{FilterItems}$(\tc^H_1, \ldots, \tc^H_b, \alpha, \colorB{\calD_1}, l)$\;
  \tcc{Replace unselected items with $\bot$}
  \ForEach{$i \in \tn$}{
    \If{$y^H_{\pi(i)} \notin \Lambda^H$}{
        [d] $E_\pks[E_\pkd[y_{\pi(i)}]] \leftarrow E_\pks[E_\pkd[\bot]]$;
    }
  }
  \tcc{Send selected items and shuffled data (data collector $\rightarrow$ shuffler)}
  [d] Send $\Lambda$ and $(E_\pks[E_\pkd[y_{\pi(1)}]], \ldots, E_\pks[E_\pkd[y_{\pi(\tn)}]])$ to the shuffler\;
  \tcc{Dummy input data addition}
  [s] Decrypt $E_\pkd[y_{\pi(1)}], \ldots, E_\pkd[y_{\pi(\tn)}]$ 
  while removing those corresponding to dummy hash data\;
  \ForEach{$i \in \Lambda$}{
    [s] $z_i \leftarrow \colorB{\calD_2}$;
    Add a dummy $E_\pkd[i]$ for $z_i$ times\;
  }
  \tcc{Random shuffling}
  [s] Let $y^*_1, \ldots, y^*_{\tn^*} \in [d]\cup\{\bot\}$ be 
  $\tn$ 
  input values and 
  dummies. 
  Sample a random permutation $\rho$ over $[\tn^*]$\;
  \tcc{Send shuffled data (shuffler $\rightarrow$ data collector)}
  [s] Send $(E_\pkd[y^*_{\rho(1)}], \ldots, E_\pkd[y^*_{\rho(\tn^*)}])$ to the data collector\;
  \tcc{Compute an unbiased estimate}
  [d] Decrypt $y^*_{\rho(1)}, \ldots, y^*_{\rho(\tn^*)}$\;
  [d] $(\tc_1, \ldots, \tc_d) 
  \leftarrow$\texttt{Count}$(y^*_{\rho(1)}, \ldots, y^*_{\rho(\tn^*)})$\;
  \ForEach{$i \in \Lambda$}{
    [d] $\hf_i \leftarrow \frac{1}{n\beta}(\tc_i-\colorB{\mu_2})$\;
  }
  \KwRet{$\Lambda$ \rm{and} $\{\hf_i | i \in \Lambda\}$}
  \caption{Our FME protocol $\colorB{\calS_{\calD^*,\beta}^{\FME}}$. 
  $\pkd$ (resp.~$\pks$) represents a public key of the data collector (resp.~shuffler). 
  }\label{alg:S_FME}
\end{algorithm*}

\noindent{\textbf{Protocol.}}~~Algorithm~\ref{alg:S_FME} 
shows an algorithmic description of our FME protocol. 
\colorB{Our protocol uses two dummy-count distributions $\calD_1$ (mean: $\mu_1$, variance: $\sigma_1^2$) and $\calD_2$ (mean: $\mu_2$, variance: $\sigma_2^2$). 
$\calD_1$ and $\calD_2$ are used for adding dummy hash values and dummy input values, respectively. 
We denote these distributions by $\calD^* = (\calD_1, \calD_2)$ and our FME protocol by $\colorB{\calS_{\calD^*,\beta}^{\FME}}$. 
In Appendix~\ref{sec:details_FME}, we show a toy example of $\calS_{\calD^*,\beta}^{\FME}$.} 

First, each user $u_i$ sends her hash value $E_\pkd[h(x_i)]$ and input value $E_\pkd[E_\pks[E_\pkd[x_i]]]$, where $\pkd$ and $\pks$ are public keys of the data collector and the shuffler, respectively (lines 1-3). 
After receiving them, 
the shuffler randomly selects each pair with probability $\beta$ (line 4). 
Then, 
for each hash value $i\in[b]$, the shuffler adds a dummy $\la E_\pkd[i], E_\pkd[E_\pks[E_\pkd[\bot]]] \ra$ for $z_i \sim \colorB{\calD_1}$ times (lines 5-7). 
Let $y^H_1, \ldots, y^H_{\tn} \in [b]$ (resp.~$y_1, \allowbreak \ldots, y_{\tn} \in [d]\cup\{\bot\}$) be the selected hash (resp.~input) values and dummies, where $\tn \in \nnints$ is the total number of these values. 
The shuffler samples a random permutation $\pi$ over $[\tn]$ and sends $(E_\pkd[y^H_{\pi(1)}], \ldots, E_\pkd[y^H_{\pi(\tn)}])$ and $(E_\pkd[E_\pks[E_\pkd[y_{\pi(1)}]]], \ldots, E_\pkd[E_\pks[E_\pkd[y_{\pi(\tn)}]]])$ to the data collector (lines 8-9). 

Then, the data collector decrypts them and calls the \texttt{Count} function, which calculates a count $\tc_i^H \in \nnints$ for each hash value $i \in [b]$ from $y^H_{\pi(1)}, \ldots, y^H_{\pi(\tn)}$ (lines 10-11). 
Based on $\tc_1^H, \ldots, \tc_b^H$, the data collector calls the \texttt{FilterItems} function, which selects a set $\Lambda^H \subseteq [b]$ of popular hash values and 
a set $\Lambda \subseteq [d]$ of the corresponding input values 
(line 12). 
We explain the \texttt{FilterItems} function later in detail. 
Then, the data collector replaces unselected items $E_\pks[E_\pkd[y_{\pi(i)}]]$, whose corresponding hash values are $y^H_{\pi(i)} \notin \Lambda^H$, with $E_\pks[E_\pkd[\bot]]$ (lines 13-17). 
Note that the data collector can do this without seeing $y_{\pi(i)}$, as she knows $y^H_{\pi(i)}$. 
The data collector sends 
$\Lambda$ and shuffled input values $(E_\pks[E_\pkd[y_{\pi(1)}]], \allowbreak \ldots, E_\pks[E_\pkd[y_{\pi(\tn)}]])$ back to the shuffler (line 18).

The shuffler decrypts 
$\tn$ 
shuffled input values 
while removing those corresponding to dummy hashes, i.e., $E_\pkd[E_\pks[E_\pkd[\bot]]]$ generated by the shuffler (line 19). 
The shuffler can remove them because she knows the random permutation $\pi$. 
Then, for each selected item $i \in \Lambda$, it adds a dummy $E_\pkd[i]$ for $z_i \sim \colorB{\calD_2}$ times (lines 20-22). 
Let $y^*_1, \ldots, y^*_{\tn^*} \in [d]$ be $\tn$ input values and dummies, where $\tn^* \in \nnints$ is the total number of these values. 
The shuffler samples a random permutation $\rho$ over $[\tn^*]$ and sends 
$(E_\pkd[y^*_{\rho(1)}], \ldots, E_\pkd[y^*_{\rho(\tn^*)}])$ to the data collector (lines 23-24). 
Finally, the data collector 
calculates an unbiased estimate $\hf_i$ of $f_i$ for $i \in \Lambda$ 
and outputs 
$\Lambda$ and 
$\{\hf_i | i \in \Lambda\}$ 
(lines 25-30). 

\smallskip{}
\noindent{\textbf{Filtering Items.}}~~Below, we explain the details of the \texttt{FilterItems} function in Algorithm~\ref{alg:S_FME} (line 12). 
This function takes two parameters as input: a significance level $\alpha \in [0,1]$ and the maximum number $l \in [b]$ of selected hashes. 
Specifically, 
we first calculate a threshold $z_{th} \in \nnints$ so that a random variable $z$ generated from $\colorB{\calD_1}$ is larger than or equal to $z_{th}$ with probability at most 
$\alpha$, 
i.e., $\Pr(z \geq z_{th}) \leq \alpha$. 
Then, we compare each count $\tc^H_i$ ($i \in [b]$) to the threshold $z_{th}$ and select the hash value $i$ if $\tc^H_i \geq z_{th}$. 
If the number of selected hash values exceeds 
$l$, 
we select $l$ hash values with the largest counts. 
Finally, we add the selected hash values to $\Lambda^H$ and the corresponding input values $x$ such that $h(x) \in \Lambda^H$ to $\Lambda$. 
Fig.~\ref{fig:filtering} shows an example of $\Lambda^H$. 

\begin{figure}[t]
  \centering
  \includegraphics[width=0.99\linewidth]{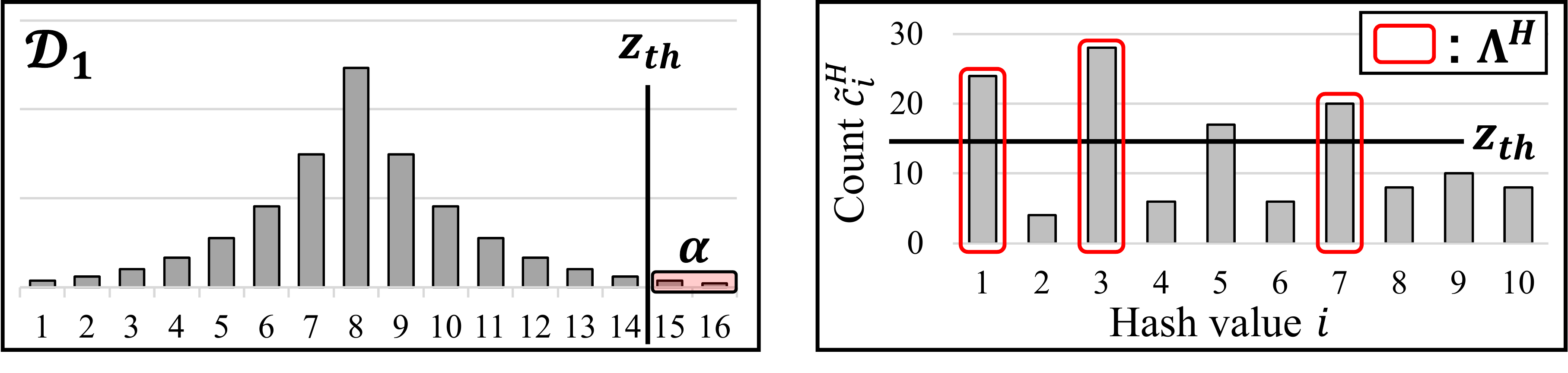}
  \vspace{-8mm}
  \caption{Example of filtering items. In this example, $\alpha=0.05$, $z_{th}=15$, $b=10$, $l=3$, and $\Lambda^H=\{1,3,7\}$.} 
  \label{fig:filtering}
\end{figure}

Suppose that no user has a hash value $i \in [b]$. 
Then, we incorrectly add $i$ to $\Lambda^H$ with probability at most $\alpha$; 
i.e., the false positive probability is at most the significance level $\alpha$, which is small. 
In addition, we add at most $l$ hash values to $\Lambda^H$. 
Thus, we can improve the efficiency by reducing $l$. 

\smallskip{}
\noindent{\textbf{\colorB{Parameters.}}}~~\colorB{$\calS_{\calD^*,\beta}^{\FME}$ has parameters $\beta$, $\alpha$, and $l$. 
We explain how to set $l$ in Section~\ref{sub:FME_optimization} and $\beta$  and $\alpha$ in Appendix~\ref{sub:FME_parameters}.}

\subsection{Theoretical Properties}
\label{sub:FME_properties}

\smallskip{}
\noindent{\textbf{Privacy and Robustness.}}~~First, we show that $\colorB{\calS_{\calD^*,\beta}^{\FME}}$ provides DP and is robust to both data poisoning and collusion attacks: 

\begin{theorem}
\label{thm:FME_DP}
\colorB{If the binary input mechanisms $\calM_{\calD_1, \beta}$ and 
$\calM_{\calD_2, 1}$ 
in Definition~\ref{def:binary_input} provide $(\frac{\epsilon_1}{2}, \frac{\delta_1}{2})$-DP and $(\frac{\epsilon_2}{2}, \frac{\delta_2}{2})$-DP, respectively, then $\calS_{\calD^*,\beta}^{\FME}$ provides 
$(\epsilon, \delta)$-CDP, where $(\epsilon, \delta) = (\epsilon_1 + \epsilon_2, \delta_1 + \delta_2)$, 
and is robust to collusion with users.} 
\end{theorem}

\begin{theorem}
\label{thm:FME_poisoning}
Let $\lambda = \frac{n'}{n+n'}$ and $f_\calT = \sum_{i \in \calT} f_i$. 
$\colorB{\calS_{\calD^*,\beta}^{\FME}}$ provides the following robustness 
against poisoning attacks:
\begin{align}
\GMGA = \textstyle{\lambda (1 - f_\calT) + \sum_{i\in\calT} \eta_i f_i,} 
\label{eq:FME_GMGA}
\end{align}
where 
$\eta_i \in [0,1]$ is the probability that the $i$-th item is not selected in the filtering step (i.e., $i \notin \Lambda$) before data poisoning. 
\end{theorem}

\colorB{Theorem~\ref{thm:FME_DP} uses the basic composition theorem \cite{DP}. 
In Appendix~\ref{sec:details_FME}, we explain that this is almost tight in our case.
$\calS_{\calD^*,\beta}^{\FME}$ spends $(\epsilon_1,\delta_1)$ and $(\epsilon_2,\delta_2)$ for hash and input values, respectively. 
When $\beta = 1$ and the same dummy-count distribution is used for $\calD_1$, $\calD_2$, and $\calD$, 
$\calS_{\calD^*,\beta}^{\FME}$ needs the total budget $(\epsilon,\delta)$ 
twice as large as the LNF protocol $\calS_{\calD, \beta}^{\LNF}$.}

In addition, 
Theorem~\ref{thm:FME_DP} assumes 
that 
the shuffler does not disclose dummies after running the protocol in the same way as \cite{Wang_PVLDB20,Murakami_SP25}. 
\colorB{This can be achieved, e.g., by using the TEE.} 
We note, however, that 
even if the shuffler discloses dummies, we can guarantee DP for the outputs by adding additional noise to the outputs before publishing them. 
In Appendix~\ref{sub:proposals_additional}, we show the accuracy is hardly affected by this additional noise. 

In Theorem~\ref{thm:FME_poisoning}, the second term in (\ref{eq:FME_GMGA}) is introduced because a target item $i \in \calT$ can be discarded in the filtering step before poisoning and selected after poisoning. 
This results in a slight increase in the overall gain. 
We note, however, that the second term in (\ref{eq:FME_GMGA}) is very small in practice, as $\eta_i$ is extremely small when $f_i$ is large (e.g., $\eta_i$ decreases exponentially as $f_i$ increases when $l=b$; 
\conference{see~\cite{Murakami_arXiv25}}\arxiv{see Appendix~\ref{sub:proof_FME_poisoning}}). 
When the second term is ignored, 
$\GMGA$ of $\colorB{\calS_{\calD^*,\beta}^{\FME}}$ 
is equal to that of 
$\calS_{\calD, \beta}^{\LNF}$ in (\ref{LNF_GMGA}). 

\smallskip{}
\noindent{\textbf{Efficiency.}}~~We next analyze the communication cost of $\colorB{\calS_{\calD^*,\beta}^{\FME}}$: 

\begin{theorem}
\label{thm:FME_communication}
Let $\tau_1, \tau_2, \tau_3 \in \nnreals$ be the size of a single ciphertext, double ciphertext, and triple ciphertext, respectively, in $\colorB{\calS_{\calD^*,\beta}^{\FME}}$. 
Then,  the total communication cost of $\colorB{\calS_{\calD^*,\beta}^{\FME}}$ is 
$C_{tot} = C_{U-S} + C_{S-D}$, where 
\begin{align*}
C_{U-S} &= (\tau_1 + \tau_3) n \\
C_{S-D} &\leq \textstyle{(2 \tau_1 + \tau_2 + \tau_3) (\beta n + \colorB{\mu_1} b) + \tau_1 (\colorB{\mu_2} + 1)\E[|\Lambda|]}, 
\end{align*}
and 
the expected number $\E[|\Lambda|]$ of selected items is: 
\begin{align}
\E[|\Lambda|] \leq 
\begin{cases}
\textstyle{\frac{(\beta n + \alpha (l - \beta n))d}{b}} 
   &   \text{(if $\beta n \leq l \leq b$)}\\
\textstyle{\frac{ld}{b}}   &   \text{(otherwise)}. 
\end{cases}
\label{eq:FME_E_Lambda}
\end{align}
\end{theorem}
For example, if we use ECIES with $256$-bit security~\cite{bouncy}, then $\tau_1$, $\tau_2$, and $\tau_3$ are $712$, $1392$, and $2072$ bits, respectively. 
By optimizing $b$, $C_{tot}$ of $\colorB{\calS_{\calD^*,\beta}^{\FME}}$ can be expressed as 
$C_{tot} = O(n + \sqrt{ld})$ when $l < \beta n$. 
See Section~\ref{sub:FME_optimization} for details. 
Similarly, the computational cost of $\colorB{\calS_{\calD^*,\beta}^{\FME}}$ can be 
$O(n + \sqrt{ld})$ 
in this case. 
\conference{See~\cite{Murakami_arXiv25}}\arxiv{See Appendix~\ref{sec:FME_computational_cost}} for details.

\smallskip{}
\noindent{\textbf{Accuracy.}}~~Finally, we analyze the accuracy of $\colorB{\calS_{\calD^*,\beta}^{\FME}}$: 

\begin{theorem}
\label{thm:FME_accuracy}
For any item $i \in \Lambda$ selected in the filtering step,  
$\colorB{\calS_{\calD^*,\beta}^{\FME}}$ outputs an unbiased estimate 
(i.e., $\E[\hf_i | \Lambda] = f_i$) 
and achieves the following variance: 
\begin{align}
\V[\hf_i | \Lambda] = \textstyle{\frac{f_i (1-\beta)}{n \beta} + \frac{\colorB{\sigma_2^2}}{n^2 \beta^2}}. 
\label{eq:FME_l2-loss}
\end{align}
$\colorB{\calS_{\calD^*,\beta}^{\FME}}$ also achieves the following expected squared error: 
\begin{align}
\textstyle{\E[(\hf_i - f_i)^2] 
= (1 - \eta_i)\V[\hf_i | \Lambda] + \eta_i f_i^2 }, 
\label{eq:FME_l2-loss_all}
\end{align}
where 
$\eta_i \in [0,1]$ is the probability that the $i$-th item is not selected in the filtering step (i.e., $i \notin \Lambda$). 
\end{theorem}
The variance in (\ref{eq:FME_l2-loss}) is the same as that of $\calS_{\calD, \beta}^{\LNF}$ in (\ref{eq:LNF_l2_loss}). 
The second term in (\ref{eq:FME_l2-loss_all}) is introduced because $\colorB{\calS_{\calD^*,\beta}^{\FME}}$ always calculates $\hf_i$ as $0$ for an unselected item $i \notin \Lambda$. 
However, the second term in (\ref{eq:FME_l2-loss_all}) is very small, as $\eta_i$ is extremely small for a large $f_i$, as explained above. 
When ignoring the second term, the $l_2$ loss of $\colorB{\calS_{\calD^*,\beta}^{\FME}}$ is almost the same as that of $\calS_{\calD, \beta}^{\LNF}$.

\smallskip{}
\noindent{\textbf{Summary.}}~~Our 
FME protocol $\colorB{\calS_{\calD^*,\beta}^{\FME}}$ can achieve almost the same accuracy and robustness as the LNF protocol $\calS_{\calD, \beta}^{\LNF}$ by using $(\epsilon,\delta)$ twice as large as $\calS_{\calD, \beta}^{\LNF}$. 
It can also achieve the communication and computational costs of $O(n + \sqrt{ld})$ by optimizing $b$, as explained below. 

\subsection{Optimizing the Range $b$ of the Hash Function}
\label{sub:FME_optimization}

As shown in Theorem~\ref{thm:FME_communication}, the communication cost $C_{tot}$ of $\colorB{\calS_{\calD^*,\beta}^{\FME}}$ depends on the 
hash range 
$b$. 
A larger $b$ results in the increase of dummy values sent from the shuffler to the data collector. 
In contrast, a smaller $b$ results in the increase of 
$\E[|\Lambda|]$ 
in (\ref{eq:FME_E_Lambda}). 
The optimal 
$b$ can be obtained by calculating $b$ that minimizes the upper bound on $C_{tot}$ in Theorem~\ref{thm:FME_communication}. 

For example, if $l = b$, then the optimal value of $b$ is given by 
$b = \sqrt{\frac{\tau_1 (\colorB{\mu_2} + 1) \beta(1 - \alpha)nd}{(2\tau_1 + \tau_2 + \tau_3) \colorB{\mu_1}}}$. 
In this case, 
$C_{tot}$ of $\colorB{\calS_{\calD^*,\beta}^{\FME}}$ can be simplified as $C_{tot} = O(n + \sqrt{nd} + \alpha d)$ by treating $\tau_1$, $\tau_2$, and $\tau_3$ as constants. 
This is larger than $O(n)$ but much smaller than $O(n + d)$ when $d \gg n$. 

If 
$l < \beta n$, 
the optimal value of $b$ 
is 
$b = \sqrt{\frac{\tau_1 (\colorB{\mu_2} + 1) ld}{(2\tau_1 + \tau_2 + \tau_3) \colorB{\mu_1}}}$. 
In this case, 
$C_{tot}$ of $\colorB{\calS_{\calD^*,\beta}^{\FME}}$ 
can be written as 
$C_{tot} = O(n + \sqrt{ld})$. 

\smallskip{}
\noindent{\textbf{Setting $l$.}}~~In practice, the frequency distribution $\bmf$ is sparse in large-domain data, and we are often interested in popular items with large frequencies \cite{Gu_USENIX20,Luo_CCS22,Zheng_WWW09}. 
Thus, in this work, we propose to set $l$ to $l = \max\{\frac{n^2}{d},c\}$, where $c \in \nats$ is some constant ($c=50$ in our experiments). 
In this setting, $C_{tot}$ does not depend on $d$ until $d$ exceeds $\frac{n^2}{c}$. 
Moreover, this setting guarantees that at least $c$ popular items are selected. 
In our experiments, we show that our FME protocol with this setting provides high accuracy and efficiency.

\section{Application to Key-Value Data}
\label{sec:key-value}
In this section, we propose a protocol for frequency and mean estimation over KV data by applying the FME protocol $\colorB{\calS_{\calD^*,\beta}^{\FME}}$ with an additional technique called \textit{TKV-FK (Transforming KV Pairs and Filtering Keys)}. 
Section~\ref{sub:KV_protocol} explains our protocol. 
Section~\ref{sub:KV_properties} shows its theoretical properties. 

\subsection{Our Protocol for KV Data}
\label{sub:KV_protocol}

\noindent{\textbf{Overview.}}~~Our protocol for KV data, denoted by $\colorB{\calS_{\calD^*,\beta}^{\KV}}$, 
first uses padding-and-sampling \cite{Gu_USENIX20}, a state-of-the-art sampling technique for a large domain. 
Then, it applies the FME protocol $\colorB{\calS_{\calD^*,\beta}^{\FME}}$ to KV data with 
an additional technique called TKV-FK, which transforms KV pairs into one-dimensional data and filters the data at a key level. 
Finally, it calculates unbiased estimates $\bmhPhi$ and $\bmhPsi$ of frequencies $\bmPhi$ and mean values $\bmPsi$, respectively. 
Below, we explain these techniques in detail. 

\smallskip{}
\noindent{\textbf{Padding-and-Sampling.}}~~We first use 
a padding-and-sampling technique \cite{Gu_USENIX20}. 
This technique samples a single KV pair for each user $u_i$ to avoid splitting the privacy budget $\epsilon$ into multiple KV pairs $x_i$. 
Specifically, 
if the number $|x_i|$ of KV pairs is smaller than a parameter $\kappa \in \nats$ called the padding length, user $u_i$ adds dummy KV pairs $\langle d+1, 0 \rangle, \ldots, \langle d+\kappa-|x_i|, 0 \rangle$ to $x_i$. 
Then, $u_i$ samples a single KV pair $\langle k_i, v_i^* \rangle$ from $x_i$ and discretizes $v_i^*$ to $v_i = 1$ with probability $\frac{1+v_i^*}{2}$ and $v_i = -1$ with probability $\frac{1-v_i^*}{2}$. 
As a result, $u_i$ obtains a single, discretized KV pair $\langle k_i, v_i \rangle \in [d] \times \{-1,1\}$. 

\smallskip{}
\noindent{\textbf{TKV-FK.}}~~After padding-and-sampling, we apply the FME protocol $\colorB{\calS_{\calD^*,\beta}^{\FME}}$ to the KV pair $\langle k_i, v_i \rangle$ ($i \in [n]$) with our TKV-FK technique. 
This technique first transforms the KV pair $\langle k_i, v_i \rangle$ of user $u_i$ into one-dimensional data $s_i \in [2d]$ by $s_i = k_i + \frac{1}{2}(v_i+1)d$. 
Then, we can directly apply the FME protocol $\colorB{\calS_{\calD^*,\beta}^{\FME}}$ by treating $(s_1,\ldots, s_n) \in [2d]^n$ as input values. 
However, this approach does not work well, as filtering is performed \textit{at a key-value pair level}. 
Fig.~\ref{fig:TKV-FK}(a) shows its example. 
In this example, the KV pair $\langle 2, 1 \rangle$ is selected in the filtering step, whereas $\langle 2, -1 \rangle$ is not. 
This causes a large positive bias in the estimate $\hPsi_2$ of the mean value $\Psi_2$ for key $2$. 
Thus, mean values $\bmPsi$ cannot be accurately estimated. 

\begin{figure}[t]
  \centering
  \includegraphics[width=0.99\linewidth]{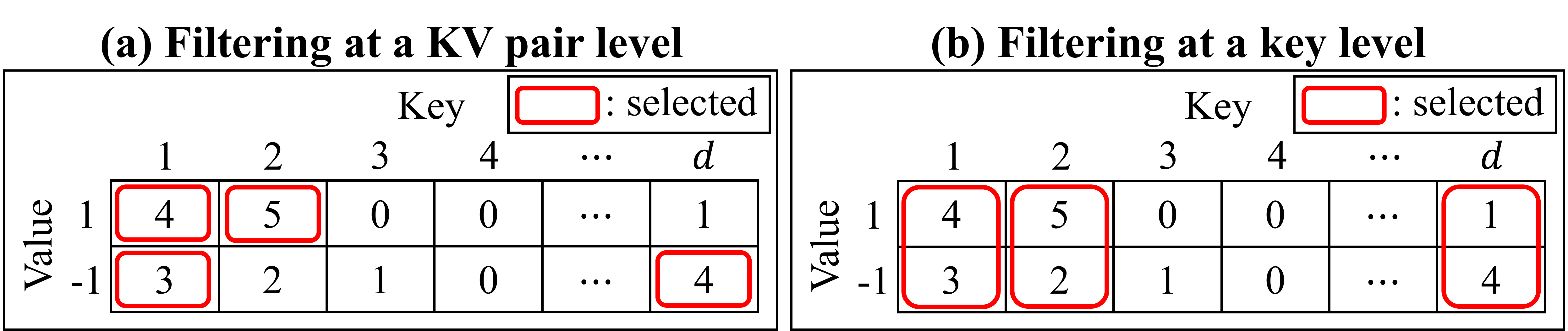}
  \vspace{-7mm}
  \caption{Example of filtering at a (a) KV pair level or (b) key level. 
  The number in the matrix represents a count of each KV pair in the first shuffled data sent from the shuffler to the data collector.} 
  \label{fig:TKV-FK}
\end{figure}

Our TKV-FK technique addresses this issue by filtering data \textit{at a key level}, as shown in Fig.~\ref{fig:TKV-FK}(b). 
Specifically, 
it makes the following three changes to the FME protocol $\colorB{\calS_{\calD^*,\beta}^{\FME}}$ (Algorithm~\ref{alg:S_FME}): 
(i) each user $u_i$ applies a hash function $h: [d] \rightarrow [b]$ to her \textit{key} $k_i$ 
and sends $E_\pkd[h(k_i)]$ and $E_\pkd[E_\pks[E_\pkd[s_i]]]$ to the shuffler (line 2), 
(ii) based on shuffled hash values, the data collector selects a set $\Lambda^H \subseteq [b]$ of popular hash values and 
a set $\Lambda \subseteq [d]$ of the corresponding \textit{keys} 
(line 12), and 
(iii) for each $i \in \Lambda$, the data collector adds $z_{i,1}, z_{i,-1} \sim \colorB{\calD_2}$ dummies to KV pairs $\langle i, 1 \rangle$ and $\langle i, -1 \rangle$, respectively  (line 21). 

\smallskip{}
\noindent{\textbf{Calculating Unbiased Estimates.}}~~Finally, we calculate unbiased estimates $\bmhPhi$ and $\bmhPsi$ as follows. 
For $i \in \Lambda$, let $\tc_{i,1}$ (resp.~$\tc_{i,-1}$) $\in \nnints$ be a count of KV pair $\langle i, 1 \rangle$ (resp.~$\langle i, -1 \rangle$) in the shuffled data the data collector decrypts. 
For $i \in \Lambda$, the data collector calculates $\hPhi_i$ and $\hPsi_i$ as follows\footnote{Note that $\hPsi_i$ is not defined for an unselected key $i \notin \Lambda$. 
We set $\hPsi_i = 1$ for $i \notin \Lambda$ to eliminate the (unnecessary) mean gain for unselected keys.}: 
\begin{align}
\textstyle{\hPhi_i = \frac{\kappa}{n\beta}(\tc_{i,1} + \tc_{i,-1} - 2\colorB{\mu_2}), \hPsi_i = \frac{\kappa}{n\beta \hPhi_i} (\tc_{i,1} - \tc_{i,-1}).} 
\label{eq:KV_proposal_estimates}
\end{align}
Thanks to our TKV-FK technique, we can calculate both $\tc_{i,1}$ and $\tc_{i,-1}$ for each selected key $i \in \Lambda$. 
Thus, the estimates in (\ref{eq:KV_proposal_estimates}) are (almost) unbiased, as shown in Section~\ref{sub:KV_properties}.

\subsection{Theoretical Properties}
\label{sub:KV_properties}

\noindent{\textbf{Privacy and Robustness.}}~~Our KV protocol $\colorB{\calS_{\calD^*,\beta}^{\KV}}$ also provides DP and is robust to collusion and poisoning attacks: 

\begin{theorem}
\label{thm:KV_DP}
\colorB{If the binary input mechanisms $\calM_{\calD_1, \beta}$ and 
$\calM_{\calD_2, 1}$ 
in Definition~\ref{def:binary_input} provide $(\frac{\epsilon_1}{2}, \frac{\delta_1}{2})$-DP and $(\frac{\epsilon_2}{2}, \frac{\delta_2}{2})$-DP, respectively, then $\calS_{\calD^*,\beta}^{\KV}$ provides $(\epsilon, \delta)$-CDP, where $(\epsilon, \delta) = (\epsilon_1 + \epsilon_2, \delta_1 + \delta_2)$, 
and is robust to collusion with users.} 
\end{theorem}

\begin{theorem}
\label{thm:KV_poisoning}
Let $\lambda = \frac{n'}{n+n'}$, $\Phi_\calT = \sum_{i \in \calT} \Phi_i$, and $\Psi_\calT = \sum_{i \in \calT} \Psi_i$. 
Let $\calU_i$ be the set of users who have key $i \in [d]$. 
For $u_j \in \calU_i$, let $\psi_{j,i} \in [-1,1]$ be the value of key $i$ held by user $u_j$. 
For $j \in [n]$, let $\xi_j = \max\{|x_j|, \kappa\}$. 
$\colorB{\calS_{\calD^*, \beta}^{\KV}}$ provides the following robustness 
against poisoning attacks:
\begin{align}
\GMTGAPhi 
&= \textstyle{\frac{\kappa}{n+n'} \left( (\sum_{i \in \calT} \sum_{u_j \in \calU_i} \frac{1}{\xi_j}) + n' \right) - \Phi_\calT} \nonumber\\
& \hspace{4.5mm} \textstyle{+ \sum_{i \in \calT} \eta_i \Phi_i} 
\label{eq:KV_GMTGAPhi_general}\\
\GMTGAPsi 
&\approx 
\left(\sum_{i \in \calT} \frac{(\sum_{u_j \in \calU_i} \frac{\psi_{j,i}}{\xi_j}) + \frac{n'}{|\calT|}}{(\sum_{u_j \in \calU_i} \frac{1}{\xi_j}) + \frac{n'}{|\calT|}} \right) 
- \Psi_\calT, 
\label{eq:KV_GMTGAPsi_general}
\end{align}
where 
$\eta_i \in [0,1]$ is the probability that the $i$-th item is not selected in the filtering step (i.e., $i \notin \Lambda$) before data poisoning. 
The approximation in (\ref{eq:KV_GMTGAPsi_general}) is obtained from a Taylor expansion $\E[\frac{X}{Y}] \approx \frac{\E[X]}{\E[Y]}$ for two random variables $X$ and $Y$. 
\end{theorem}

Theorem~\ref{thm:KV_poisoning} states that 
$\GMTGAPhi$ and $\GMTGAPsi$ of $\colorB{\calS_{\calD^*,\beta}^{\KV}}$ in (\ref{eq:KV_GMTGAPhi_general}) and (\ref{eq:KV_GMTGAPsi_general}) do 
not 
depend on the privacy budget $\epsilon$. 
It is shown in \cite{Wu_USENIX22} that the existing KV protocols 
become vulnerable to data poisoning as $\epsilon$ increases or decreases. 
For example, in PrivKVM \cite{Ye_SP19}, 
$\GMTGAPhi$ increases (resp.~decreases) as $\epsilon$ decreases when $|\calT| = 1$ (resp.~$|\calT| \geq 3$). 
In PCKV-GRR and PCKV-UE, $\GMTGAPhi$ increases as $\epsilon$ decreases. 
In contrast, $\colorB{\calS_{\calD^*, \beta}^{\KV}}$ does not suffer from such fluctuation in $\GMTGAPhi$ and $\GMTGAPsi$. 

Moreover, $\GMTGAPhi$ and $\GMTGAPsi$ of $\colorB{\calS_{\calD^*,\beta}^{\KV}}$ 
are much smaller than those of the existing KV protocols analyzed in~\cite{Wu_USENIX22}. 
\conference{See~\cite{Murakami_arXiv25}}\arxiv{See Appendix~\ref{sub:proof_thm_KV_poisoning}} for details. 
In our experiments, we show that $\colorB{\calS_{\calD^*,\beta}^{\KV}}$ is much more robust than the existing protocols. 

\smallskip{}
\noindent{\textbf{Efficiency.}}~~Our KV protocol $\colorB{\calS_{\calD^*,\beta}^{\KV}}$ achieves the same efficiency 
as $\colorB{\calS_{\calD^*,\beta}^{\FME}}$. 
Specifically, $\colorB{\calS_{\calD^*,\beta}^{\KV}}$ can achieve the communication and computational costs of 
$C_{tot} = O(n + \sqrt{ld})$. 

\smallskip{}
\noindent{\textbf{Accuracy.}}~~Finally, we 
show 
the accuracy of $\colorB{\calS_{\calD^*, \beta}^{\KV}}$: 
\begin{theorem}
\label{thm:KV_accuracy}
If $\kappa \geq |x_j|$ for any $j\in[n]$, then for any key $i \in \Lambda$ selected in the filtering step,  
$\colorB{\calS_{\calD^*, \beta}^{\KV}}$ outputs almost unbiased estimates. 
Specifically, $\colorB{\calS_{\calD^*, \beta}^{\KV}}$ achieves: 
\begin{align}
\E[\hPhi_i | \Lambda] &= \Phi_i \label{eq:KV_Phi_bias} \\
\V[\hPhi_i | \Lambda] &= \textstyle{\frac{\Phi_i(\kappa - \beta)}{n \beta} + \frac{2 \kappa^2 \colorB{\sigma_2^2}}{n^2 \beta^2}} \label{eq:KV_Phi_variance} \\
\E[\hPsi_i | \Lambda] &\approx \Psi_i \label{eq:KV_Psi_bias} \\
\V[\hPsi_i | \Lambda] &\lesssim \textstyle{\frac{\kappa^2}{n \beta^2} ( 2(q_{i} - q_{i}^2 + r_{i} - r_{i}^2)- \frac{\beta}{\kappa}(1-\frac{\beta}{\kappa}))}, \label{eq:KV_Psi_variance}
\end{align}
where $q_{i} = \frac{\beta(1 + \Psi_i)}{2 \kappa}$ and $r_{i} = \frac{\beta (1 - \Psi_i)}{2 \kappa}$. 
The approximations in (\ref{eq:KV_Psi_bias}) and (\ref{eq:KV_Psi_variance}) are obtained from Taylor expansions $\E[\frac{X}{Y}] \approx \frac{\E[X]}{\E[Y]}$ and $\V[\frac{X}{Y}] \approx \frac{\V[X]}{\E[Y]}$ for two random variables $X$ and $Y$. 
In addition, $\colorB{\calS_{\calD^*, \beta}^{\KV}}$ achieves the following expected $l_2$ loss: 
\begin{align*}
\textstyle{\E[(\hPhi_i - \Phi_i)^2]} 
&= \textstyle{(1 - \eta_i)\V[\hPhi_i | \Lambda] + \eta_i \Phi_i^2 } \\
\textstyle{\E[(\hPsi_i - \Psi_i)^2]} 
&= \textstyle{(1 - \eta_i)\V[\hPsi_i | \Lambda] + \eta_i (1 - \Psi_i)^2 }, 
\end{align*}
where 
$\eta_i \in [0,1]$ is the probability that the $i$-th key is not selected in the filtering step (i.e., $i \notin \Lambda$). 
\end{theorem}
In our experiments, we show that $\colorB{\calS_{\calD^*, \beta}^{\KV}}$ provides higher accuracy than the existing KV protocols. 

\section{Experimental Evaluation}
\label{sec:exp}

\subsection{Experimental Set-up}
\label{sub:setup}
\noindent{\textbf{Datasets.}}~~We conducted experiments using four real datasets: 
\begin{itemize}
\item \textbf{Foursquare~\cite{Yang_TIST16}}: Location 
dataset with $n=18201$ users in New York. 
We divided the city into $1000 \times 1000$ regions ($d=1000000$) at regular intervals. 
\item \textbf{AOL~\cite{Pass_InfoScale06}}: Web access dataset. 
Following \cite{Luo_CCS22}, we extracted $n=10000$ users and used the first three characters of each URL as an item 
($d=2^{24}=16777216$). 
\item \textbf{E-Commerce~\cite{ecommerce}}: Clothing review dataset with $n=23486$ users, $d=1206$ keys (items), and $23486$ ratings. 
\item \textbf{Amazon~\cite{amazon_dataset}}: Amazon rating dataset with $n=1210271$ users, $d=249274$ keys, and $2023070$ ratings. 
\end{itemize}
The first two are categorical, and the last two are KV datasets. 

\smallskip{}
\noindent{\textbf{Protocols.}}~~We compared our FME protocol with the CH protocol 
and existing shuffle protocols. 
In our FME protocol, we set $\alpha = 0.05$ and used the asymmetric geometric distribution \cite{Murakami_SP25} with $\beta=1$ as \colorB{dummy-count distributions $\calD_1$ and $\calD_2$ (i.e., $(\epsilon_1,\delta_1)=(\epsilon_2,\delta_2)$)}. 
Then, we optimized the hash range $b$ as described in Section~\ref{sub:FME_optimization} and set the maximum number $l$ of selected hashes to $l = b$ (denoted by \ProposalL{}) or $\max\{\frac{n^2}{d},50\}$ (\ProposalS{}). 
For KV data, we used our TKV-FK technique. 
\colorB{In Appendix~\ref{sub:FME_parameters}, we also evaluated our FME protocol when we changed $\epsilon_1$ $(= \epsilon - \epsilon_2)$, $\alpha$, and $l$.} 

For existing protocols, we evaluated twelve protocols. 
Specifically, for categorical data, we evaluated four pure shuffle protocols using the GRR~\cite{Wang_PVLDB20}, RAPPOR~\cite{Erlingsson_CCS14}, OUE~\cite{Wang_USENIX17}, and OLH~\cite{Wang_USENIX17}. 
For these protocols, we used a numerical upper bound in \cite{Feldman_FOCS21} because it is tighter than Theorem~\ref{thm:privacy_amplification} (we confirmed that the upper bound is very close to the lower bound in \cite{Feldman_FOCS21}). 
We also evaluated three multi-message protocols in \cite{Luo_CCS22} (for a large domain), \cite{Balcer_ITC20,Cheu_SP22}, denoted by \LWYL{}, \BC{}, and \CM{}, respectively. 
We used their amplification results for these protocols. 
Since these 
protocols 
assume that $\epsilon$ is within a certain range, 
we evaluated them only within the range. 
For \CM{}, we generated $10$ dummy values per user in the same way as \cite{Cheu_SP22}. 
We also compared our protocols with the LNF protocol \cite{Murakami_SP25} in terms of efficiency. 
Following \cite{Erlingsson_CCS14,Wang_USENIX17}, we used a significance threshold, which assigns $0$ to an estimate below a threshold for each protocol. 

For KV data, we evaluated four pure shuffle protocols based on PrivKVM \cite{Ye_SP19}, PrivKVM* \cite{Ye_TDSC23}, PCKV-GRR \cite{Gu_USENIX20}, and PCKV-UE \cite{Gu_USENIX20} using the numerical upper bound in \cite{Feldman_FOCS21}.
We did not evaluate the protocol in \cite{Zhu_IS23}, as it leaks the number of KV pairs held by each user and fails to provide DP. 
Following~\cite{Gu_USENIX20}, we clipped frequency estimates to 
$[0,1]$ and set the padding length $\kappa$ to $\kappa=1$ (resp.~$3$) in the E-Commerce (resp.~Amazon) dataset. 

\smallskip{}
\noindent{\textbf{Performance Metrics.}}~~Since $d$ is large in our experiments, most items have low or zero frequencies. 
Thus, we evaluated the accuracy for top-$50$ items (keys) with the largest frequencies. 
Specifically, we evaluated the MSE (Mean Squared Error) over the $50$ items. 
\colorB{Here, we varied $\epsilon$ from $0.1$ to $5$, as DP with this range of $\epsilon$ provides theoretical privacy guarantees against the inference of input values. See Appendix~\ref{sec:DP_epsilon} for details.}

For robustness to collusion attacks, 
we refer to $\epsilon$ when no (resp.~$|\Omega|$) users collude with the data collector as a \textit{target $\epsilon$} (resp.~\textit{actual $\epsilon$}). 
We set the target $\epsilon$ to $0.1$ and evaluated the actual $\epsilon$ while changing $|\Omega|$. 
For robustness to poisoning attacks, we evaluated the maximum gains $\GMGA$, $\GMTGAPhi$, and $\GMTGAPsi$. 
We ran each protocol $10$ times and averaged the MSE and 
the gains. 
For efficiency, we evaluated $C_{tot}$ and measured the run time using a workstation with Intel Xeon W-2295 (3.00 GHz, 18 Cores) and 256 GB main memory. 

\smallskip{}
\noindent{\textbf{User Sampling.}}~~In the KV datasets, $n$ is large. 
Thus, for each protocol (except for the LNF protocol), we introduced user sampling, which randomly samples users with probability $0.05$ before running the protocol to improve efficiency. 
We evaluated the MSE between the estimates and the true frequencies \textit{before} user sampling. 

\subsection{Experimental Results}
\label{sub:results}

\noindent{\textbf{Accuracy.}}~~First, we evaluated the relationship between the MSE and $\epsilon$. 
Fig.~\ref{fig:res_MSE} shows the results. 

Fig.~\ref{fig:res_MSE} shows that \ProposalL{} significantly outperforms the existing protocols. 
This is because our FME protocol with $l=b$ achieves almost the same accuracy as the LNF protocol, which is shown to be very accurate in \cite{Murakami_SP25}, by doubling $(\epsilon,\delta)$. 
Fig.~\ref{fig:res_MSE} also shows that \ProposalS{} provides almost the same MSE as \ProposalL{}, which means that we can improve the efficiency without affecting accuracy by reducing $l$ as proposed in Section~\ref{sub:FME_optimization}. 

In Appendix~\ref{sub:comparison_other_baselines}, we also show that the CH protocol has poor accuracy for unpopular items due to hash collision and cannot be improved by using user/group-dependent hashes. 
Moreover, in Appendix~\ref{sub:effectiveness_TKV_FK}, we show the effectiveness of our TKV-FK technique through an ablation study. 

\begin{figure}[t]
  \centering
  \includegraphics[width=0.99\linewidth]{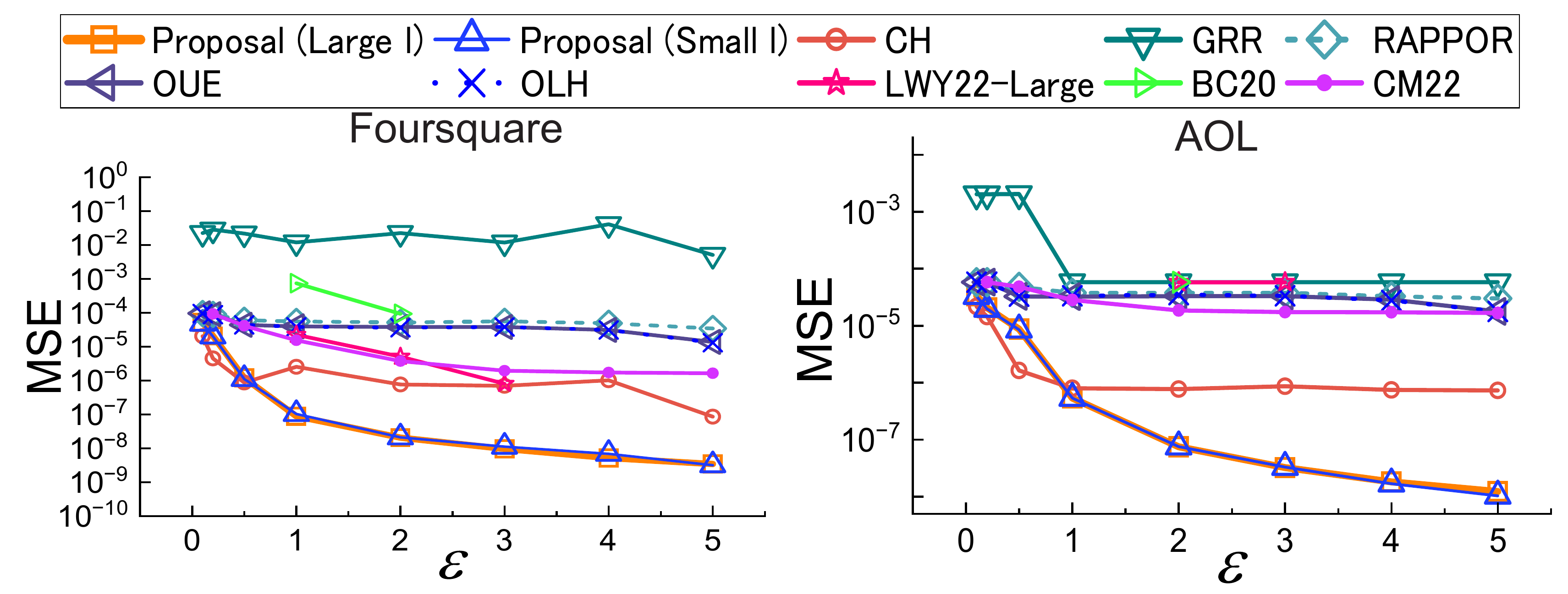}
  \includegraphics[width=0.99\linewidth]{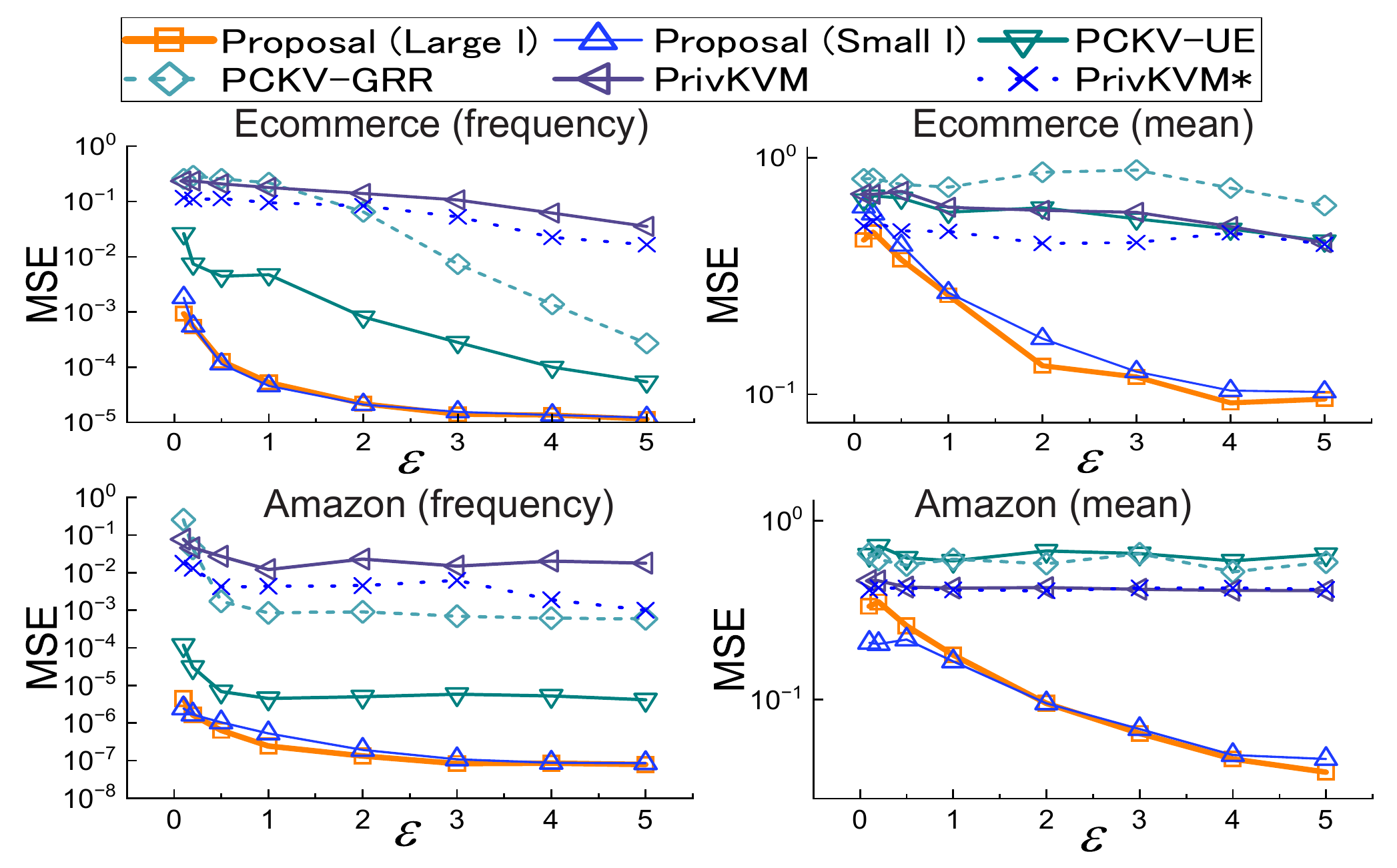}
  \vspace{-8mm}
  \caption{MSE vs. $\epsilon$ ($\delta=10^{-12}$).} 
  \label{fig:res_MSE}
\vspace{1mm}
  \centering
  \includegraphics[width=0.99\linewidth]{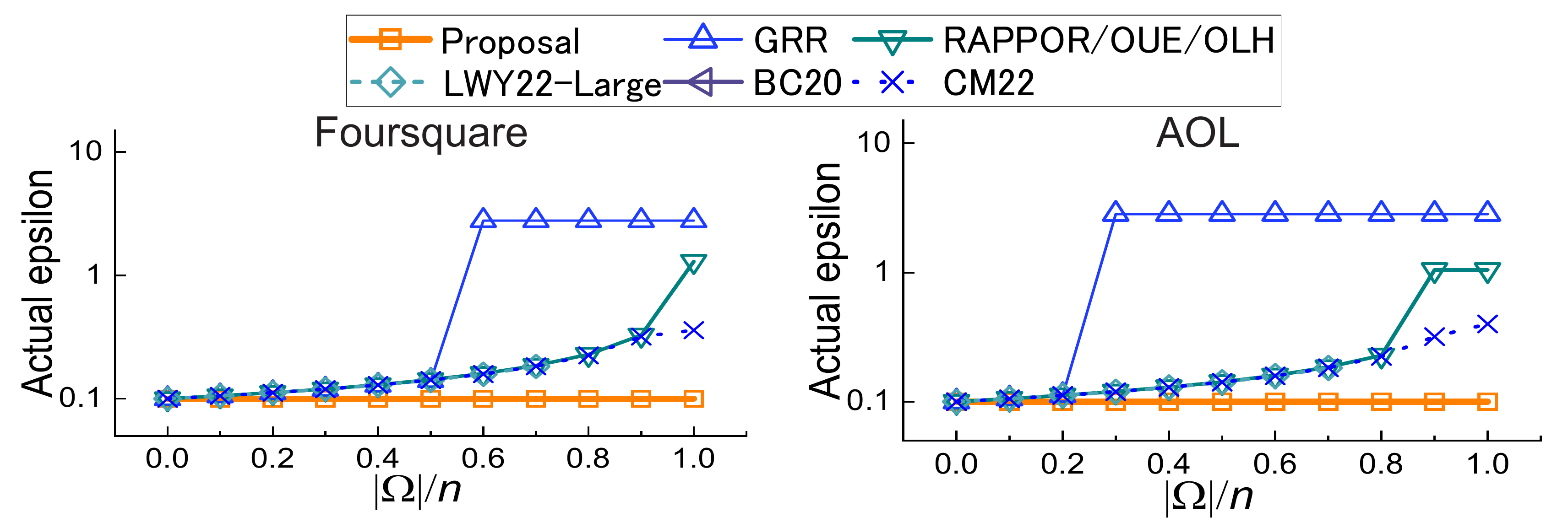}
  \includegraphics[width=0.99\linewidth]{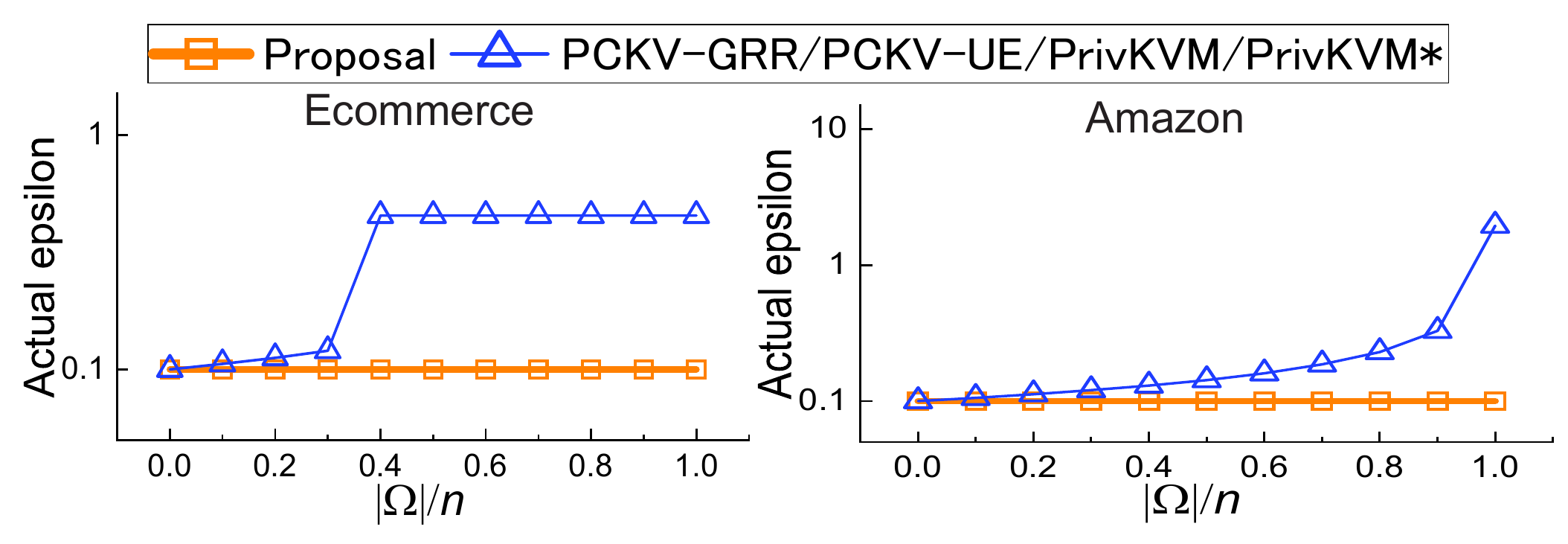}
  \vspace{-8mm}
  \caption{Actual $\epsilon$ vs. the ratio $|\Omega|/n$ of colluding users (target $\epsilon=0.1$, $\delta=10^{-12}$).} 
  \label{fig:res_collusion}
\end{figure}

\smallskip{}
\noindent{\textbf{Robustness to Collusion Attacks.}}~~Next, we evaluated the robustness to collusion with users. 
Fig.~\ref{fig:res_collusion} shows the relationship between the actual $\epsilon$ and the ratio $|\Omega|/n$ of colluding users. 

We observe that in all the existing protocols, the actual $\epsilon$ rapidly increases with an increase in $|\Omega|$. 
This vulnerability is inevitable in the existing protocols because they add noise on the user side. 
In Appendix~\ref{sub:existing_defenses}, we also show that the defense in \cite{Wang_PVLDB20} is insufficient in that the actual $\epsilon$ still increases with an increase in $|\Omega|$. 
In contrast, the actual $\epsilon$ always coincides with the target $\epsilon$ in our proposals, 
demonstrating the robustness of our proposals to collusion with users. 

In Appendix~\ref{sub:parameterx_collusion_poisoning}, we also set $|\Omega|/n = 0.1$ and evaluate the relationship between the actual $\epsilon$ and the target $\epsilon$. 

\begin{figure}[t]
  \centering
  \includegraphics[width=0.99\linewidth]{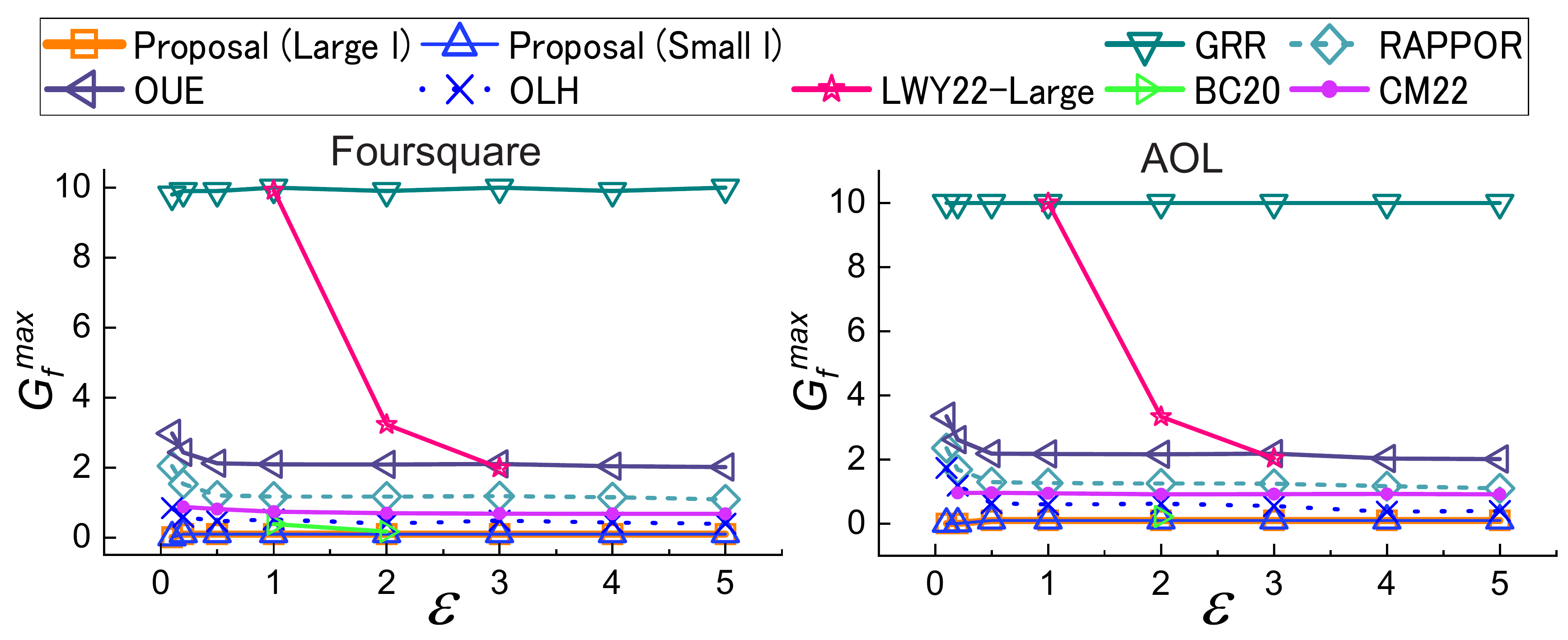}
  \includegraphics[width=0.99\linewidth]{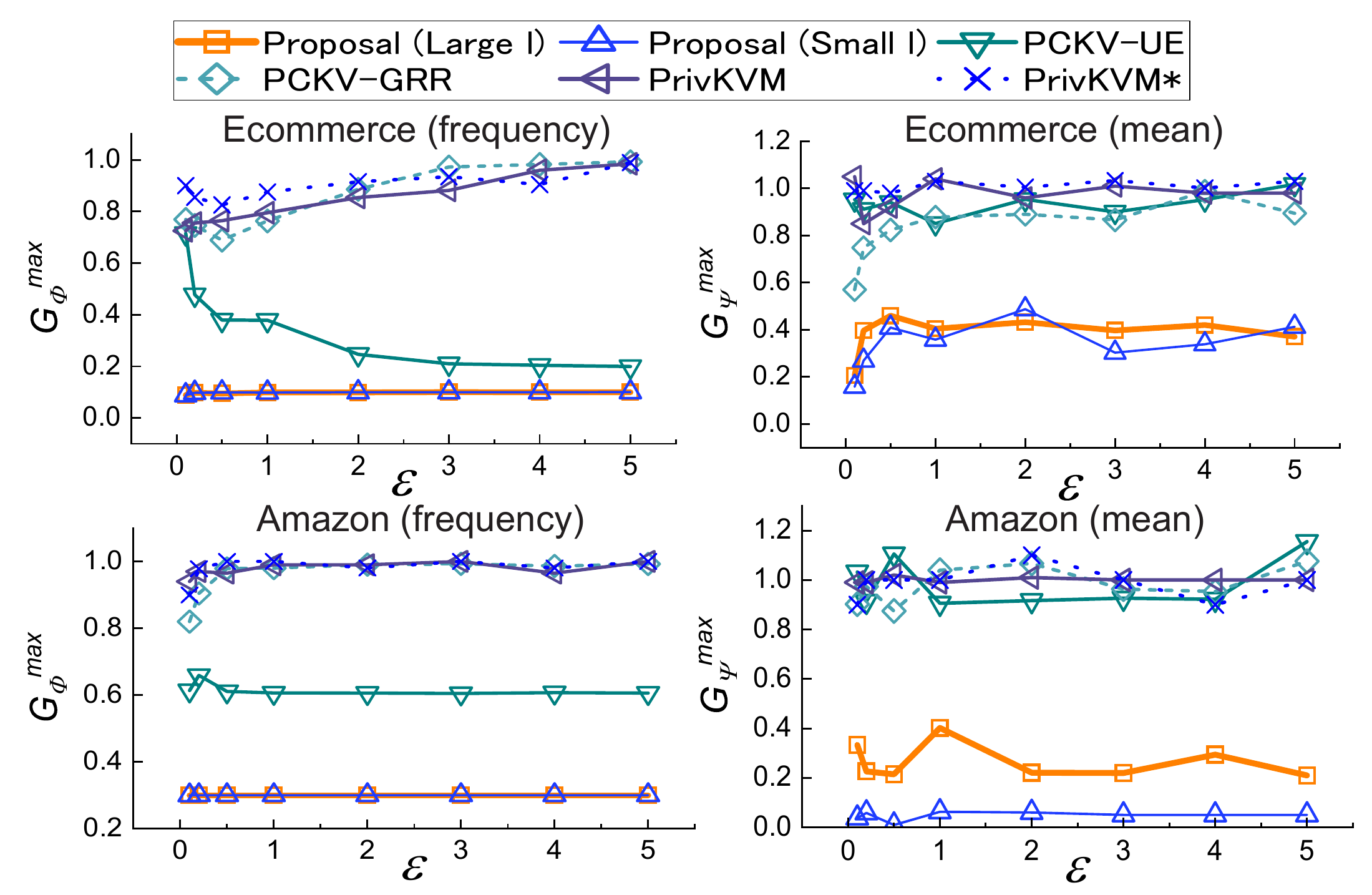}
  \vspace{-8mm}
  \caption{Maximum gains $\GMGA$, $\GMTGAPhi$, and $\GMTGAPsi$ vs. $\epsilon$ ($\lambda=0.1$, $\delta=10^{-12}$).} 
  \label{fig:res_poisoning}
\vspace{1mm}
  \centering
  \includegraphics[width=0.99\linewidth]{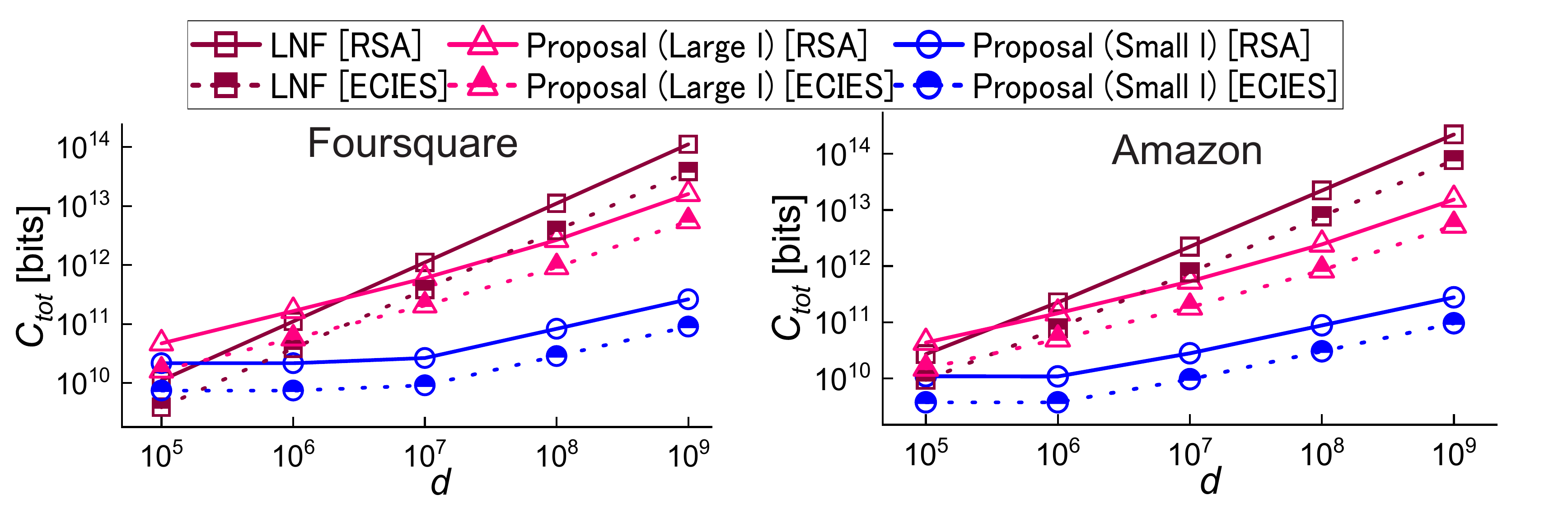}
  \includegraphics[width=0.99\linewidth]{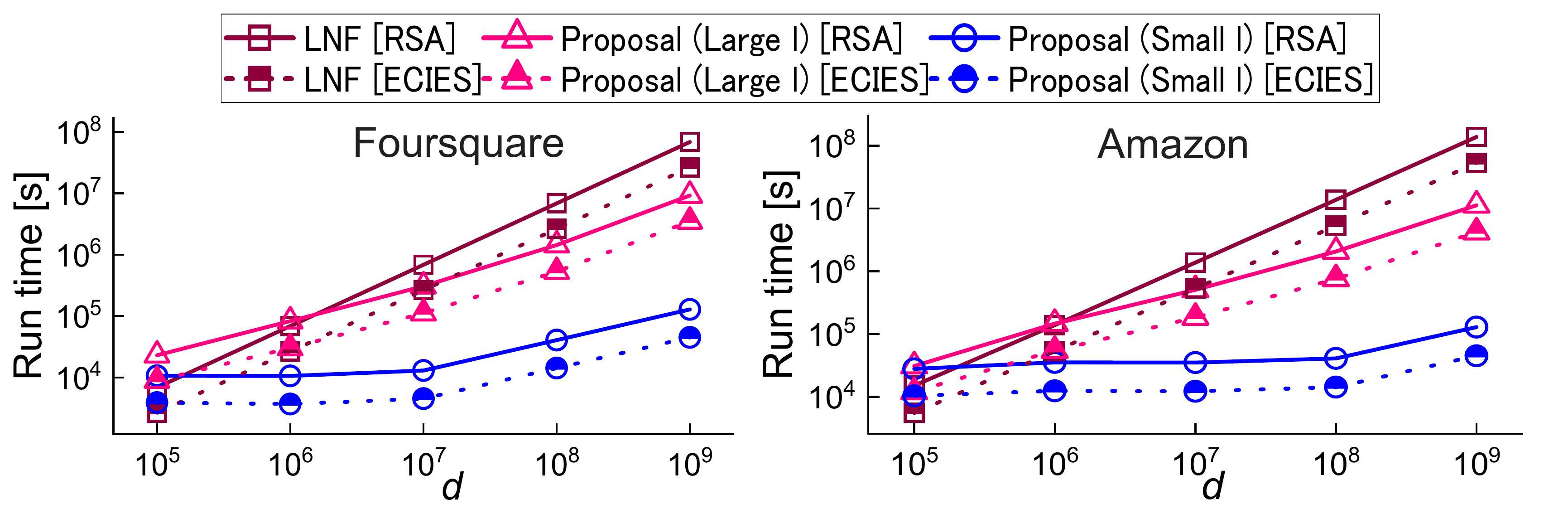}
  \vspace{-8mm}
  \caption{Computational cost $C_{tot}$ and run time ($\epsilon=1$, $\delta=10^{-12}$).} 
  \label{fig:res_efficiency}
\end{figure}

\smallskip{}
\noindent{\textbf{Robustness to Poisoning Attacks.}}~~We also evaluated the robustness to poisoning attacks. 
Here, we 
randomly selected $|\calT|=10$ target items for categorical data and $|\calT|=1$ target item for KV data. 
Then, we set the fraction $\lambda$ of fake users to $\lambda = 0.1$; we also changed $\lambda$ in Appendix~\ref{sub:parameterx_collusion_poisoning}. 

Fig.~\ref{fig:res_poisoning} shows the results.\footnote{In Fig.~\ref{fig:res_poisoning}, we omit the gain for the CH protocol because it is very close to the gains for our proposals.} 
We observe 
that the existing protocols suffer from large maximum gains. 
In Appendix~\ref{sub:existing_defenses}, we also show that the defense in~\cite{Wu_USENIX22} has limited effectiveness. 
In contrast, 
our proposals always achieve small 
$\GMGA$, $\GMTGAPhi$, and $\GMTGAPsi$, 
and are robust to poisoning attacks.

\smallskip{}
\noindent{\textbf{Efficiency.}}~~Finally, we evaluated 
$C_{tot}$ and the run time. 
Specifically, we measured $C_{tot}$ and the run time of the LNF protocol, \ProposalL{}, and \ProposalS{} in the Foursquare and Amazon datasets. 
For an encryption scheme, we used the $2048$-bit RSA or ECIES with $256$-bit security in~\cite{bouncy}. 
Then, we calculated the time to encrypt or decrypt a single message using multiple encryption and estimated $C_{tot}$ and the run time when we changed $d$ from $10^5$ and $10^9$. 

Fig.~\ref{fig:res_efficiency} shows the results. 
We observe that $C_{tot}$ and the run time of the LNF protocol are linear in $d$ and extremely large when $d$ is large; e.g., $C_{tot}$ is about $100$ Terabits and the run time is about $3$ years when $d=10^9$ and ECIES is used. 
\ProposalS{} addresses this issue. 
Specifically, in \ProposalS{}, both $C_{tot}$ and the run time do not depend on $d$ until $d=10^6$ or $10^7$ and then increase in $O(\sqrt{d})$, which is consistent with our theoretical results in Section~\ref{sub:FME_optimization}. 

For example, when $d=10^9$, \ProposalS{} reduces $C_{tot}$ from about $100$ Terabits to $260$ Gigabits and the run time from about $3$ years to $1$ day. 
Thus, 
accurate, robust, and efficient data analysis over large-domain categorical and KV data 
is now possible under DP.

\section{Conclusion}
\label{sec:conclusion}
We proposed the FME protocol for large-domain categorical and KV data 
and showed its effectiveness through theoretical analysis and extensive experiments. 
\colorB{LDP protocols for categorical data 
serve as a basis for many complex tasks, such as frequent itemset mining \cite{Tong_CCS24}, ranking estimation \cite{Zhan_arXiv25}, and range queries \cite{Liao_arXiv25}. 
Thus, we believe our FME protocol can also be used as a building block for such tasks. 
For future work, we would like to generalize our protocol for such tasks.}

\section*{Acknowledgment}
This study was supported in part by JSPS KAKENHI 22H00521, 24H00714, 24K20775, JST NEXUS JPMJNX25C2, JST AIP Acceleration Research JPMJCR22U5, and JST CREST JPMJCR22M1.

\bibliographystyle{IEEEtran}
\bibliography{main_short}

\appendices
\section{Basic Notations}
\label{sec:notation_table}
Table~\ref{tab:notations} shows the basic notations in this paper. 

\section{\colorB{More Details on the Baselines}}
\label{sec:details_baselines}

\subsection{\colorB{Algorithmic Description of the LNF Protocol}}
\label{sub:LNF_algorithm}
\colorB{Algorithm~\ref{alg:S_LNF} shows an algorithmic description of the LNF protocol $\calS_{\calD, \beta}^{\LNF}$. 
The \texttt{Count} function (line 11) calculates a count $\tc_i$ for each item $i \in [d]$ from $y_{\pi(1)}, \ldots, y_{\pi(\tn)}$.}

\subsection{\colorB{DP and Robustness of the CH Protocol}}
\label{sub:CH_DP_robustness}

\colorB{We show DP and the robustness of the CH protocol $\calS_{\calD, \beta}^{\CH}$:}
\begin{theorem}
\label{thm:CH_DP}
If the binary input mechanism $\calM_{\calD, \beta}$ in Definition~\ref{def:binary_input} provides $(\frac{\epsilon}{2}, \frac{\delta}{2})$-DP, then $\calS_{\calD, \beta}^{\CH}$ provides $(\epsilon, \delta)$-CDP 
and is robust to collusion with users. 
\end{theorem}

\begin{theorem}
\label{thm:CH_poisoning}
Let $\lambda = \frac{n'}{n+n'}$ and $f_\calT = \sum_{i \in \calT} f_i$. 
$\calS_{\calD, \beta}^{\CH}$ provides the following robustness against poisoning attacks:
\begin{align}
\GMGA = \lambda (|\calT| - f_\calT).
\label{eq:CH_GMGA}
\end{align}
\end{theorem}

\begin{table}[t]
\caption{Basic notations.}
\vspace{-3mm}
\centering
\hbox to\hsize{\hfil
\begin{tabular}{l|l}
\hline
Symbol		&	Description\\
\hline
$n$            &   Number of users.\\
$d$            &	Number of items.\\
$u_i$          &	$i$-th user.\\
$x_i$          &	Input value of user $u_i$.\\
$\calX$        &	Space of input data.\\
$f_i$        &	Frequency of item $i$ in categorical data. $\bmf = (f_1, \ldots, f_d)$.\\
$\hf_i$        &	Estimate of $f_i$. $\bmhf = (\hf_1, \ldots, \hf_d)$.\\
$\Phi_i$        &	Frequency of key $i$ in KV data. $\bmPhi = (\Phi_1, \ldots, \Phi_d)$.\\
$\hPhi_i$        &	Estimate of $\Phi_i$. $\bmhPhi = (\hPhi_1, \ldots, \hPhi_d)$.\\
$\Psi_i$        &	Mean value of key $i$ in KV data. $\bmPsi = (\Psi_1, \ldots, \Psi_d)$.\\
$\hPsi_i$        &	Estimate of $\Psi_i$. $\bmhPsi = (\hPsi_1, \ldots, \hPsi_d)$.\\
$\calT$         &   Set of target items.\\
$n'$            &   Number of fake users.\\
$\lambda$       &   Fraction of fake users ($\lambda = \frac{n}{n+n'}$).\\
$\GMGA$         &   Maximum gain in categorical data.\\
$\GMTGAPhi$     &   Maximum frequency gain in KV data.\\
$\GMTGAPsi$     &   Maximum mean gain in KV data.\\
$C_{tot}$       &   Expected number of bits sent from one party to another.\\ 
$h$             &   Hash function.\\
$b$             &   Range of hash function $h$.\\
\colorB{$\calD^*$}         &   \colorB{Dummy-count distributions ($\calD^* = (\calD_1,\calD_2)$; $\calD_1$ has mean} \\
                &   \colorB{$\mu_1$ and variance $\sigma_1^2$; $\calD_2$ has mean $\mu_2$ and variance $\sigma_2^2$).}\\
$\beta$         &   Sampling probability.\\
$\alpha$        &   Significance level.\\
$l$             &   Maximum number of selected hashes.\\
$\Lambda^H$     &   Set of selected hash values after filtering.\\
$\Lambda$       &   Set of selected items after filtering.\\
$\kappa$        &   Padding length in the padding-and-sampling technique.\\
\hline
\end{tabular}
\hfil}
\label{tab:notations}
\end{table}

\setlength{\algomargin}{5mm}
\begin{algorithm}[t]
  \SetAlgoLined
  \KwData{Input values $(x_1, \ldots, x_n) \in [d]^n$, 
  \#items $d \in \nats$, 
  dummy-count distribution $\calD$ 
  (mean: $\mu$, variance: $\sigma^2$), 
  sampling probability $\beta \in [0,1]$.
  }
  \KwResult{Estimates $\bmhf = (\hf_1, \cdots, \hf_d)$.}
  \tcc{Send input values}
  \ForEach{$i \in [n]$}{
    [$u_i$] Send $E_\pkd[x_i]$ to the shuffler\;
  }
  \tcc{Random sampling}
  [s] Sample $E_\pkd[x_i]$ $(i\in[n])$ with probability $\beta$\;
  \tcc{Dummy data addition.}
  \ForEach{$i \in [d]$}{
    [s] $z_i \leftarrow \calD$;
    Add a dummy $E_\pkd[i]$ for $z_i$ times\;
  }
  \tcc{Random shuffling}
  [s] Let $y_1, \ldots, y_{\tn} \in [d]$ the selected input values and dummies. 
  Sample a random permutation $\pi$ over $[\tn]$\;
  \tcc{Send shuffled values}
  [s] Send shuffled values $(E_\pkd[y_{\pi(1)}], \ldots, E_\pkd[y_{\pi(\tn)}])$ to the data collector\;
  \tcc{Compute an unbiased estimate}
  [d] Decrypt $y_{\pi(1)}, \ldots, y_{\pi(\tn)}$\;
  [d] $(\tc_1, \ldots, \tc_d) 
  \leftarrow$\texttt{Count}$(y_{\pi(1)}, \ldots, y_{\pi(\tn)})$\;
  \ForEach{$i \in [d]$}{
    [d] $\hf_i \leftarrow \frac{1}{n\beta}(\tc_i-\mu)$\;
  }
  \KwRet{$\bmhf = (\hf_1, \cdots, \hf_d)$}
  \caption{LNF protocol $\calS_{\calD, \beta}^{\LNF}$ \cite{Murakami_SP25}. 
  [$u_i$], [s], and [d] represents that the process is run by user $u_i$, the shuffler, the data collector. 
  $\pkd$ represents a public key of the data collector. 
  }\label{alg:S_LNF}
\end{algorithm}

\subsection{\colorB{Other Baselines}}
\label{sub:other_baselines}

\colorB{Below, we explain two other baselines than $\calS_{\calD, \beta}^{\LNF}$ and $\calS_{\calD, \beta}^{\CH}$:} 

\smallskip{}
\noindent{\textbf{User/Group-Dependent Hash Protocol.}}~~We can 
consider 
a variant of $\calS_{\calD, \beta}^{\CH}$ that uses 
a different hash function $h_i$ for each user $u_i$ 
(or each user group) 
to avoid the hash collision among users. 
However, this variant also results in low accuracy, as it needs to add dummy values for each hash function $h_i$ and each hash value in $[b]$ to provide DP. 
In Appendix~\ref{sub:comparison_other_baselines}, we show that this variant provides \textit{worse} accuracy than $\calS_{\calD, \beta}^{\CH}$. 

\smallskip{}
\noindent{\textbf{Protocol in~\cite{Luo_CCS22}.}}~~Yet another baseline is a protocol for large-domain data in~\cite{Luo_CCS22}, denoted by \LWYL{}. 
Specifically, \LWYL{} applies a hash function $h_i$ different for each user $u_i$ and then repeatedly adds dummy values to a tuple of a hash function and a hash value uniformly chosen from $\calH \times [b]$. 
This protocol also does not work well. 
This is because it generates dummy values uniformly at random, which is shown to be ineffective in \cite{Murakami_SP25}. 
In Section~\ref{sec:exp}, we 
show that \LWYL{} does not provide high accuracy. 

\section{\colorB{More Details on the FME Protocol}}
\label{sec:details_FME}

\noindent{\textbf{\colorB{Toy Example of $\calS_{\calD^*,\beta}^{\FME}$}.}}~~Assume 
that $n=5$, $d=8$, $b=4$, $(x_1, \ldots, x_5) = (2,8,4,8,2)$, 
$(h(x_1), \ldots, h(x_5)) \allowbreak = (1,1,3,1,1)$, $\beta=1$, and 
the binomial distribution $B(2,0.5)$ is used as 
\colorB{$\calD_1$ and $\calD_2$}. 
The shuffler adds $(z_1,z_2,z_3,z_4) \allowbreak = (1,0,1,1)$ dummies for hash values. 
Then, the shuffled data are, e.g., 
$(y^H_{\pi(1)}, \ldots, \allowbreak y^H_{\pi(8)}) = (1,3,1,4,1,1,1,3)$ and  
$(y_{\pi(1)}, \ldots, y_{\pi(8)}) \allowbreak = (8,4,2,\bot,8,\bot,2,\bot)$. 
The data collector filters items and 
selects $\Lambda = \{2,8\}$ and $\Lambda_H=\{1\}$. 
In this case, the shuffled input values become $(y_{\pi(1)}, \ldots, y_{\pi(8)}) \allowbreak = (8,\bot,2,\bot,8,\bot,2,\bot)$ ($4$ is replaced with $\bot$). 
Finally, the shuffler removes $\bot$ generated by the shuffler and 
adds $(z_2,z_8) = (1,2)$ dummies for input values. 
Then, the shuffled input values are, e.g., 
$(y^*_{\rho(1)}, \ldots, y^*_{\rho(8)}) \allowbreak = (2,8,\bot,8,2,8,2,8)$, 
and the estimate $\bmhf$ is $\bmhf \allowbreak = (0,0.4,0,0,0,0,0,0.6)$. 

\smallskip{}
\noindent{\textbf{\colorB{On the Composition in $\colorB{\calS_{\calD^*,\beta}^{\FME}}$}.}}~~\colorB{Below, we explain that the basic composition theorem in Theorem~\ref{thm:FME_DP} is almost tight. 
To show this, we consider the optimal composition theorem~\cite{Kairouz_ICML15}. 
This theorem provides the tightest bound when each sub-mechanism provides $(\epsilon,\delta)$-DP. 
However, even this theorem cannot improve the bound in Theorem~\ref{thm:FME_DP}, as our protocol composes \textit{only two} sub-mechanisms -- one outputting hash values and the other outputting input values. 
Specifically, assume that each sub-mechanism provides $(\frac{\epsilon}{2},\frac{\delta}{2})$-DP, i.e., $\epsilon_1 = \epsilon_2 = \frac{\epsilon}{2}$ and $\delta_1 = \delta_2 = \frac{\delta}{2}$ in Theorem~\ref{thm:FME_DP}. 
Then, the theorem in \cite{Kairouz_ICML15} states that the entire protocol provides $(\epsilon, \delta - (\frac{\delta}{2})^2)$-DP (or $(0,\delta')$-DP with an extremely large $\delta'$), 
which is almost equivalent to $(\epsilon,\delta)$-DP. 
Thus, the basic composition theorem is almost tight in our case.} 

\section{\colorB{$\varepsilon$ in DP and Privacy Guarantees}}
\label{sec:DP_epsilon}
\colorB{Below, 
we show the relationship between $\epsilon$ in DP and privacy guarantees against the inference of input values through hypothesis testing interpretations \cite{Kairouz_ICML15}. 
Specifically, DP considers two neighboring databases, $D$ and $D'$, that differ in the input value of the victim (see Definition~\ref{def:DP}). 
Thus, given an output $Y$ of a randomized algorithm $\calM$, we can define the following hypotheses: $H_0$: ``$Y$ came from $D$.''; $H_1$: ``$Y$ came from $D'$.'' 
Assume that, given $D$ and $D'$, the attacker who obtains $Y$ guesses which of $H_0$ or $H_1$ is correct. 
The attacker may choose $H_1$ when $H_0$ is true (type I error). 
Conversely, the attacker may choose $H_0$ when $H_1$ is true (type II error). 
Let $p_I, p_{II} \in [0,1]$ be the probabilities of type I and II errors, respectively. 
Then, DP is closely related to $p_I$ and $p_{II}$:}

\begin{theorem}[\cite{Kairouz_ICML15}]
\label{thm:DP_hypothesis}
\colorB{A randomized algorithm $\calM$ provides $(\epsilon,\delta)$-DP if and only if the following inequalities holds for any neighboring databases $D$ and $D'$ and any $Y \in \mathrm{Range}(\calM)$:}
\begin{align}
\colorB{p_I + e^\epsilon p_{II} \geq 1 - \delta,~~ e^\epsilon p_I + p_{II} \geq 1 - \delta.}
\label{eq:DP_hypothesis}
\end{align}
\end{theorem}
\colorB{Theorem~\ref{thm:DP_hypothesis} states that $(\epsilon,\delta)$-DP is equivalent to lower bounding type I and II errors by (\ref{eq:DP_hypothesis}). 
Fig.~\ref{fig:res_epsilon_hypothesis} shows the relationship between $\epsilon$ and the lower bound on the error probability $p^*$ when the type I and II errors are equal, i.e., $p^* = p_I = p_{II}$. 
Note that $p^*$ is the error probability when the attacker is given two candidates for the victim's input value. 
The inference of the victim's input value is much more difficult in practice, as there are $d$ ($\gg 2$) candidates for it.} 

\colorB{Fig.~\ref{fig:res_epsilon_hypothesis} shows that when $\epsilon$ is close to $0$, $p^*$ is close to $0.5$ (the error probability of a random guess). 
In contrast, when $\epsilon \geq 5$, $p^*$ is almost zero. 
This means that DP cannot provide strong privacy guarantees in this case, as claimed in \cite{DP_Li}. 
Taking this into account, we varied $\epsilon$ from $0.1$ to $5$ in Section~\ref{sec:exp}.} 

\section{Additional Experiments}
\label{sec:additional}

\subsection{\colorB{Changing Parameters in the FME Protocol}}
\label{sub:FME_parameters}
\colorB{We evaluated the MSE of $\calS_{\calD^*,\beta}^{\FME}$ when we changed parameters $\epsilon_1$ $(=\epsilon - \epsilon_2)$, $\alpha$, and $l$ in the Foursquare dataset. 
Here, we set $\epsilon_1 = 0.5\epsilon$, $\alpha = 0.05$, and $l = \max\{\frac{n^2}{d},50\}$ as default values and changed each parameter while fixing the others.} 

\begin{figure}[t]
  \centering
  \includegraphics[width=0.7\linewidth]{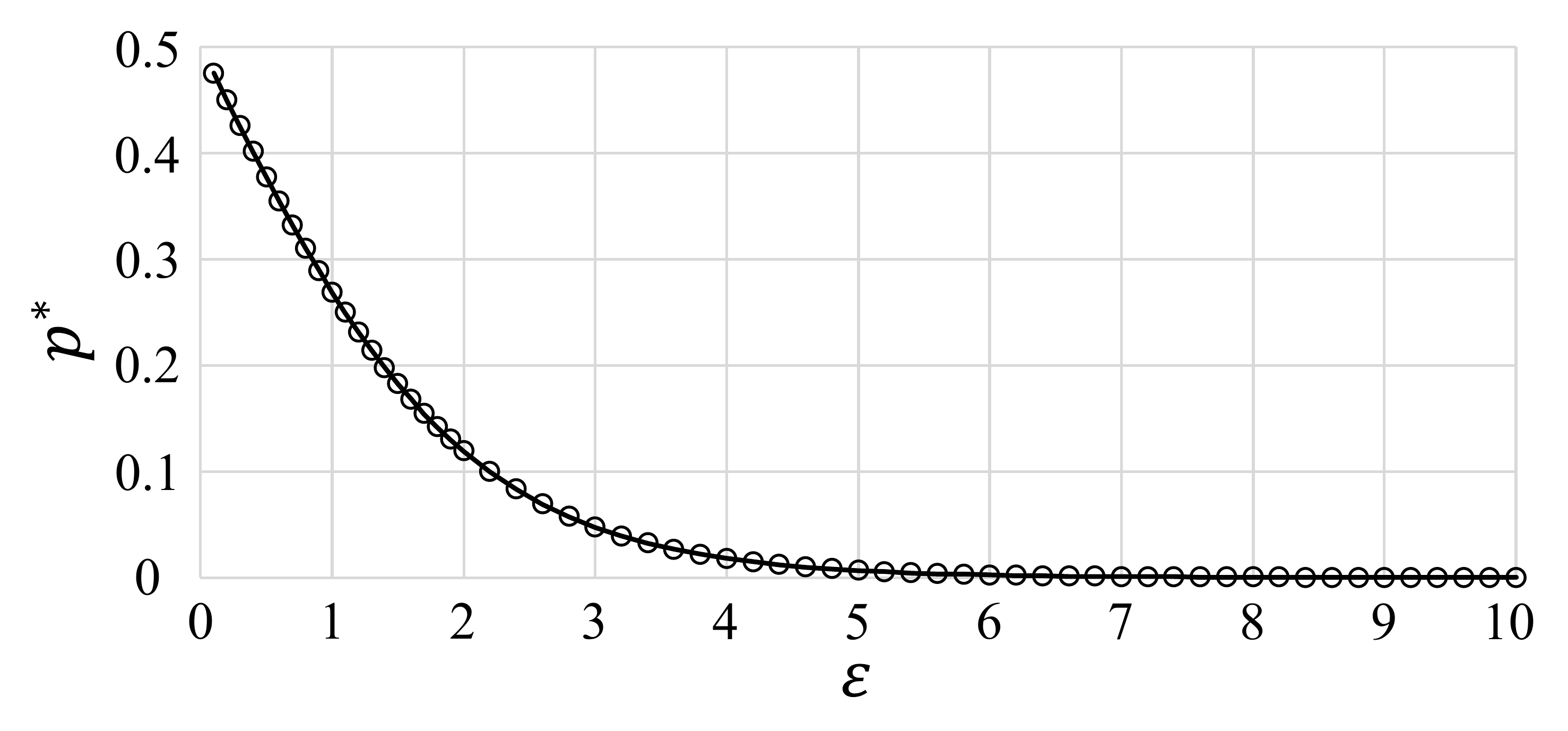}
  \vspace{-5mm}
  \caption{\colorB{$\epsilon$ in DP and the lower bound on the attacker's error probability $p^*$ $(= p_I = p_{II})$ obtained from (\ref{eq:DP_hypothesis}) ($\delta=10^{-12}$).}} 
  \label{fig:res_epsilon_hypothesis}
\vspace{1mm}
  \centering
  \includegraphics[width=0.99\linewidth]{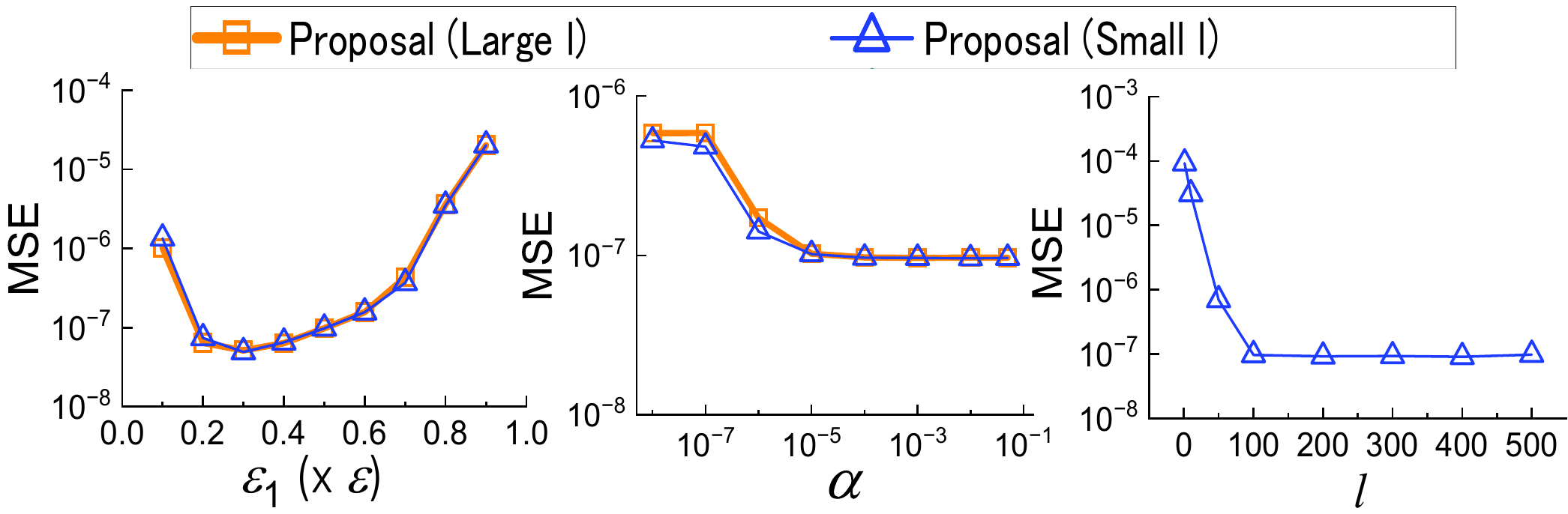}
  \vspace{-7mm}
  \caption{\colorB{MSE of our FME protocol when changing parameters $\epsilon_1$ $(=\epsilon - \epsilon_2)$, $\alpha$, and $l$ in the Foursquare dataset ($\epsilon = 1$, $\delta=10^{-12}$).}}
  \label{fig:res_changing_FME_parameters}
\end{figure}

\colorB{Fig.~\ref{fig:res_changing_FME_parameters} shows the results. 
The MSE is roughly the same when $\epsilon_1$ is between $0.2\epsilon$ and $0.5\epsilon$, which indicates that this range of $\epsilon_1$ can balance the trade-off between the noise for hash values and the noise for input values. 
Fig.~\ref{fig:res_changing_FME_parameters} also shows that the MSE rapidly increases as we decrease $\alpha$ and $l$ from $10^{-6}$ and $50$, respectively. 
This is because $\calS_{\calD^*,\beta}^{\FME}$ filters out almost all items in this case. 
Our suggestion is to avoid such extreme settings; e.g., if we are interested in the frequency of top-$k$ items, 
$\alpha$ and $l$ should be: $\alpha \geq 10^{-5}$ and $l \geq k$.}

\colorB{Note that $\calS_{\calD^*,\beta}^{\FME}$ also has the sampling probability $\beta$ as a parameter. 
\cite{Murakami_SP25} shows that the communication cost is improved by reducing $\beta$. 
However, it comes at the cost of accuracy. 
Since we can significantly improve the communication cost by setting $l$ small, our suggestion is to set $\beta=1$.}

\subsection{Existing Defenses against Collusion and Poisoning Attacks}
\label{sub:existing_defenses}
We evaluated the existing defenses against collusion and poisoning attacks not evaluated in Section~\ref{sec:exp}. 

\smallskip{}
\noindent{\textbf{Collusion Attacks.}}~~For collusion attacks, we evaluated a defense in \cite{Wang_PVLDB20}, which adds dummies uniformly at random from the domain of noisy data. 
This defense increases the MSE by $(1+\upsilon)^2$ times by adding $\upsilon n$ dummies, where $\upsilon \in \nnreals$. 
We set $\upsilon = 0.5$, in which case the MSE is increased by $2.25$ times. 
We applied the defense in \cite{Wang_PVLDB20} to the existing protocols and 
evaluated the relationship between the actual $\epsilon$ and the ratio $|\Omega|/n$ of colluding users using the Foursquare dataset. 

Fig.~\ref{fig:res_existing_defenses_collusion} shows the results. 
The actual $\epsilon$ still increases with an increase in $|\Omega|/n$, which indicates that the defense in \cite{Wang_PVLDB20} is insufficient as a defense against collusion attacks. 

\smallskip{}
\noindent{\textbf{Poisoning Attacks.}}~~For poisoning attacks, 
we evaluated a defense 
in \cite{Wu_USENIX22}. 
The defense in \cite{Wu_USENIX22} 
detects fake users based on isolation forest in KV statistics estimation.
Note that the isolation forest simply divides users into two groups. 
Since we can set the size of each group in the isolation forest, we assume that the data collector knows the number $n'$ of fake users, 
divides the users into a group with $n'$ users and another group, and treats the former group as a fake group. 
Following \cite{Wu_USENIX22}, we used the FPR (False Positive Rate) and FNR (False Negative Rate) as performance measures. 
The FPR (resp.~FNR) is the ratio of genuine users decided as fake (resp.~fake users decided as genuine). 
We evaluated the FPR and FNR of the defense in \cite{Wu_USENIX22} applied to PCKV-GRR/UE using the Amazon dataset. 

\begin{figure}[t]
  \centering
  \includegraphics[width=0.7\linewidth]{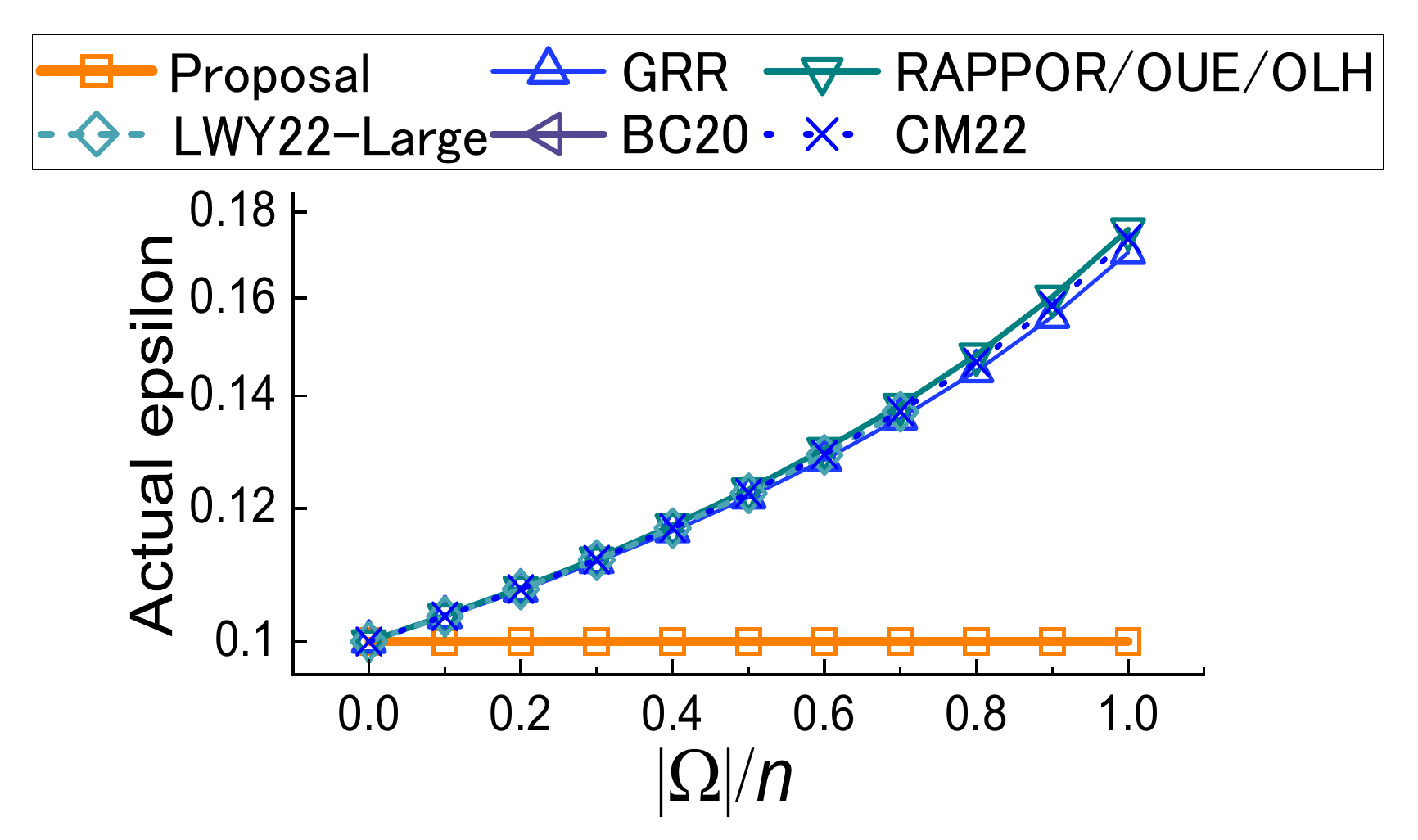}
  \vspace{-5mm}
  \caption{Actual $\epsilon$ vs. the ratio $|\Omega|/n$ of colluding users when the defense in \cite{Wang_NDSS20} is applied to the existing protocols in the Foursquare dataset ($\epsilon=0.1$, $\delta=10^{-12}$).} 
  \label{fig:res_existing_defenses_collusion}
\vspace{1mm}
  \centering
  \includegraphics[width=0.99\linewidth]{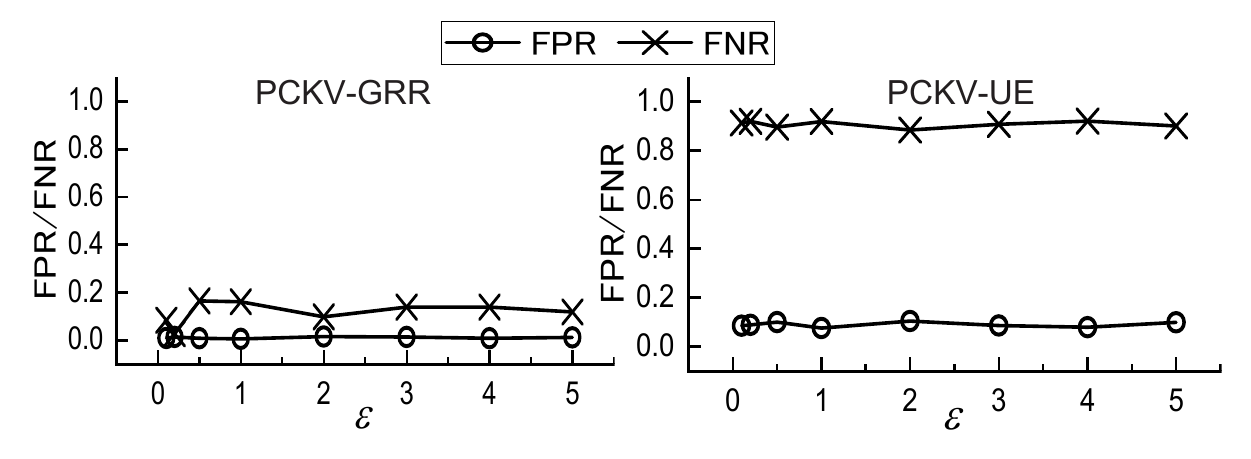}
  \vspace{-9mm}
  \caption{FPR and FNR of the defense in \cite{Wu_USENIX22} in the Amazon dataset ($\lambda=0.1$, $\delta=10^{-12}$).} 
  \label{fig:res_existing_defenses_poisoning}
\end{figure}

Fig.~\ref{fig:res_existing_defenses_poisoning} shows the results. 
Although the FPR and FNR are low for PCKV-GRR, the FNR is very high for PCKV-UE. 
This is because, for PCKV-UE, the M2GA sets bits corresponding to the target keys to $1$ and randomly samples other bits to evade detection \cite{Wu_USENIX22}. 
Our result shows that the effectiveness of the defense in \cite{Wu_USENIX22} is limited for PCKV-GRR. 
This is consistent with the experimental results in \cite{Wu_USENIX22}. 

\subsection{User/Group-Dependent Hash Protocol}
\label{sub:comparison_other_baselines}
We also evaluated the UH (User-Dependent Hash) and GH (Group-Dependent Hash) protocols. 

\smallskip{}
\noindent{\textbf{UH/GH Protocol.}}~~Below, we explain the GH protocol, as it is more general than the UH protocol. 
Let $g \in [n]$ be the number of groups. 
For each $i \in [g]$, the data collector randomly selects a hash function $h_i: [d] \rightarrow [b]$ from $\calH$. 
Then, the data collector selects a group ID $r_i \in [g]$ for each user $u_i$ ($i \in [n]$) and sends $r_i$ and 
$h_{r_i}$ to $u_i$. 
User $u_i$ ($i \in [n]$) sends a pair $\la r_i, h_{r_i}(x_i) \ra$ of the group ID and the hash value to the shuffler. 

The shuffler randomly selects each pair with probability $\beta$. 
For each group ID $j \in [g]$ and each hash value $k \in [b]$, the shuffler randomly generates $z_{jk}$ from the dummy-count distribution $\calD$ ($z_{jk} \sim \calD$) and adds a dummy pair $\la j, k \ra$ for $z_{jk}$ times. 
The shuffler randomly shuffles the selected hash values and dummies to the data collector. 
Finally, the data collector calculates an unbiased estimate $\hf_i$ ($i \in [d]$) of $f_i$ as: $\hf_i = \frac{b}{n\beta(b-1)}(\sum_{(j,k) \in S_i}\tc_{jk} - \frac{n\beta}{b} - g\mu)$, where $\tc_{jk}$ is the number of $\la j, k \ra$ in the shuffled data, and $S_i = \{(j,k) | h_j(i) = k\}$, 
i.e., the set $\la j, k \ra$ of pairs that could be produced from item $i$. 

Note that this protocol is equivalent to a protocol that independently applies the CH protocol to each group. 
Thus, the GH protocol inherits the theoretical properties of the CH protocol, e.g., $(\epsilon, \delta)$-CDP. 
The UH protocol is a special case of the GH protocol where $g=n$ and $r_i = i$ ($i \in [n]$).

\smallskip{}
\noindent{\textbf{Experimental Results.}}~~We compared our proposals with the UH and GH protocols using the Foursquare dataset. 
Note that the GH protocol 
is inefficient when $g$ is large. 
To reduce the run time, we randomly selected $n=100$ users and divided the city into $100 \times 100$ regions ($d=10000$) at regular intervals. 
Then, we evaluated the MSE over all items (regions) and the MSE over the top-50 items while changing $g$ in the GH protocol from $1$ to $100$. 
Note that the CH and UH protocols are the GH protocol with $g=1$ and $100$, respectively. 

\begin{figure}[t]
  \centering
  \includegraphics[width=0.99\linewidth]{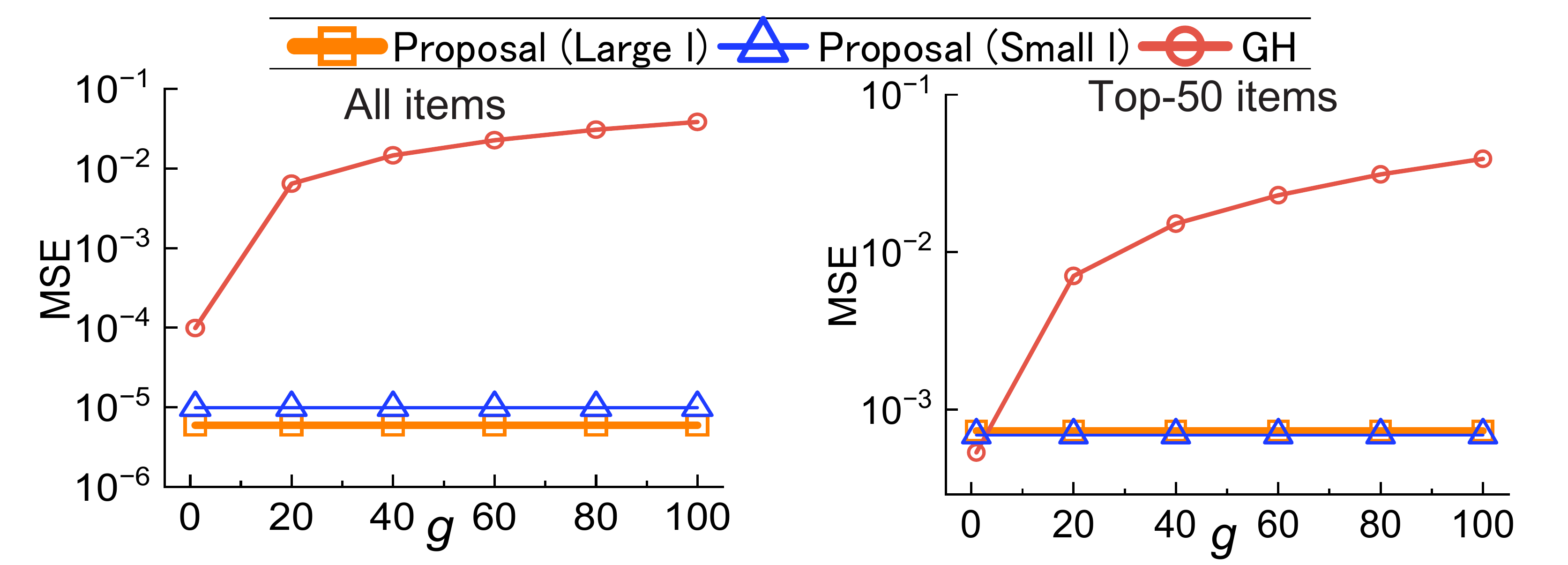}
  \vspace{-8mm}
  \caption{MSE of our proposal and the CH/GH/UH protocol in the Foursquare dataset (left: MSE over all items, right: MSE over the top-$50$ items). 
  The CH and UH protocols are the GH protocol with $g=1$ and $100$, respectively.} 
  \label{fig:res_GH}
\vspace{1mm}
  \centering
  \includegraphics[width=0.99\linewidth]{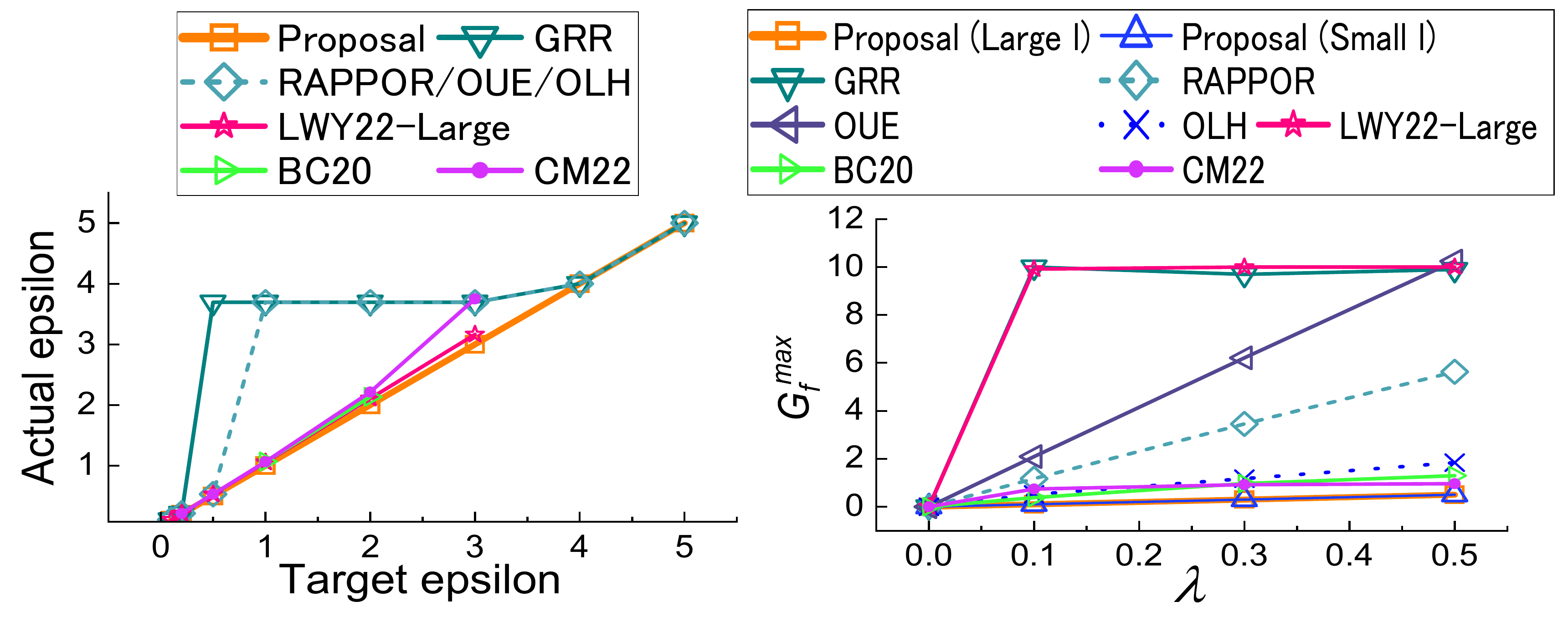}
  \vspace{-8mm}
  \caption{Robustness against collusion and poisoning attacks when varying the target $\epsilon$ and the fraction $\lambda$ of fake users in the Foursquare dataset (left: $|\Omega|/n = 0.1$, $\delta = 10^{-12}$; right: $\epsilon = 1$, $\delta = 10^{-12}$).} 
  \label{fig:res_changing_parameters}
\end{figure}

Fig.~\ref{fig:res_GH} shows the results. 
In the GH protocol, the MSE increases as $g$ increases. 
This is because the GH protocol needs to add dummies for each group. 
In addition, the CH protocol (i.e., GH with $g=1$) suffers from a large MSE over all items. 
This is because the CH protocol has poor accuracy for unpopular items due to the hash collision. 
The poor accuracy cannot be improved by using user/group-dependent hashes. 

\subsection{Changing Parameters in Collusion and Poisoning Attacks}
\label{sub:parameterx_collusion_poisoning}
We evaluated the robustness against collusion and poisoning attacks when varying the target $\epsilon$ and the fraction $\lambda$ of fake users in the Foursquare dataset. 
Fig.~\ref{fig:res_changing_parameters} shows the results. 

\begin{figure}[t]
  \centering
  \includegraphics[width=0.99\linewidth]{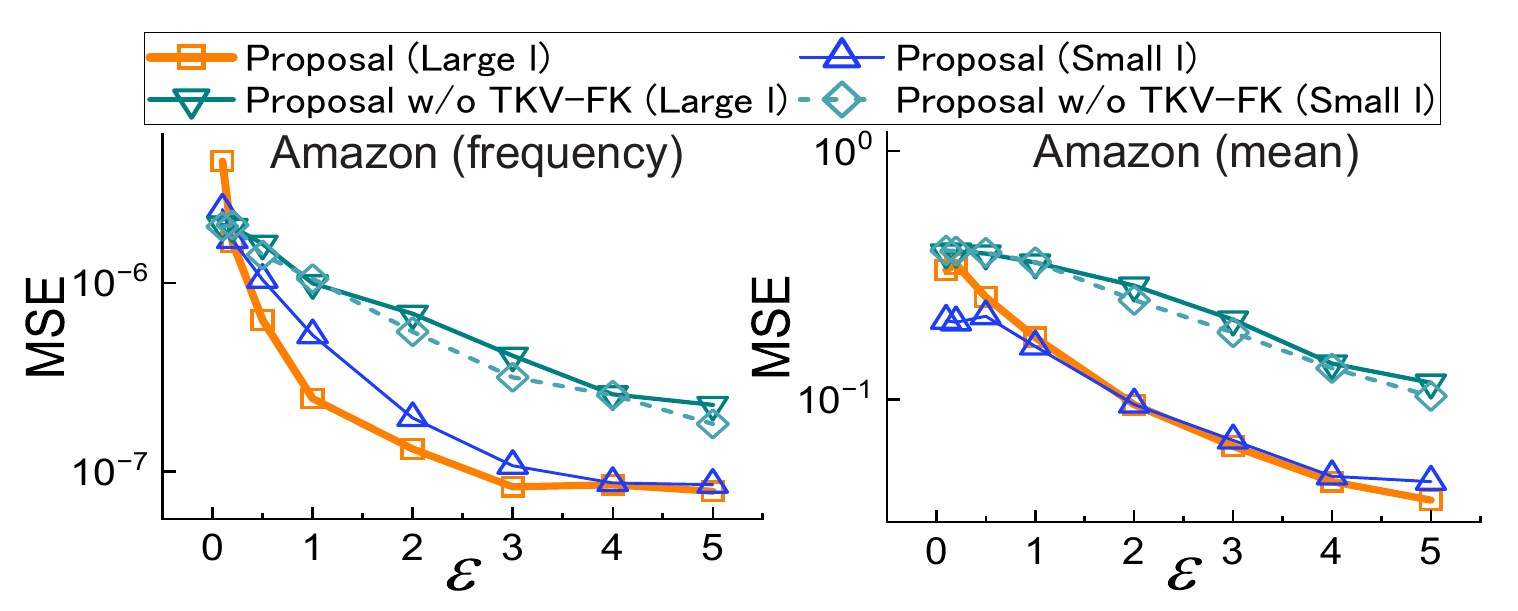}
  \vspace{-8mm}
  \caption{MSE of our proposals with/without the TKV-FK technique in the Amazon dataset  ($\delta=10^{-12}$).} 
  \label{fig:res_wo_TKV-FK}
\vspace{1mm}
  \centering
  \includegraphics[width=0.99\linewidth]{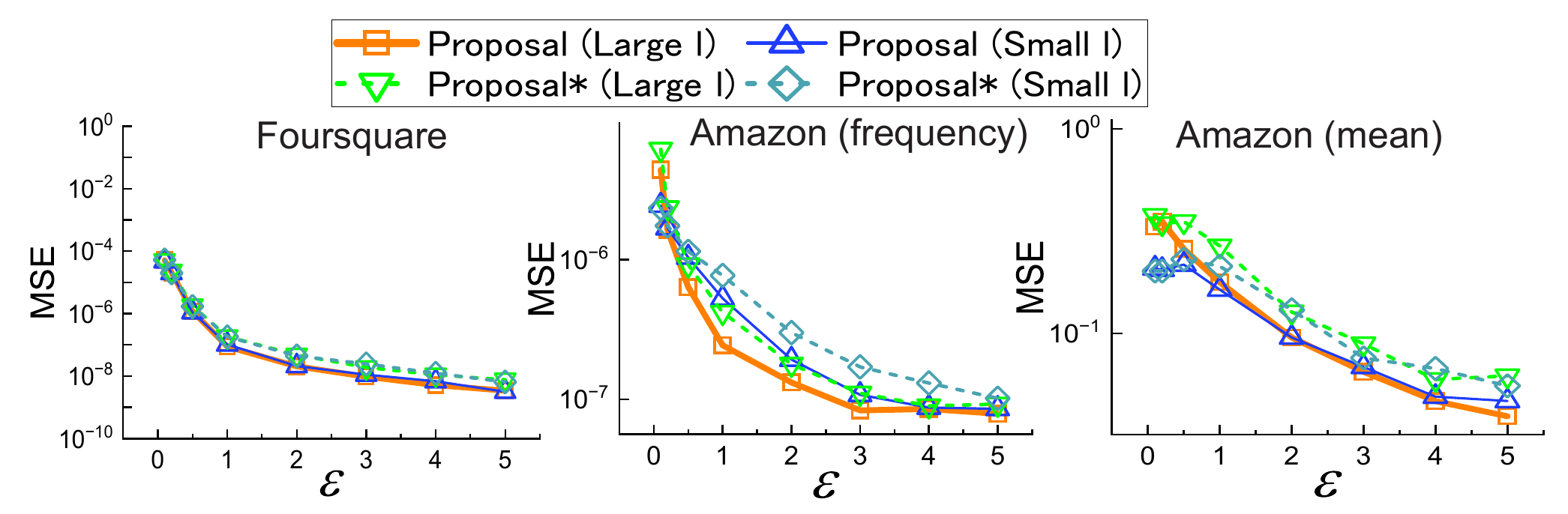}
  \vspace{-8mm}
  \caption{MSE of our proposals with additional noise in the Foursquare and Amazon datasets ($\delta=10^{-12}$).} 
  \label{fig:res_additional_noise}
\end{figure}

Fig.~\ref{fig:res_changing_parameters} shows that the actual $\epsilon$ rapidly increases as the target $\epsilon$ increases in the existing protocols. 
In contrast, the actual $\epsilon$ is always the same as with the target $\epsilon$ in our FME protocol, demonstrating robustness against collusion attacks. 
Our proposals also achieve the smallest gain $\GMGA{}$ and is the most robust to poisoning attacks for all values of $\lambda$. 

\subsection{Effectiveness of TKV-FK}
\label{sub:effectiveness_TKV_FK}
For an ablation study, we evaluated our proposal \textit{with} our TKV-FK technique (Fig.~\ref{fig:TKV-FK}(b)) and our proposal \textit{without} our TKV-FK technique (Fig.~\ref{fig:TKV-FK}(a)) in the Amazon dataset. 

Fig.~\ref{fig:res_wo_TKV-FK} shows the results. 
The MSE for means is significantly improved by introducing TKV-FK, demonstrating the effectiveness of TKV-FK for reducing bias. 
The MSE for frequencies is also significantly reduced by introducing TKV-FK. 
This is because our proposal with TKV-FK effectively finds popular items by filtering data at a key level. 

\subsection{Our Proposals with Additional Noise}
\label{sub:proposals_additional}

Theorem~\ref{thm:FME_DP} assumes 
that the shuffler does not disclose dummies after running the protocol. 
However, even if the shuffler discloses dummies, we can guarantee DP for the outputs 
of $\calS_{\calD^*,\beta}^{\FME}$ by adding additional DP noise before publishing them. 
We denote this modified protocol by \ProposalA{}. 

Specifically, \ProposalA{} adds noise randomly generated from the two-sided geometric distribution 
$\Geo(e^{-\epsilon/4})$ 
\cite{Ghosh_SICOMP12} with parameter $e^{-\epsilon/4}$ to each of the counts $\tc^H_1, \ldots, \tc^H_b$ (cf. Algorithm~\ref{alg:S_FME}, line 11). 
Since the sensitivity of the count is $2$, adding noise from $\Geo(e^{-\epsilon/4})$ provides $\frac{\epsilon}{2}$-DP for $\Lambda$. 
Similarly, \ProposalA{} adds noise from $\Geo(e^{-\epsilon/4})$ to each of the counts $\tc_1, \ldots, \tc_d$ (cf. Algorithm~\ref{alg:S_FME}, line 26) to provide $\frac{\epsilon}{2}$-DP for $\bmf$. 
Then, 
publishing $\Lambda$ and $\bmf$ provides $\epsilon$-DP, even if the shuffler discloses dummies after running the protocol. 

We evaluated the MSE of this 
protocol using the Foursquare and Amazon datasets. 
Fig.~\ref{fig:res_additional_noise} shows the results. 
The accuracy is hardly affected by introducing additional noise.

\arxiv{
\section{Proofs of Statements for Baselines}
\label{sec:proofs_baseline}

\subsection{Proof of Theorem~\ref{thm:CH_accuracy}}
First, we show that $\calS_{\CH}$ outputs an unbiased estimate. 
For $i \in [d]$, $n f_i$ users have item $i$, and $n (1-f_i)$ users do not have item $i$. 
Each of the former users 
increases the count $\tc_{h(i)}$ by one with probability $\beta$. 
Due to the hash collision, 
each of the latter users also increases $\tc_{h(i)}$ by one with probability $\frac{\beta}{b}$. 
In addition, the dummy data addition increases $\tc_{h(i)}$ by $z_{h(i)}$, whose expectation is $\E[z_{h(i)}] = \mu$. 
Thus, for $i \in [d]$, 
\begin{align*}
\textstyle{\E[\tc_{h(i)}] = n \beta f_i \cdot 1 + n \beta (1 - f_i) \cdot \frac{1}{b} + \mu}
\end{align*}
and 
\begin{align*}
\textstyle{\E[\hf_i] 
= \frac{b}{n\beta(b-1)}(\E[\tc_{h(i)}] - \frac{n\beta}{b} - \mu) 
= f_i,} 
\end{align*}
which means that $\calS_{\CH}$ outputs an unbiased estimate. 

Next, we analyze 
the expected squared error. 
Since $\bmhf$ is an unbiased estimate of $\bmf$, the expected squared error is equal to the variance: 
\begin{align}
\textstyle{\E [ (\hf_i - f_i)^2 ]} 
= \textstyle{\V[\hf_i]} 
= \textstyle{\frac{b^2}{n^2 \beta^2 (b-1)^2} \V[\tc_{h(i)}]}
\label{eq:CH_l2_loss_variance}
\end{align}
Below, we analyze $\V[\tc_{h(i)}]$. 
As explained above, $\tc_{h(i)}$ can be increased by users who do not have item $i$. 
For $j\in[d]$, let $\tc_{h(i),j}$ be the number of users who have item $j$ and have increased $\tc_{h(i)}$. 
Then, $\tc_{h(i)} = \sum_{j=1}^d \tc_{h(i),j} + z_{h(i)}$, and 
\begin{align}
&\V[\tc_{h(i)}] 
= \textstyle{\V[\sum_{j=1}^d \tc_{h(i),j}] + \V[z_{h(i)}]} \nonumber\\
&= \textstyle{\sum_{j=1}^d \V[\tc_{h(i),j}] + \sum_{j \ne j'} \text{Cov}(\tc_{h(i),j}, \tc_{h(i),j'}) + \sigma^2} \nonumber\\
&= \textstyle{\sum_{j=1}^d \V[\tc_{h(i),j}] + \sigma^2} \nonumber\\
& \hspace{5mm} \text{(as the hash function $h$ is 2-wise independent).}
\label{eq:CH_V_tc_hi}
\end{align}
For $j \ne i$, let $X_j$ be a random variable representing the number of users who have item $j$ and are selected after random sampling. 
Then, by the law of total variance, 
\begin{align*}
\V[\tc_{h(i),j}] 
&= \E[\V[\tc_{h(i),j} | X_j]] + \V[\E[\tc_{h(i),j} | X_j]] \\
&= \textstyle{\E[X_j^2 \cdot \frac{1}{b} \cdot (1-\frac{1}{b})] + \V[X_j \cdot \frac{1}{b}]}.
\end{align*}
Since $X_j$ is randomly generated from the Binomial distribution $B(n f_j, \beta)$, we have $\E[X_j] = n f_j \beta$, $\V[X_j] = n f_j \beta (1 - \beta)$, and $\E[X_j^2] = \E[X_j]^2 + \V[X_j]$. Thus, 
\begin{align}
\V[\tc_{h(i),j}] 
= \textstyle{\frac{(n^2 f_j^2 \beta^2 + n f_j \beta (1 - \beta))(b-1)}{b^2} + \frac{n f_j \beta (1 - \beta)}{b^2}}.
\label{eq:CH_V_tc_hi_j}
\end{align}
Recall that $n f_i$ users have item $i$, and each of them increases $\tc_{h(i)}$ by one with probability $\beta$. 
Thus, $\V[\tc_{h(i),i}] = n f_i \beta (1 - \beta)$. 
By this and (\ref{eq:CH_l2_loss_variance}), (\ref{eq:CH_V_tc_hi}), and (\ref{eq:CH_V_tc_hi_j}), we obtain (\ref{eq:CH_l2_loss}). 
\qed

\subsection{Proof of Theorem~\ref{thm:CH_DP}}
\label{sub:proof_CH_DP}
Let $D=(x_1,\ldots,x_n)$ and $D'=(x'_1,\ldots,x'_n)$ be two neighboring databases. 
Let $D_h=(h(x_1),\ldots,h(x_n))$ and $D'_h=(h(x'_1),\ldots,h(x'_n))$. 
Let $\calG_{\calD, \beta}^{\LNF}$ be an augmented shuffler part of $\calS_{\calD, \beta}^{\LNF}$, which runs lines 1-8 of Algorithm~\ref{alg:S_LNF} and outputs the shuffled data $(E_\pkd[y_{\pi(1)}], \ldots, E_\pkd[y_{\pi(\tn)}])$. 
Similarly, Let $\calG_{\calD, \beta}^{\CH}$ be an augmented shuffler part of $\calS_{\calD, \beta}^{\CH}$, which applies a hash function $h$ to all input values, runs lines 1-8 of Algorithm~\ref{alg:S_LNF}, and outputs the shuffled data. 

$\calG_{\calD, \beta}^{\CH}$ is equivalent to $\calG_{\calD, \beta}^{\LNF}$ where 
$D_h$ is used as a database. 
In other words, $\calG_{\calD, \beta}^{\CH}$ can be regarded as the LNF protocol preceded with hashing, i.e., $\calG_{\calD, \beta}^{\CH} = \calG_{\calD, \beta}^{\LNF} \circ h$. 
If $\calM_{\calD, \beta}$ provides $(\frac{\epsilon}{2}, \frac{\delta}{2})$-DP, then $\calG_{\calD, \beta}^{\LNF}$ provides 
$(\epsilon,\delta)$-CDP~\cite{Murakami_SP25}. 
In addition, $D$ and $D'$ remain neighboring after applying the hash function $h$; i.e., $D_h$ and $D'_h$ also differ on (at most) one entry.  
Thus, $\calG_{\calD, \beta}^{\CH}$ also provides 
$(\epsilon,\delta)$-CDP. 
By the immunity to post-processing \cite{DP}, $\calS_{\calD, \beta}^{\CH}$ provides 
$(\epsilon,\delta)$-CDP. 

The robustness of $\calG_{\calD, \beta}^{\CH}$ to collusion with users can be shown in the same way. 
Specifically, 
$\calG_{\calD, \beta}^{\LNF}$ is robust to collusion with users~\cite{Murakami_SP25}, and 
two neighboring databases $D$ and $D'$ remain neighboring after applying the hash function $h$. 
Thus, $\calG_{\calD, \beta}^{\CH}$ is also robust to collusion with users. 
By the immunity to post-processing, $\calS_{\calD, \beta}^{\CH}$ is robust to collusion with users. 
\qed

\subsection{Proof of Theorem~\ref{thm:CH_poisoning}}
\label{sub:proof_CH_poisoning}
Before data poisoning, 
$\E[\hf_i]$ can be written as follows: 
\begin{align*}
\E[\hf_i] = \textstyle{\frac{b}{n\beta(b-1)}(\E[\tc_{h(i)}] - \frac{n\beta}{b} - \mu).}
\end{align*}
Since $\tc_{h(i)}$ is the count of hash $h(i)$ sent to the data collector, the overall gain $G_f(\bmm)$ 
$(= \sum_{i \in \calT}(\E[\hf'_i] - \E[\hf_i]))$ 
is maximized 
when each of $n'$ fake users 
sends a message $m^*$ such that all target items $\calT$ are mapped to $m^*$, i.e., $m^* = h(i)$ for all $i \in \calT$ and $\bmm = (m^*, \ldots, m^*)$. 
Each fake message is sampled with probability $\beta$. 
In addition, $\calS_{\calD, \beta}^{\CH}$ outputs an unbiased estimate, i.e., $\E[\hf_i] = \frac{b}{n\beta(b-1)}(\E[\tc_{h(i)}] - \frac{n\beta}{b} - \mu) = f_i$ (see Theorem~\ref{thm:CH_accuracy}). 
Thus, 
we have 
\begin{align*}
&\GMGA 
= \textstyle{\max_{\bmm} \sum_{i \in \calT}(\E[\hf'_i] - \E[\hf_i])} \\
&= \textstyle{\sum_{i \in \calT} \left(\frac{b}{(n+n')\beta(b-1)}(\E[\tc_{h(i)}] + n'\beta - \frac{(n+n')\beta}{b} - \mu) - f_i \right)} \\
&= \textstyle{\sum_{i \in \calT} \left(\frac{n}{n+n'} f_i + \frac{b}{(n+n')\beta(b-1)}(n'\beta - \frac{n'\beta}{b}) - f_i \right)} \\
&= \textstyle{\lambda (|\calT| - f_\calT),} 
\end{align*}
where $\lambda = \frac{n'}{n+n'}$ and $f_\calT = \sum_{i \in \calT} f_i$. 
The first term $\lambda |\calT|$ is proportional to $|\calT|$ because each fake user increases the estimates for all target items $\calT$ by sending a message $m^*$ such that all target items $\calT$ are mapped to $m^*$. 
\qed

Note that in practice, it may be difficult for the attacker to find a message $m^*$ such that all target items $\calT$ are mapped to $m^*$. 
For example, if only one target item is mapped to $m^*$, then $\GMGA$ of the CH protocol is $\GMGA = \lambda (1 - f_\calT)$, which is the same as $\GMGA$ of the LNF protocol in (\ref{LNF_GMGA}). 

\section{Proofs of Statements in Section~\ref{sec:proposed}}
\label{sec:proofs_proposed}
\subsection{Proof of Theorem~\ref{thm:FME_DP}}
\label{suc:thm:FME_DP_proof}

Let $\colorB{\calG_{\calD^*, \beta}^{\FME\plain}}$ be an algorithm that takes $(x_1, \ldots, x_n)$ as input and outputs $(y^H_{\pi(1)}, \ldots, y^H_{\pi(\tn)})$ and $(y^*_{\rho(1)}, \ldots, y^*_{\rho(\tn^*)})$ in Algorithm~\ref{alg:S_FME}, i.e., \textit{plaintext messages} the data collector obtains in $\colorB{\calS_{\calD^*, \beta}^{\FME}}$. 
We divide the proof of Theorem~\ref{thm:FME_DP} into two parts. 
First, we prove that 
$\colorB{\calG_{\calD^*, \beta}^{\FME\plain}}$ 
provides $(\epsilon, \delta)$-DP (Lemma~\ref{lem:FME_plain_DP}). 
Then, we prove that $\colorB{\calS_{\calD^*, \beta}^{\FME}}$ provides $(\epsilon, \delta)$-CDP and is robust to collusion with users (Lemma~\ref{lem:FME_CDP}). 

\begin{lemma}
\label{lem:FME_plain_DP}
\colorB{If the binary input mechanisms $\calM_{\calD_1, \beta}$ and 
$\calM_{\calD_2, 1}$ 
provide $(\frac{\epsilon_1}{2}, \frac{\delta_1}{2})$-DP and $(\frac{\epsilon_2}{2}, \frac{\delta_2}{2})$-DP, respectively, then $\calG_{\calD^*, \beta}^{\FME\plain}$ provides 
$(\epsilon, \delta)$-DP, where $(\epsilon, \delta) = (\epsilon_1 + \epsilon_2, \delta_1 + \delta_2)$, and is robust to collusion with users.} 
\end{lemma}

\begin{proof}[Proof of Lemma~\ref{lem:FME_plain_DP}]
Assume that \colorB{$\calM_{\calD_1, \beta}$ and 
$\calM_{\calD_2, 1}$ 
provide $(\frac{\epsilon_1}{2}, \frac{\delta_1}{2})$-DP and $(\frac{\epsilon_2}{2}, \frac{\delta_2}{2})$-DP, respectively.} 
Let $\colorB{\calG_{\calD^*, \beta}^{\FME\plain 1}}$ be an algorithm that takes $(x_1, \ldots, x_n)$ as input and outputs $(y^H_{\pi(1)}, \ldots, y^H_{\pi(\tn)})$. 
As with the augmented shuffler part $\calG_{\calD, \beta}^{\CH}$ of the CH protocol 
(see Appendix~\ref{sub:proof_CH_DP}), 
$\colorB{\calG_{\calD^*, \beta}^{\FME\plain 1}}$ applies a hash function $h$ to input values $x_1,\ldots,x_n$, 
runs lines 1-8 of Algorithm~\ref{alg:S_LNF}, and outputs the shuffled data. 
Thus, by Theorem~\ref{thm:CH_DP}, $\colorB{\calG_{\calD^*, \beta}^{\FME\plain 1}}$ provides 
\colorB{$(\epsilon_1, \delta_1)$-DP} 
and is robust to collusion with users. (Note that $\colorB{\calG_{\calD^*, \beta}^{\FME\plain 1}}$ provides DP rather than CDP because it outputs plaintext messages.)

Similarly, let $\colorB{\calG_{\calD^*, \beta}^{\FME\plain 2}}$ be an algorithm that takes $(x_1, \ldots, x_n)$ as input and outputs $(y^*_{\rho(1)}, \ldots, y^*_{\rho(\tn^*)})$. 
Note that $(y^*_{\rho(1)}, \ldots, y^*_{\rho(\tn^*)})$ is randomly shuffled using $\rho$. 
In addition, the data collector knows the number of bots $\bot$ in $(y^*_{\rho(1)}, \ldots, y^*_{\rho(\tn^*)})$, as she has replaced unselected items with $\bot$ by herself. 
Thus, $\bot$ in  $(y^*_{\rho(1)}, \ldots, y^*_{\rho(\tn^*)})$ does not leak any information to the data collector. 
Then, $\colorB{\calG_{\calD^*, \beta}^{\FME\plain 2}}$ is equivalent to the LNF protocol~\cite{Murakami_SP25} (with sampling probability $= 1$) applied to the selected items $\Lambda$. 
\colorB{If 
$\calM_{\calD_2, 1}$ 
provides $(\frac{\epsilon_2}{2}, \frac{\delta_2}{2})$-DP, then the LNF protocol provides $(\epsilon_2, \delta_2)$-DP and is robust to collusion with users (see Theorem~2 in~\cite{Murakami_SP25}).} 
Thus, $\colorB{\calG_{\calD^*, \beta}^{\FME\plain 2}}$ also provides 
\colorB{$(\epsilon_2, \delta_2)$-DP} 
and is robust to collusion with users. 

Then, by the basic composition theorem~\cite{DP}, \colorB{$\calG_{\calD^*, \beta}^{\FME\plain} = (\calG_{\calD^*, \beta}^{\FME\plain 1}, \calG_{\calD^*, \beta}^{\FME\plain 2})$ provides $(\epsilon, \delta)$-DP, where $(\epsilon, \delta) = (\epsilon_1 + \epsilon_2, \delta_1 + \delta_2)$, and is robust to collusion with users.} 
\end{proof}

\begin{lemma}
\label{lem:FME_CDP}
If $\colorB{\calG_{\calD^*, \beta}^{\FME\plain}}$ provides $(\epsilon, \delta)$-DP and is robust to collusion with users, then $\colorB{\calS_{\calD^*, \beta}^{\FME}}$ provides $(\epsilon, \delta)$-CDP and is robust to collusion with users. 
\end{lemma}

\begin{proof}[Proof of Lemma~\ref{lem:FME_CDP}]
We prove 
this lemma 
by reducing DP (and the robustness) of $\colorB{\calS_{\calD^*, \beta}^{\FME}}$ to 
the security 
of the PKE scheme. 
See Appendix~\ref{sec:computational_DP_FME} for details. 
\end{proof}

By Lemmas~\ref{lem:FME_plain_DP} and \ref{lem:FME_CDP}, Theorem~\ref{thm:FME_DP} holds. 
\qed

\subsection{Proof of Theorem~\ref{thm:FME_poisoning}}
\label{sub:proof_FME_poisoning}
Before data poisoning, $\E[\hf_i]$ can be written as 
\begin{align*}
\E[\hf_i] = \textstyle{\frac{1}{n\beta}(\E[\tc_i] - \colorB{\mu_2}).}
\end{align*}
Since $\tc_i$ is the number of item $i$ sent to the data collector, the overall gain $G_f(\bmm)$ 
$(= \sum_{i \in \calT}(\E[\hf'_i] - \E[\hf_i]))$ 
is maximized when (i) each of $n'$ fake users sends some target item $i \in \calT$ (and its hash value) and (ii) all target users are selected in the filtering step (i.e., $\calT \subseteq \Lambda$). 
Since $\colorB{\calS_{\calD^*, \beta}^{\FME}}$ outputs an unbiased estimate for any item $i \in \Lambda$ selected in the filtering step (see Theorem~\ref{thm:FME_accuracy}), we have 
\begin{align*}
\GMGA = \textstyle{\max_{\bmm} \sum_{i \in \calT}(\E[\hf'_i] - \E[\hf_i]) = \sum_{i \in \calT}(f'_i - \E[\hf_i])}
\end{align*}
(by (ii)), where $f'_i = \frac{1}{n+n'}(\sum_{j=1}^n \mathone_{x_j = i} + \sum_{j=1}^{n'} \mathone_{y_j = i})$ and $y_j$ is an input value of the $j$-th fake user sent to the shuffler. 
In addition, we have $(n+n') (\sum_{i \in \calT} f'_i) = n (\sum_{i \in \calT} f_i) + n'$ in this case (by (i)). 
Moreover, we have $\E[\hf_i] = (1 - \eta_i) \cdot f_i + \eta_i \cdot 0$, as $i \in \Lambda$ with probability $1 - \eta_i$. 
Thus, we have 
\begin{align*}
\GMGA 
&= \textstyle{\frac{n'}{n+n'} + \frac{n}{n+n'} (\sum_{i \in \calT} f_i) - \sum_{i \in \calT} (1 - \eta_i) f_i} \\
&= \textstyle{\lambda (1 - f_\calT) + \sum_{i\in\calT} \eta_i f_i,}
\end{align*}
which proves (\ref{eq:FME_GMGA}). 
\qed

\smallskip{}
\noindent{\textbf{Probability $\eta_i$.}}~~The probability $\eta_i$ in (\ref{eq:FME_GMGA}) 
rapidly decreases as $f_i$ increases. 
For example, assume that $l=b$. 
In this case, the $i$-th item is selected if the count of the corresponding hash value is larger than or equal to $c_{th}$. 
Since $n f_i$ users have the $i$-th item and each of them is selected by random sampling with probability $\beta$, the probability $\eta_i$ is upper bounded as 
\begin{align*}
\eta_i \leq \Pr(B(nf_i,\beta) < c_{th})
\end{align*}
When $c_{th} = (1 - \zeta) nf_i \beta$ for some $\zeta \in [0,1]$, we can use the Chernoff bound as follows: 
\begin{align*}
\eta_i \leq \Pr(B(nf_i,\beta) < c_{th}) \leq e^{\frac{-\zeta^2 nf_i \beta}{2}},
\end{align*}
which means that $\eta_i$ decreases exponentially as $f_i$ increases. 

\subsection{Proof of Theorem~\ref{thm:FME_communication}}
In $\calS_{\FME}$, user $u_i$ sends $\la E_\pkd[h(x_i)], E_\pkd[E_\pks[E_\pkd[x_i]]] \ra$, whose size is $\tau_1 + \tau_3$, to the shuffler. 
Thus, $C_{U-S} = (\tau_1 + \tau_3) n$. 
Then, the shuffler samples $\la E_\pkd[h(x_i)], \allowbreak E_\pkd[E_\pks[E_\pkd[x_i]]] \ra$ with probability $\beta$ and adds $\colorB{\mu_1}$ dummies on average for each $i \in [b]$. 
Therefore, the size of the first shuffled data 
is $(\tau_1 + \tau_3)(\beta n + \colorB{\mu_1} b)$. 
Then, the data collector 
sends the set $\Lambda$ of selected items and the second shuffled data $(E_\pks[E_\pkd[y_{\pi(1)}]], \ldots, \allowbreak E_\pks[E_\pkd[y_{\pi(\tn)}]])$ to the shuffler. 
The expected size of these data is 
at most $\tau_1 \E[|\Lambda|] + \tau_2 (\beta n + \colorB{\mu_1} b)$, as the size of each item in $\Lambda$ is at most $\tau_1$. 
Finally, the shuffler adds $\colorB{\mu_2}$ dummies on average for each $i \in \Lambda$ and sends the third shuffled data, whose expected size is 
at most 
$\tau_1 (\beta n + \colorB{\mu_1} b + \colorB{\mu_2} \E[|\Lambda|])$, 
to the data collector. 
Therefore, the communication cost between the shuffler and the data collector is 
$C_{S-D} 
\leq 
(2 \tau_1 + \tau_2 + \tau_3)(\beta n + \colorB{\mu_1} b) + \tau_1 (\colorB{\mu_2} + 1) \E[|\Lambda|]$. 

The remaining issue is an upper bound on the expected number $\E[|\Lambda|]$ of selected items. 
Below, we consider average values in all cases. 
Assume that $\beta n \leq l \leq b$. 
After random sampling, $\beta n$ input values are selected, and there are at most $\beta n$ 
corresponding hash values. 
All of these hash values can be selected in the filtering step. 
In addition, other $(l - \beta n)$ hash values 
can be selected with probability $\alpha$ in the filtering step. 
Note that $\calS_{\FME}$ uses a universal hash function $h$. 
Therefore, for each hash value $j \in [b]$, there are on average $\frac{d}{b}$ items $i \in [d]$ that support it (i.e., $j=h(i)$). 
Thus, $\E[|\Lambda|]$ can be upper bounded as follows: 
$\E[|\Lambda|] \leq \frac{(\beta n + \alpha (l - \beta n))d}{b}$. 

Next, assume that $l \leq \beta n$. 
In this case, at most $l$ hash values are selected in the filtering step. 
Thus, we have $\E[|\Lambda|] \leq \frac{ld}{b}$, which proves Theorem~\ref{thm:FME_communication}. 
\qed

\subsection{Proof of Theorem~\ref{thm:FME_accuracy}}
\label{sub:proof_FME_accuracy}

\noindent{\textbf{Expectation and Variance of $\hf_i$.}}~~In $\colorB{\calS_{\calD^*,\beta}^{\FME}}$, $n f_i$ users have item $i$, and each of them is selected by random sampling with probability $\beta$. 
In addition, for any selected item $i \in \Lambda$, $z_i$ dummy values are added after filtering, and $\E[z_i] = \colorB{\mu_2}$. 
Thus, for $i \in \Lambda$, 
\begin{align*}
\E[\hf_i | \Lambda] 
= \textstyle{\frac{1}{n\beta}(\E[\tc_i | \Lambda] - \colorB{\mu_2})}
= \textstyle{\frac{1}{n\beta}(n f_i \beta + \colorB{\mu_2} - \colorB{\mu_2})}
= f_i, 
\end{align*}
which means that $\hf_i$ is unbiased. 

Next, we analyze 
the variance. 
For any selected item $i \in \Lambda$, $\V[\hf_i | \Lambda]$ can be written as 
\begin{align}
\V[\hf_i | \Lambda] = \textstyle{\frac{\V[\tc_i | \Lambda]}{n^2 \beta^2}}.
\label{eq:FME_l2_loss_Lambda}
\end{align}
By the law of total variance, 
$\V[\tc_i | \Lambda]$ 
can be written as 
\begin{align}
\V[\tc_i | \Lambda] 
&= \E[\V[\tc_i | z_i, \Lambda]] + \V[\E[\tc_i | z_i, \Lambda]] \nonumber\\
&= \E[n f_i \beta (1 - \beta)] + \V[nf_i \beta + z_i] \nonumber\\
&= n f_i \beta (1 - \beta) + \V[z_i] 
= n f_i \beta (1 - \beta) + \colorB{\sigma_2^2}.
\label{eq:FME_V_Lambda}
\end{align}
By (\ref{eq:FME_l2_loss_Lambda}) and (\ref{eq:FME_V_Lambda}), we have 
\begin{align*}
\V[\hf_i | \Lambda] 
= \textstyle{\frac{f_i (1 - \beta)}{n \beta} + \frac{\colorB{\sigma_2^2}}{n^2 \beta^2}}. 
\end{align*}

\smallskip{}
\noindent{\textbf{Expected Squared Error.}}~~Item $i$ is not selected in the filtering step (i.e., $i \notin \Lambda$) with probability $\eta_i$, and $\hf_i = 0$ in this case. 
Thus, the expected squared error is written as 
\begin{align*}
\textstyle{\E[(\hf_i - f_i)^2] = (1 - \eta_i)\V[\hf_i | \Lambda] + \eta_i f_i^2 }. 
\end{align*}
\qed

\section{Computational DP of the FME PROTOCOL}
\label{sec:computational_DP_FME}

We formally prove that our FME protocol $\colorB{\calS_{\calD^*, \beta}^{\FME}}$ provides $(\epsilon,\delta)$-CDP and is robust to collusion with users if plaintext messages for the data collector 
provide $(\epsilon,\delta)$-DP and are robust to collusion with users (Lemma~\ref{lem:FME_CDP} in Appendix~\ref{sec:proofs_proposed}). 
Specifically, we reduce CDP (and the robustness) of $\colorB{\calS_{\calD^*, \beta}^{\FME}}$ to 
the security 
of the PKE scheme. 

\colorB{Below, we assume that an attacker $\calA$ is either the shuffler or the data collector without loss of generality, as they receive messages during the protocol (unlike other parties, such as users except the victim and parties outside of the system).} 

\subsection{Basic Facts on Public-Key Encryption}
A PKE (Public-Key Encryption) 
scheme consists of three algorithms $\Sigma=(\mathsf{Gen},\mathsf{Enc},\mathsf{Dec})$.
Consider the following security game between a challenger $\mathcal{C}$ and an attacker $\mathcal{A}$:
\begin{itemize}
    \item $\mathcal{C}$ computes a pair of public and secret keys: $(\mathsf{pk},\mathsf{sk})\leftarrow\mathsf{Gen}(1^\lambda)$.
    \item $\mathcal{C}$ gives $\mathsf{pk}$ to $\mathcal{A}$.
    \item $\mathcal{A}$ gives a pair $(m_0,m_1)$ of challenge messages of the same length to $\mathcal{C}$.
    \item $\mathcal{C}$ chooses $b\in\{0,1\}$ uniformly at random, computes a ciphertext $c_b\leftarrow\mathsf{Enc}(\mathsf{pk},m_b)$, and gives $c_b$ to $\mathcal{A}$.
    \item $\mathcal{A}$ outputs a bit $b'$.    
\end{itemize}
We say that $\Sigma$ is IND-CPA (Indistinguishability under Chosen-Plaintext Attack) secure \cite{Katz_book} if there exists a negligible function $\mathsf{negl}(\lambda)$ such that
\begin{align*}
    \Pr[b'=b]\leq\textstyle{\frac{1}{2}+\mathsf{negl}(\lambda)}
\end{align*} 
in the above security game.

It is known that an IND-CPA secure PKE scheme generally satisfies multi-message security.
Specifically, 
consider the following modified security game between 
$\mathcal{C}$ and 
$\mathcal{A}$:
\begin{itemize}
    \item $\mathcal{C}$ computes a pair of public and secret keys: $(\mathsf{pk},\mathsf{sk})\leftarrow\mathsf{Gen}(1^\lambda)$.
    \item $\mathcal{C}$ gives $\mathsf{pk}$ to $\mathcal{A}$.
    \item $\mathcal{A}$ gives a pair $(\mathbf{m}_0,\mathbf{m}_1)$ of two tuples of challenge messages of the same length to $\mathcal{C}$, where $\mathbf{m}_b=(m_{b,1},\ldots,m_{b,p})$.
    \item $\mathcal{C}$ chooses $b\in\{0,1\}$ uniformly at random, computes a tuple of ciphertexts $\mathbf{c}_b=(c_{b,1},\ldots,c_{b,p})$ where $c_{b,i}\leftarrow\mathsf{Enc}(\mathsf{pk},m_{b,i})$, and gives $\mathbf{c}_b$ to $\mathcal{A}$.
    \item $\mathcal{A}$ outputs a bit $b'$.    
\end{itemize}
Then, there exists a negligible function $\mathsf{negl}(\lambda)$ such that
\begin{align*}
    \Pr[b'=b]\leq\textstyle{\frac{1}{2}+\mathsf{negl}(\lambda)}
\end{align*} 
in the above security game.
In the following, we simply denote a ciphertext of $x$ by $E_{\mathsf{pk}}[x]=\mathsf{Enc}(\mathsf{pk},x)$.

\subsection{DP for the Data Collector}

Let $(x_1,\ldots,x_n)$ be inputs of $n$ users.
Observe that the view of the data collector in the 
FME 
protocol consists of
\begin{align*}
    &(E_{\pkd}[y_{\pi(i)}^H])_{i\in[\tn]},~
    (E_{\pkd}[E_{\pks}[E_{\pkd}[y_{\pi(i)}]]])_{i\in[\tn]},~\\
    &(E_{\pkd}[y_{\rho(i)}^*])_{i\in[\tn^*]}.
\end{align*}
This is simulated from the following elements
\begin{align}\label{view:1}
    (y_{\pi(i)}^H)_{i\in[\tn]},~
    (E_{\pks}[E_{\pkd}[y_{\pi(i)}]])_{i\in[\tn]},~
    (y_{\rho(i)}^*)_{i\in[\tn^*]}.
\end{align}
We argue that the elements~(\ref{view:1}) are computationally indistinguishable from 
\begin{align}\label{view:2}
    (y_{\pi(i)}^H)_{i\in[\tn]},~
    (E_{\pks}[E_{\pkd}(\bot)])_{i\in[\tn]},~
    (y_{\rho(i)}^*)_{i\in[\tn^*]},
\end{align}
where all the $y_{\pi(i)}$'s are replaced with $\bot$.
Indeed, assume otherwise that there exists an efficient algorithm $\mathcal{A}$ that distinguishes (\ref{view:1}) and (\ref{view:2}) with non-negligible advantage.
Then, we can construct an attacker $\mathcal{A}'$ breaking the multi-message security of the underlying PKE scheme (with respect to the public key $\pks$) as follows:
Given the public key $\pks$, $\mathcal{A}'$ computes $(y_{\pi(i)}^H)_{i\in[\tn]},(y_{\pi(i)})_{i\in[\tn]},(y_{\rho(i)}^*)_{i\in[\tn^*]}$ by simulating the execution of the protocol involving users with inputs $(x_i)_{i\in[n]}$, the shuffler, and the data collector.
Then, $\mathcal{A}'$ chooses $\mathbf{m}_0=(E_{\pkd}[y_{\pi(i)}])_{i\in[\tn]}$ and $\mathbf{m}_1=(E_{\pkd}[\bot])_{i\in[\tn]}$ as two tuples of challenge messages.
Upon receiving ciphertexts $\mathbf{c}_b$ encrypting $\mathbf{m}_b$, $\mathcal{A}'$ gives $(y_{\pi(i)}^H)_{i\in[\tn]},\mathbf{c}_b,(y_{\rho(i)}^*)_{i\in[\tn^*]}$ to $\mathcal{A}$ and outputs the bit that it outputs.
If $b=0$, then $\mathcal{A}$ receives (\ref{view:1}). 
If $b=1$, then it receives (\ref{view:2}). 
Thus, $\mathcal{A}'$ breaks the multi-message security of the PKE scheme. 

By assumption, 
the mechanism outputting $(y_{\pi(i)}^H)_{i\in[\tn]}, \allowbreak (y_{\rho(i)}^*)_{i\in[\tn^*]}$ in (\ref{view:2}) provides $(\epsilon,\delta)$-DP and is robust to collusion with users.
Since the other elements in (\ref{view:2}) are independent of users' inputs, a mechanism outputting the elements in (\ref{view:2}) also provides $(\epsilon,\delta)$-DP and is robust to collusion with users.
Therefore, 
a mechanism outputting the elements in (\ref{view:1}) (i.e., the view of the data collector) 
provides $(\epsilon,\delta)$-CDP and is robust to collusion with users. 
\qed

\subsection{DP for the Shuffler}

Let $(x_1,\ldots,x_n)$ be inputs of $n$ users.
Observe that the view of the shuffler in the 
FME 
protocol is simulated from the following elements:
\begin{align}
    &(E_{\pkd}[h(x_i)])_{i\in[n]},~
    (E_{\pkd}[E_{\pks}[E_{\pkd}[x_i]]])_{i\in[n]},~\nonumber\\
    &
    \Lambda,~
    (E_{\pks}[E_{\pkd}[y_{\pi(i)}]])_{i\in[\tn]}.\label{view:3}
\end{align}
First, we consider the elements obtained by 
replacing $(E_{\pkd}[h(x_i)])_{i\in[n]}$ and $(E_{\pkd}[E_{\pks}[E_{\pkd}[x_i]]])_{i\in[n]}$ in (\ref{view:3}) with $(E_{\pkd}[h(a)])_{i\in[n]}$ and $(E_{\pkd}[E_{\pks}[E_{\pkd}[a]]])_{i\in[n]}$, respectively, where $a$ is a default value independent of users' inputs (e.g., $a=1$), i.e., 
\begin{align}
    &(E_{\pkd}[h(a)])_{i\in[n]},~
    (E_{\pkd}[E_{\pks}[E_{\pkd}[a]]])_{i\in[n]},~\nonumber\\
    &
    \Lambda,~
    (E_{\pks}[E_{\pkd}[y_{\pi(i)}]])_{i\in[\tn]},\label{view:4}
\end{align}
We argue that the elements in (\ref{view:3}) are computationally indistinguishable from those in (\ref{view:4}).
Indeed, assume that there exists a distinguisher $\mathcal{A}$ between (\ref{view:3}) and (\ref{view:4}).
We can construct an attacker $\mathcal{A}'$ against the multi-message security of the underlying PKE scheme (with respect to the public key $\pkd$) as follows:
Given the public key $\pkd$, $\mathcal{A}'$ computes $(h(x_i))_{i\in[n]},
\Lambda,(y_{\pi(i)})_{i\in[\tn]}$ by simulating the execution of the protocol involving users with inputs $(x_i)_{i\in[n]}$, the shuffler, and the data collector.
$\mathcal{A}'$ also computes $(E_{\pks}[E_{\pkd}[x_i]])_{i\in[n]}$ and $(E_{\pks}[E_{\pkd}[y_{\pi(i)}]])_{i\in[\tn]}$ following the same procedure as the protocol.
Then, $\mathcal{A}'$ chooses $\mathbf{m}_0=((h(x_i))_{i\in[n]},(E_{\pks}[E_{\pkd}[x_i]])_{i\in[n]})$ and $\mathbf{m}_1=((h(a))_{i\in[n]},(E_{\pks}[E_{\pkd}[a]])_{i\in[n]})$ as two tuples of challenge messages.
Upon receiving ciphertexts $\mathbf{c}_b$ encrypting $\mathbf{m}_b$, $\mathcal{A}'$ gives 
\begin{align*}
    &\mathbf{c}_b,~
    \Lambda,~
    (E_{\pks}[E_{\pkd}[y_{\pi(i)}]])_{i\in[\tn]}
\end{align*}
to $\mathcal{A}$ and outputs the bit that it outputs.
If $b=0$, then $\mathcal{A}$ receives (\ref{view:3}). 
If $b=1$, then it receives (\ref{view:4}). 
Thus, $\mathcal{A}'$ breaks the multi-message security of the PKE scheme with the same advantage as $\mathcal{A}$. 

Second, we consider the elements obtained by 
replacing $(E_{\pks}[E_{\pkd}[y_{\pi(i)}]])_{i\in[\tn]}$ in (\ref{view:4}) with $(E_{\pks}[E_{\pkd}[\bot]])_{i\in[\tn]}$, i.e., 
\begin{align}
    &(E_{\pkd}[h(a)])_{i\in[n]},~
    (E_{\pkd}[E_{\pks}[E_{\pkd}[a]]])_{i\in[n]},~\nonumber\\
    &
    \Lambda,~
    (E_{\pks}[E_{\pkd}[\bot]])_{i\in[\tn]},\label{view:5}
\end{align}
It follows from a similar argument that the elements in (\ref{view:4}) are computationally indistinguishable from those in (\ref{view:5}).
If there exists a distinguisher $\mathcal{A}$ between (\ref{view:4}) and (\ref{view:5}), then we can construct an attacker $\mathcal{A}'$ against the PKE scheme (with respect to the public key $\pkd$):
$\mathcal{A}'$ prepares 
\begin{align*}
    &(E_{\pkd}[h(a)])_{i\in[n]},~
    (E_{\pkd}[E_{\pks}[E_{\pkd}[a]]])_{i\in[n]},~
    \Lambda,
\end{align*}
and the $y_{\pi(i)}$'s, following the same process as (\ref{view:4}).
Then, $\mathcal{A}'$ chooses $\mathbf{m}_0=(y_{\pi(i)})_{i\in[\tn]}$ and $\mathbf{m}_1=(\bot)_{i\in[\tn]}$ as two tuples of challenge messages.
Upon receiving ciphertexts $\mathbf{c}_b=(c_{b,i})_{i\in[\tn]}$ encrypting $\mathbf{m}_b$, $\mathcal{A}'$ gives 
\begin{align*}
    &\mathbf{c}_b,~
    \Lambda,~
    (E_{\pks}[c_{b,i}])_{i\in[\tn]}
\end{align*}
to $\mathcal{A}$, where $E_{\pks}[c_{b,i}]$ is computed by $\mathcal{A}'$ itself.
$\mathcal{A}'$ finally outputs the bit that $\mathcal{A}$ outputs.
The reduction correctly works since if $b=0$, then $\mathcal{A}$ receives the elements in (\ref{view:4}) and if $b=1$, then it receives those in (\ref{view:5}).

By assumption, 
the view of the data collector 
provides 
$(\epsilon,\delta)$-CDP and is robust to collusion with users. 
Thus, 
a mechanism outputting (\ref{view:5}) also achieves 
$(\epsilon,\delta)$-CDP and is robust to collusion with users, as 
$\Lambda$ 
is included in the view of the data collector and the other elements are independent of users' inputs.
We therefore conclude that a mechanism outputting the elements in (\ref{view:3}) (i.e., the view of the shuffler) 
provides $(\epsilon,\delta)$-CDP and is robust to collusion with users. 
\qed

\section{Computational Cost of the FME Protocol}
\label{sec:FME_computational_cost}
Below, we analyze the computational cost of the FME protocol $\colorB{\calS_{\calD^*, \beta}^{\FME}}$. 
In our analysis, we use a hash function $h$ used in \cite{Luo_CCS22} and in our experiments. 
Specifically, we use $h$ defined by $h(x) = ((a_1 x + a_0) \bmod p) \bmod b$, where $p \in [d, 2d)$ is a prime, $a_1 \in [p-1]$, and $a_0 \in [p]$. 
Let $y = (a_1 x + a_0) \bmod p$. 

In $\colorB{\calS_{\calD^*, \beta}^{\FME}}$, we need to calculate the set $\Lambda$ of selected items corresponding to the set $\Lambda^H$ of selected hash values (Algorithm~\ref{alg:S_FME}, line 12). 
For each hash value $z \in \Lambda^H$, we find the corresponding items $x$ such that $z = h(x)$. 
This can be done as follows: 
(i) calculate all values of $y$ such that $z = y \bmod b$ (there are about $\frac{p}{b}$ such values), and 
(ii) for each value of $y$, calculate $x$ such that $x = a_1^{-1} y - a_0 \bmod p$. 
Assume that $a_1^{-1}$ is precomputed, e.g., by using the extended Euclidean algorithm. 
Then, the steps (i) and (ii) can be computed in $O(\frac{d}{b})$ time. 
Thus, $\Lambda$ can be calculated from $\Lambda^H$ in $O(\frac{ld}{b})$ time. 
The remaining process of $\colorB{\calS_{\calD^*, \beta}^{\FME}}$ can be calculated in $O(n + b + |\Lambda|)$ time. 
Thus, the computational cost of $\colorB{\calS_{\calD^*, \beta}^{\FME}}$ can be expressed as $O(n + b + |\Lambda| + \frac{ld}{b})$ in total. 

For example, 
when $l < \beta n$, the optimal value of $b$ can be expressed as $b = O(\sqrt{ld})$ (see Section~\ref{sub:FME_optimization}). 
In addition, $\E[|\Lambda|] \leq \frac{ld}{b}$ in this case. 
Thus, the computational cost of $\colorB{\calS_{\calD^*, \beta}^{\FME}}$ is $O(n + \sqrt{ld})$. 

\section{Proofs of Statements in Section~\ref{sec:key-value}}
\label{sec:proofs_key-value}
\subsection{Proof of Theorem~\ref{thm:KV_DP}}
Let $\colorB{\calS_{\calD^*, \beta}^{\FME+}}$ be our FME protocol with the TKV-FK technique. 
$\colorB{\calS_{\calD^*, \beta}^{\FME+}}$ only differs from $\colorB{\calS_{\calD^*, \beta}^{\FME}}$ (Algorithm~\ref{alg:S_FME}) in that it filters data at a key level. 

Let $\colorB{\calG_{\calD^*, \beta}^{\FME\plain 1+}}$ be an algorithm that takes 
$(\langle k_1, v_1 \rangle, \ldots, \allowbreak \langle k_n, v_n \rangle)$ 
as input, applies a hash function $h: [d] \rightarrow [b]$ to keys $k_1, \ldots, k_n \in [d]$, and outputs shuffled hash values. 
As with the CH protocol, 
$\colorB{\calG_{\calD^*, \beta}^{\FME\plain 1+}}$ applies $h$ to input keys $k_1,\ldots,k_n$, 
runs lines 1-8 of Algorithm~\ref{alg:S_LNF}, and outputs the shuffled data. 
Thus, 
$\colorB{\calG_{\calD^*, \beta}^{\FME\plain 1+}}$ provides 
\colorB{$(\epsilon_1, \delta_1)$-DP} 
and is robust to collusion with users in the same way as $\colorB{\calG_{\calD^*, \beta}^{\FME\plain 1}}$ in the proof of Lemma~\ref{lem:FME_plain_DP}. 
The remaining part of $\colorB{\calS_{\calD^*, \beta}^{\FME+}}$ is the same as that of $\colorB{\calS_{\calD^*, \beta}^{\FME}}$. 
Therefore, $\colorB{\calS_{\calD^*, \beta}^{\FME+}}$ provides $(\epsilon, \delta)$-CDP and is robust to collusion with users in the same way as $\colorB{\calS_{\calD^*, \beta}^{\FME}}$. 

Let $\calS^{\PS}$ be the padding-and-sampling. 
Then, $\colorB{\calS_{\calD^*, \beta}^{\KV}}$ can be regarded as $\colorB{\calS_{\calD^*, \beta}^{\FME+}}$ preceded with $\calS^{\PS}$, i.e., $\colorB{\calS_{\calD^*, \beta}^{\KV}} = \colorB{\calS_{\calD^*, \beta}^{\FME+}} \circ \calS^{\PS}$. 
Let $D = (x_1, \ldots, x_n)$ and $D' = (x'_1, \ldots, x'_n)$ two neighboring databases. 
Let $D_S = (s_1, \ldots, s_n)$ and $D'_S = (s'_1, \ldots, s'_n)$ be the corresponding databases after applying $\calS^{\PS}$. 
The padding-and-sampling $\calS^{\PS}$ selects one KV pair $s_i$ from a set of KV pairs $x_i$ for each user $u_i$. 
Thus, $D_S$ and $D'_S$ remain neighboring after applying $\calS^{\PS}$; i.e., $D_S$ and $D'_S$ also differ on (at most) one entry. 

Since $\colorB{\calS_{\calD^*, \beta}^{\KV}} = \colorB{\calS_{\calD^*, \beta}^{\FME+}} \circ \calS^{\PS}$ and $D_S$ and $D'_S$ remain neighboring after applying $\calS^{\PS}$, $\colorB{\calS_{\calD^*, \beta}^{\KV}}$ also provides $(\epsilon, \delta)$-CDP and is robust to collusion with users. 
\qed

\subsection{Proof of Theorem~\ref{thm:KV_poisoning}}
\label{sub:proof_thm_KV_poisoning}
For $i\in[d]$, let $\tc_{i,1}'$ (resp.~$\tc_{i,-1}'$) $\in \nnints$ be the number of KV pairs $\la i, 1 \ra$ (resp.~$\la i, -1 \ra$) the data collector receives after data poisoning. 
Since $i \in \Lambda$ with probability $1 - \eta_i$ and 
$\E[\hPhi_i | \Lambda] = \Phi_i$ for any $i \in \Lambda$ (see Theorem~\ref{thm:KV_accuracy}), we have $\E[\hPhi_i] = (1 - \eta_i) \cdot \Phi_i + \eta_i \cdot 0$. 
Thus, the frequency gain $G_\Phi(\bmm)$ can be written as 
\begin{align}
G_\Phi(\bmm) 
&= \textstyle{\sum_{i \in \calT} (\E[\hPhi'_i] - \E[\hPhi_i])} \nonumber\\
&= \textstyle{\sum_{i \in \calT} ( \frac{\kappa}{(n+n')\beta} (\E[\tc_{i,1}'] + \E[\tc_{i,-1}'] - 2\colorB{\mu_2})} \nonumber\\
& \hspace{16mm} 
\textstyle{- (1 - \eta_i) \Phi_i)}. 
\label{eq:KV_G_Phi_bmm1}
\end{align}
Similarly, the mean gain $G_\Psi(\bmm)$ can be written as 
\begin{align}
G_\Psi(\bmm) 
&= \textstyle{\sum_{i \in \calT} (\E[\hPsi'_i] - \E[\hPsi_i])} \nonumber\\
&\approx \textstyle{\sum_{i \in \calT} \left( \frac{\kappa}{(n+n')\beta} \cdot \frac{\E[\tc_{i,1}'] - \E[\tc_{i,-1}']}{\E[\hPhi'_i]} - \E[\hPsi_i] \right)} \nonumber\\
&= \textstyle{\sum_{i \in \calT} \left( \frac{\E[\tc_{i,1}'] - \E[\tc_{i,-1}']}{\E[\tc_{i,1}'] + \E[\tc_{i,-1}'] - 2\colorB{\mu_2}} - \E[\hPsi_i] \right)}. \nonumber\\
&\lesssim \textstyle{\sum_{i \in \calT} \left( \frac{\E[\tc_{i,1}'] - \E[\tc_{i,-1}']}{\E[\tc_{i,1}'] + \E[\tc_{i,-1}'] - 2\colorB{\mu_2}} - \Psi_i \right)},
\label{eq:KV_G_Psi_bmm1}
\end{align}
where the last inequality is obtained from Taylor expansion $\E[\frac{X}{Y}] \approx \frac{\E[X]}{\E[Y]}$ for two random variables $X$ and $Y$. 
The last inequality holds because $\E[\hPsi_i | \Lambda] \approx \Psi_i$ for any $i \in \Lambda$ (see Theorem~\ref{thm:KV_accuracy}) and 
$\hPsi_i = 1$ for any $i \notin \Lambda$. 

Both (\ref{eq:KV_G_Phi_bmm1}) and (\ref{eq:KV_G_Psi_bmm1}) are maximized 
when (i) each of $n'$ fake users sends $\la i, 1 \ra$ (and its hash value) for a randomly selected target item $i \in \calT$ and (ii) all target users are selected in the filtering step (i.e., $\calT \subseteq \Lambda$). 
In this case, (on average) $\frac{n'}{|\calT|}$ fake users select item $\la i, 1 \ra$, and each of them is selected with probability $\beta$ after random sampling. 
Thus, for $i \in \calT$, we have 
\begin{align}
\E[\tc_{i,1}'] = \textstyle{\E[\tc_{i,1}] + \frac{\beta n'}{|\calT|}}, ~~ 
\E[\tc_{i,-1}'] = \textstyle{\E[\tc_{i,-1}]}. 
\label{eq:KV_E_tc_i_pm1}
\end{align}

Below, we calculate $\E[\tc_{i,1}]$ and $\E[\tc_{i,-1}]$. 
For each user $u_j \in \calU_i$, key $i$ is selected with probability $\frac{\beta}{\xi_j}$ after padding-and-sampling and random sampling, and the corresponding value becomes $1$ (resp.~$-1$) with probability $\frac{1 + \psi_{j,i}}{2}$ (resp.~$\frac{1 - \psi_{j,i}}{2}$). 
In addition, (on average) $\colorB{\mu_2}$ dummy values are added to each of $\la i, 1 \ra$ and $\la i, -1 \ra$. 
Thus, 
$\E[\tc_{i,1}]$ and $\E[\tc_{i,-1}]$ in (\ref{eq:KV_E_tc_i_pm1}) are: 
\begin{align}
\E[\tc_{i,1}] &= \textstyle{(\sum_{u_j \in \calU_i} \frac{\beta}{\xi_j} \cdot \frac{1 + \psi_{j,i}}{2}) + \colorB{\mu_2}} \label{eq:KV_E_tc_i_p1_general}\\
\E[\tc_{i,-1}] &= \textstyle{(\sum_{u_j \in \calU_i} \frac{\beta}{\xi_j} \cdot \frac{1 - \psi_{j,i}}{2}) + \colorB{\mu_2}}. \label{eq:KV_E_tc_i_m1_general}
\end{align}
By (\ref{eq:KV_G_Phi_bmm1}), (\ref{eq:KV_G_Psi_bmm1}), (\ref{eq:KV_E_tc_i_pm1}), (\ref{eq:KV_E_tc_i_p1_general}), and (\ref{eq:KV_E_tc_i_m1_general}), 
the maximum gains are: 
\begin{align}
\GMTGAPhi &
= \textstyle{\sum_{i \in \calT} \left( \frac{\kappa}{(n+n')\beta} ((\sum_{u_j \in \calU_i} \frac{\beta}{\xi_j}) + \frac{\beta n'}{|\calT|}) - (1 - \eta_i) \Phi_i \right)} \nonumber\\
&= \textstyle{\frac{\kappa}{n+n'} \left( (\sum_{i \in \calT} \sum_{u_j \in \calU_i} \frac{1}{\xi_j}) + n' \right) \hspace{-0.5mm} - \hspace{-0.5mm} \Phi_\calT \hspace{-0.5mm} + \hspace{-0.5mm} \sum_{i \in \calT} \eta_i \Phi_i} \nonumber
\end{align}
\begin{align}
\GMTGAPsi 
&\approx 
\sum_{i \in \calT} \left( \frac{(\sum_{u_j \in \calU_i} \frac{\beta \psi_{j,i}}{\xi_j}) + \frac{\beta n'}{|\calT|}}{(\sum_{u_j \in \calU_i} \frac{\beta}{\xi_j}) + \frac{\beta n'}{|\calT|}} - \Psi_i \right) \nonumber\\
&= 
\left( \sum_{i \in \calT} \frac{(\sum_{u_j \in \calU_i} \frac{\psi_{j,i}}{\xi_j}) + \frac{n'}{|\calT|}}{(\sum_{u_j \in \calU_i} \frac{1}{\xi_j}) + \frac{n'}{|\calT|}} \right) - \Psi_\calT, \nonumber
\end{align}
which proves (\ref{eq:KV_GMTGAPhi_general}) and (\ref{eq:KV_GMTGAPsi_general}). 
\qed

\smallskip{}
\noindent{\textbf{Comparison with the Existing KV Protocols.}}~~The prior work \cite{Wu_USENIX22} analyzes the robustness of the existing KV protocols when the number $|x_i|$ of KV pairs does not exceed the padding length $\kappa$ for any $i \in [n]$. 
In this case, $\xi_j = \max\{|x_j|, \kappa\} = \kappa$, and 
(\ref{eq:KV_GMTGAPhi_general}) and (\ref{eq:KV_GMTGAPsi_general}) can be simplified as follows: 
\begin{align}
\GMTGAPhi &
\hspace{-0.5mm}=\hspace{-0.5mm} \textstyle{\frac{\kappa}{n+n'} \left( (\sum_{i \in \calT} \sum_{u_j \in \calU_i} \frac{1}{\kappa}) + n' \right) \hspace{-0.5mm} - \hspace{-0.5mm} \Phi_\calT \hspace{-0.5mm} + \hspace{-0.5mm} \sum_{i \in \calT} \eta_i \Phi_i} \nonumber\\
&=\hspace{-0.5mm} \textstyle{\lambda (\kappa - \Phi_\calT) + \sum_{i \in \calT} \eta_i \Phi_i} \nonumber\\
& \text{(as $\textstyle{\lambda = \frac{n'}{n+n'}}$, $\textstyle{\Phi_\calT = \sum_{i \in \calT} \Phi_i}$, and $|\calU_i| = n \Phi_i$)}
\label{eq:GMTGAPhi_simple}
\end{align}
\begin{align}
\GMTGAPsi 
&\approx \left( \sum_{i \in \calT} \frac{(\sum_{u_j \in \calU_i} \frac{\psi_{j,i}}{\kappa}) + \frac{n'}{|\calT|}}{(\sum_{u_j \in \calU_i} \frac{1}{\kappa}) + \frac{n'}{|\calT|}} \right) - \Psi_\calT \nonumber\\
&= 
\textstyle{\left( \sum_{i \in \calT} \frac{n \Phi_i \Psi_i |\calT| + n' \kappa}{n \Phi_i |\calT| + n' \kappa} \right) - \Psi_\calT} \nonumber\\
& \hspace{5mm} \text{(as $\textstyle{\Psi_i = \frac{1}{|\calU_i|} \sum_{u_j \in \calU_i} \psi_{j,i}}$ and $|\calU_i| = n \Phi_i$)} \nonumber\\
&= \textstyle{\left(\sum_{i \in \calT} \frac{(1-\lambda)\Phi_i \Psi_i |\calT| + \lambda \kappa}{(1-\lambda)\Phi_i |\calT| + \lambda \kappa} \right) - \Psi_\calT}. 
\label{eq:GMTGAPsi_simple}
\end{align}
The overall gains $\GMTGAPhi$ and $\GMTGAPsi$ in (\ref{eq:GMTGAPhi_simple}) and (\ref{eq:GMTGAPsi_simple}) are the same as the overall gains of RKVA (Random Key-Value Pair Attack) against PCKV-GRR and PCKV-UE \cite{Wu_USENIX22}. 
In RKVA, each fake user randomly selects a target key $i \in \calT$, applies the LDP mechanism to $\la i, 1 \ra$, and sends the result. 
Thus, RKVA can be regarded as an \textit{input poisoning attack}
\cite{Li_USENIX23}, which injects fake input to the LDP mechanism, and is much weaker than M2GA, the optimal \textit{output poisoning attack}. 
In other words, our protocol $\colorB{\calS_{\calD^*, \beta}^{\KV}}$ is much more robust than the existing protocols, such as PCKV-GRR and PCKV-UE. 
This is consistent with our experimental results.

\subsection{Proof of Theorem~\ref{thm:KV_accuracy}}

\noindent{\textbf{Expectation and Variance of $\hPhi_i$.}}~~For $i \in \Lambda$, there are $n \Phi_i$ users who have key $i$. 
For each of them, key $i$ is selected with probability $\frac{\beta}{\kappa}$ after padding-and-sampling and random sampling.  
In addition, (on average) $\colorB{\mu_2}$ dummy values are added to both $\la i, 1 \ra$ and $\la i, -1 \ra$. 
Thus, 
\begin{align*}
\E[\tc_{i,1} + \tc_{i,-1} | \Lambda] = \textstyle{\frac{n \beta \Phi_i}{\kappa} + 2 \colorB{\mu_2}}, 
\end{align*}
and 
\begin{align*}
\E[\hPhi_i | \Lambda] = \textstyle{\frac{\kappa}{n\beta}(\E[\tc_{i,1} + \tc_{i,-1} | \Lambda] - 2\colorB{\mu_2})} = \Phi_i, 
\end{align*}
which proves (\ref{eq:KV_Phi_bias}). 
In addition, we have 
\begin{align}
\V[\hPhi_i | \Lambda] 
&= \textstyle{\frac{\kappa}{n^2 \beta^2}(\V[\tc_{i,1} + \tc_{i,-1} | \Lambda])} \nonumber\\
&= \textstyle{\frac{\kappa^2}{n^2 \beta^2}(n \Phi_i \frac{\beta}{\kappa} (1 - \frac{\beta}{\kappa}) + 2 \colorB{\sigma_2^2})} \nonumber\\
&= \textstyle{\frac{\Phi_i(\kappa - \beta)}{n \beta} + \frac{2 \kappa^2 \colorB{\sigma_2^2}}{n^2 \beta^2}},
\label{eq:V_tc_i_1_m1_Lambda}
\end{align}
as each of $n \Phi_i$ users increases $\tc_{i,1} + \tc_{i,-1}$ by $1$ with probability $\frac{\beta}{\kappa}$ and dummy values of variance $\colorB{\sigma_2^2}$ are added to both $\la i, 1 \ra$ and $\la i, -1 \ra$. 
Therefore, (\ref{eq:KV_Phi_variance}) also holds.

\smallskip{}
\noindent{\textbf{Expectation and Variance of $\hPsi_i$.}}~By Taylor expansion, 
\begin{align}
\E[\hPsi_i | \Lambda] 
\approx \textstyle{\frac{\kappa}{n\beta} \frac{\E[\tc_{i,1} - \tc_{i,-1} | \Lambda]}{\E[\hPhi_i | \Lambda]}} 
= \textstyle{\frac{\kappa}{n\beta} \frac{\E[\tc_{i,1} - \tc_{i,-1} | \Lambda]}{\Phi_i}}
\label{eq:E_hps_i_Lambda}
\end{align}
for $i \in \Lambda$. 
Let $\calU_i$ be the set of users who have key $i$ ($|\calU_i| = n \Phi_i$), and  
let $\psi_{j,i} \in [-1,1]$ be the value of key $i$ held by user $u_j \in \calU_i$. 
For each user $u_j \in \calU_i$, key $i$ is selected with probability $\frac{\beta}{\kappa}$ after padding-and-sampling and random sampling, and the corresponding value becomes $1$ (resp.~$-1$) with probability $\frac{1 + \psi_{j,i}}{2}$ (resp.~$\frac{1 - \psi_{j,i}}{2}$). 
Moreover, the average of $\psi_{j,i}$ over $\calU_i$ is $\Psi_i$, and (on average) $\colorB{\mu_2}$ dummy values are added to both $\la i, 1 \ra$ and $\la i, -1 \ra$. 
Thus, we have 
\begin{align}
\E[\tc_{i,1} - \tc_{i,-1} | \Lambda] 
= \textstyle{\frac{\beta n \Phi_i}{\kappa} \cdot \Psi_i + \colorB{\mu_2} - \colorB{\mu_2}}
= \textstyle{\frac{\beta n \Phi_i}{\kappa} \cdot \Psi_i}.
\label{eq:E_tc_i1_tc_im1_Lambda}
\end{align}
By (\ref{eq:E_hps_i_Lambda}) and (\ref{eq:E_tc_i1_tc_im1_Lambda}), we have
\begin{align*}
\E[\hPsi_i | \Lambda] \approx \Psi_i, 
\end{align*}
which proves (\ref{eq:KV_Psi_bias}). 
In addition, by Taylor expansion, 
\begin{align}
\V[\hPsi_i | \Lambda] \approx \textstyle{\frac{\kappa^2}{n^2 \beta^2} \frac{\V[\tc_{i,1} - \tc_{i,-1} | \Lambda]}{\E[\hPhi_i | \Lambda]}} 
= \textstyle{\frac{\kappa^2}{n^2 \beta^2} \frac{\V[\tc_{i,1} - \tc_{i,-1} | \Lambda]}{\Phi_i}}. 
\label{eq:V_hps_i_Lambda}
\end{align}
Since 
$\V[\tc_{i,1} \pm \tc_{i,-1} | \Lambda] = \V[\tc_{i,1} | \Lambda] + \V[\tc_{i,-1} | \Lambda] \pm \text{Cov}(\tc_{i,1}, \allowbreak \tc_{i,-1} | \Lambda)$, we have 
\begin{align}
&\V[\tc_{i,1} - \tc_{i,-1} | \Lambda] \nonumber\\
&= 2( \V[\tc_{i,1} | \Lambda] + \V[\tc_{i,-1} | \Lambda]) - \V[\tc_{i,1} + \tc_{i,-1} | \Lambda].
\label{eq:V_tc_i_1_minus_V_tc_i_m1}
\end{align}
As with (\ref{eq:V_tc_i_1_m1_Lambda}), $\V[\tc_{i,1} + \tc_{i,-1} | \Lambda]$ can be written as 
\begin{align}
\V[\tc_{i,1} + \tc_{i,-1} | \Lambda] = \textstyle{n \Phi_i \frac{\beta}{\kappa} (1 - \frac{\beta}{\kappa}) + 2 \colorB{\sigma_2^2}}.
\label{eq:V_tc_i_1_m1_Lambda2}
\end{align}
Below, we calculate $\V[\tc_{i,1} | \Lambda]$. 
For $j \in \calU_i$, let $q_{j,i} \in [0,1]$ be the probability that user $u_j$ adds $\tc_{i,1}$ by one. 
Then, we have 
\begin{align}
q_{j,i} &= \textstyle{\frac{\beta}{\kappa} \cdot \frac{1 + \psi_{j,i}}{2}} \label{eq:V_tc_i_q_ji}\\
\Psi_i &= \textstyle{\frac{1}{|\calU_i|} \sum_{u_j \in \calU_i} \psi_{j,i} = \frac{1}{n\Phi_i} \sum_{u_j \in \calU_i} \psi_{j,i}}. \label{eq:V_tc_i_Psi_i}
\end{align}
Since dummy values of variance $\colorB{\sigma_2^2}$ are added to $\la i, 1 \ra$, 
\begin{align}
\V[\tc_{i,1} | \Lambda] 
&= \textstyle{\sum_{u_j \in \calU_i} q_{j,i}(1 - q_{j,i}) + \colorB{\sigma_2^2}} \nonumber\\
&= \textstyle{\sum_{u_j \in \calU_i} q_{j,i} - \sum_{u_j \in \calU_i} q_{j,i}^2 + \colorB{\sigma_2^2}} \nonumber\\
&\leq \textstyle{\sum_{u_j \in \calU_i} q_{j,i} - \frac{1}{|\calU_i|} (\sum_{u_j \in \calU_i} q_{j,i})^2 + \colorB{\sigma_2^2}} \nonumber\\
&\hspace{4.8mm} \text{(by the Cauchy–Schwarz inequality)} \nonumber\\
&= \textstyle{\frac{\beta}{\kappa} \cdot \frac{n\Phi_i (1+\Psi_i)}{2} - \frac{1}{n \Phi_i} \cdot (\frac{\beta}{\kappa} \cdot \frac{n\Phi_i (1+\Psi_i)}{2})^2 + \colorB{\sigma_2^2}} \nonumber\\
&\hspace{4.8mm} \text{(by (\ref{eq:V_tc_i_q_ji}) and (\ref{eq:V_tc_i_Psi_i})} \nonumber\\
&= n \Phi_i (q_i - q_i^2) + \colorB{\sigma_2^2}, 
\label{V_tc_i_nPhi_i_q_i}
\end{align}
where $q_i = \frac{\beta(1 + \Psi_i)}{2 \kappa}$. 
Similarly, $\V[\tc_{i,-1} | \Lambda]$ can be written as 
\begin{align}
\V[\tc_{i,-1} | \Lambda] \leq n \Phi_i (r_i - r_i^2) + \colorB{\sigma_2^2},
\label{V_tc_i_nPhi_i_r_i}
\end{align}
where $r_i = \frac{\beta(1 - \Psi_i)}{2 \kappa}$. 
By (\ref{eq:V_tc_i_1_minus_V_tc_i_m1}), (\ref{eq:V_tc_i_1_m1_Lambda2}), (\ref{V_tc_i_nPhi_i_q_i}), and (\ref{V_tc_i_nPhi_i_r_i}), 
\begin{align*}
\V[\tc_{i,1} - \tc_{i,-1} | \Lambda] 
\leq \textstyle{2 n \Phi_i (q_i - q_i^2 + r_i - r_i^2) - n \Phi_i \frac{\beta}{\kappa} (1 - \frac{\beta}{\kappa})}. 
\end{align*}
By (\ref{eq:V_hps_i_Lambda}), 
\begin{align*}
\V[\hPsi_i | \Lambda] \lesssim \textstyle{\frac{\kappa^2}{n \beta^2} ( 2(q_{i} - q_{i}^2 + r_{i} - r_{i}^2)- \frac{\beta}{\kappa}(1-\frac{\beta}{\kappa}))}, 
\end{align*}
which proves (\ref{eq:KV_Psi_variance}). 

\smallskip{}
\noindent{\textbf{Expected Squared Error.}}~~Item $i$ is not selected in the filtering step (i.e., $i \notin \Lambda$) with probability $\eta_i$ and that $\hPhi_i = 0$ and $\hPsi_i = 1$ in this case. 
Thus, we have 
\begin{align*}
\textstyle{\E[(\hPhi_i - \Phi_i)^2]} 
&= \textstyle{(1 - \eta_i)\V[\hPhi_i | \Lambda] + \eta_i \Phi_i^2 }, \\
\textstyle{\E[(\hPsi_i - \Psi_i)^2]} 
&= \textstyle{(1 - \eta_i)\V[\hPsi_i | \Lambda] + \eta_i (1 - \Psi_i)^2 }. 
\end{align*}
\qed
}

\end{document}